\documentclass[letterpaper,11pt]{article}
\usepackage[cm]{fullpage}

\RequirePackage[OT1]{fontenc}
\RequirePackage{amsthm,amsmath}
\RequirePackage[authoryear,round]{natbib}
\RequirePackage[colorlinks,citecolor=blue,urlcolor=blue]{hyperref}
\usepackage{amssymb}
\usepackage{graphicx}
\usepackage{url}
\usepackage{epsfig}
\usepackage{algorithm}
\usepackage{algorithmic}
\usepackage{color}
\usepackage{subcaption}
\usepackage{verbatim}
\usepackage{marginnote}
\usepackage{xifthen}

\newtheorem{definition}{Definition} 
\newtheorem{theorem}{Theorem}
\newtheorem{lemma}{Lemma}
\newtheorem{proposition}{Proposition}
\newtheorem{fact}{Fact}
\newtheorem{corollary}{Corollary}

\numberwithin{theorem}{section}
\numberwithin{lemma}{section}
\numberwithin{proposition}{section}
\numberwithin{fact}{section}
\usepackage{xparse}
\usepackage{sectsty}
\allsectionsfont{\sffamily\mdseries\upshape}

\usepackage[nottoc,notlof,notlot]{tocbibind} 
\usepackage[titles,subfigure]{tocloft}

\newcommand{\mline}{\;\middle\vert\;}

\newcommand{\pr}[1]{\mathbb{P}\left( #1 \right)}

\DeclareDocumentCommand{\prm}{ O{default1} O{default2} m }{\underset{\substack{#1 \\ #2 }}{\mathbb{P}}\left( #3 \right)}
\DeclareDocumentCommand{\Em}{ O{default1} O{default2} m }{\underset{\substack{#1 \\ #2 }}{\mathbb{E}}\left[ #3 \right]}
\newcommand{\E}[1]{\mathbb{E}\left[ #1 \right]}

\newcommand{\eps}{\epsilon}
\renewcommand{\O}[1]{\mathcal{O}\left(#1\right)}
\newcommand{\Otilde}[1]{\tilde{\mathcal{O}}\left(#1\right)}

\newcommand{\Om}[1]{\Omega\left(#1\right)}

\graphicspath {{Code/}}

\newcommand{\R}{\mathbb{R}}
\newcommand{\Exp}{\mathbb{E}}
\newcommand{\Normal}{\mathcal{N}}

\newcommand{\ux}{\underline{x}}
\newcommand{\Lap}{\mathrm{Lap}}
\newcommand{\GS}{\mathrm{GS}}
\newcommand{\Mrange}{M_{\mathit{range}}}
\newcommand{\poly}{\mathrm{poly}}

\title{Finite Sample Differentially Private Confidence Intervals\footnote{Presented at TPDP 2017 and to appear at ITCS 2018.}}
\date{}		
\author{Vishesh Karwa\thanks{Department of Statistics, The Ohio State University, \texttt{karwa.8@osu.edu}. Work done while a Post doctoral fellow in the Center for Research on Computation \& Society, The Department of Statistics and the School of Engineering and Applied Sciences at Harvard University, supported by the NSF grant CNS-1237235.
}
\and  Salil Vadhan\thanks{Center for Research on Computation \& Society and School of Engineering \& Applied Sciences, Harvard University, \texttt{salil\textunderscore vadhan@harvard.edu},  \protect\url{http://seas.harvard.edu/\~salil}. Supported by NSF grant CNS-1237235, a Simons Investigator Award, a grant from the Sloan Foundation, and Cooperative Agreement CB16ADR0160001 with the Census Bureau.}}
            
\begin{document}

\maketitle	
\begin{abstract}
	We study the problem of estimating finite sample confidence intervals of the mean of a normal population under the constraint of differential privacy. We consider both the known and unknown variance cases and construct differentially private algorithms to estimate confidence intervals. Crucially, our algorithms guarantee a finite sample coverage, as opposed to an asymptotic coverage.  Unlike most previous differentially private algorithms, we do not require the domain of the samples to be bounded. We also prove lower bounds on the expected size of any differentially private confidence set showing that our the parameters are optimal up to polylogarithmic factors. 
			
			
\end{abstract}

\section{Introduction}  \label{sec:intro}
\subsection{Overview} \label{sec:intro-overview}

Differential privacy~\citep{DMNS06} is a strong and by now widely accepted definition of privacy for statistical analysis of datasets with sensitive information about individuals.   While there is now a rich and flourishing body of research on differential privacy, extending well beyond theoretical computer science, the following three basic goals for research in the area have not been studied in combination with each other:

\paragraph{Differentially private statistical inference:}  The vast majority of work in differential privacy has
studied how well one can approximate statistical properties of the dataset itself, i.e. empirical quantities, rather than
inferring statistics of an underlying {\em population} from which a dataset is drawn.  Since the latter is the ultimate goal of most data analysis, it should also be a more prominent object of study in the differential privacy literature.

\paragraph{Conservative statistical inference:}  An important purpose of statistical inference is to limit the chance that data analysts draw incorrect conclusions because their dataset may not accurately reflect the underlying population, for example due to the sample size being too small.   For this reason, classical statistical inference also offers measures of statistical significance such as $p$-values and confidence intervals.   Constructing such measures for differentially private algorithms is more complex, as one must also take into account the additional noise that is introduced for the purpose of privacy protection.  For this reason, we advocate that differentially private inference procedures should be {\em conservative}, and err on the side of underestimating statistical significance, even at 
small sample sizes and for all settings of other parameters.

\paragraph{Rigorous analysis of the inherent price of privacy:}  As has been done extensively in the differential privacy
literature for empirical statistics, we should also investigate the fundamental ``privacy--utility tradeoffs'' for (conservative) differentially private statistical inference.   This involves both designing and analyzing differentially private statistical inference procedures, as well as proving negative results about the performance
that can be achieved, using the best non-private procedures as a benchmark.\\
\ \\
\ \\
In this paper, we pursue all of these goals, using as a case study the problem of constructing a confidence interval for the mean of normal data.  The latter is one of the most basic problems in statistical inference, yet already turns out to be nontrivial to fully understand under the constraints of differential privacy.  We expect that most of our modeling and methods will find analogues for other inferential problems (e.g. hypothesis testing, Bayesian credible intervals, non-normal data, and estimating statistics other than the mean).

\subsection{Confidence Intervals for a Normal Mean} \label{sec:intro-normal}

We begin by recalling the problem of constructing a $(1-\alpha)$-level 
confidence interval for a normal mean without privacy.
Let $X_1, \ldots, X_n$ be an independent and identically distributed (\emph{iid}) random sample from a normal distribution with an unknown mean $\mu$ and variance $\sigma^2 $.  The goal is to design an estimator
$I$ that given $X_1,\ldots,X_n$, outputs an interval $I(X_1,\ldots,X_n)\subseteq \R$ such that  
$$\mathbb P \left(I(X_1,\ldots,X_n)\ni \mu \right) \geq 1-\alpha,$$
for all $\mu$ and $\sigma$.  Here $1-\alpha$ is called the {\em coverage probability}.  Given a desired coverage probability, the goal is minimize the {\em expected length} of the interval, namely $\Exp[|I(X_1, \ldots, X_n)|]$.

\paragraph{Known Variance.}
In the case that variance $\sigma^2$ is known (so only $\mu$ is unknown), the classic confidence interval
for a normal mean is:
$$I(X_1,\ldots,X_n) = \bar{X} \pm \frac{\sigma}{\sqrt{n}}\cdot z_{1-\alpha/2},$$ 
where $\bar{X}$ is the sample mean and $z_{a}$ represents the $a^{th}$ quantile of a standard normal distribution.\footnote{The proof that this is in fact a $(1-\alpha)$-confidence interval follows by observing that 
$\sqrt{n}\cdot (\bar{X}-\mu)$ has a standard normal distribution, and $[-z_{1-\alpha/2},z_{1-\alpha/2}]$ covers
a $1-\alpha$ fraction of the mass of this distribution.}
It is known that this interval has the smallest expected size among all $1-\alpha$ level confidence sets for a normal mean, see for example, \cite{lehmann2006testing}.  In this case, the length of the confidence interval is fixed and equal to $$|I(X_1,\ldots,X_n) | = (2\sigma z_{1 - \alpha/2})/\sqrt{n} = \Theta \left(\sigma \sqrt{\log(1/\alpha)/n}\right).$$ 

\paragraph{Unknown Variance.}
In the case that the variance $\sigma^2$ is unknown, the variance must be estimated from the data itself, and the classic confidence interval is:
$$I(X_1,\ldots,X_n) = \bar{X} \pm \frac{s}{\sqrt{n}}\cdot t_{n-1,1-\alpha/2},$$ 
where $s^2$ is the {\em sample variance} defined by
$$s^2 = \frac{1}{n-1} \sum_{i=1}^n (X_i - \bar{X})^2,$$
and $t_{n-1,a}$ is the $a^{th}$ quantile of a $t$-distribution with $n-1$ degrees of freedom (see the appendix for definitions).\footnote{Again the proof follows by observing that $s\cdot (\bar{X}-\mu)$ follows a $t$ distribution, with no
dependence on the unknown parameters.}  Now the length of the interval is a random variable with expectation
$$\Exp[|I(X_1,\ldots,X_n)|] = \frac{2\sigma}{\sqrt{n}}\cdot k_n\cdot t_{n-1,1 - \alpha/2}= \Theta \left(\sigma \sqrt{\log(1/\alpha)/n}\right),$$
for an appropriate constant $k_n=1-O(1/n)$. (See \cite{lehmann2006testing}.) 

\paragraph{Relation to Hypothesis Tests.}
In general, including both cases above, a confidence interval for a population parameter also gives rise to
hypothesis tests, which is often how the confidence intervals are used in applied statistics.  For example, if our null hypothesis is that the mean $\mu$ is nonnegative, then we could reject the null hypothesis if the interval $I(X_1,\ldots,X_n)$ does not intersect the positive real line.  The significance level of this hypothesis test is
thus at least $1-\alpha$.  Minimizing the length of the confidence interval corresponds to being able to reject the
alternate hypotheses that are closer to the null hypothesis; that is, when the confidence interval is of length at most $\beta$ and $\mu$ is distance greater than $\beta$ from the null hypothesis, then the test will reject with probability 
at least $1-\alpha$.

\subsection{Differential Privacy} \label{sec:intro-dp}

Let $\underline{x} = (x_1, \ldots, x_n)$ be a dataset of $n$ elements where each $x_i \in \Omega$. In the problems that we consider, $\Omega = \mathbb{R}$ and $\underline{x} \in \mathbb{R}^n$. Two datasets $\underline{x}$ and $\underline{x}'$, both of size $n$, are called \emph{neighbors} if they differ by one element.  
\begin{definition}[Differential Privacy, \cite{DMNS06}]
	A randomized algorithm $M: \Omega^n \rightarrow \Omega_{M}$ is $(\eps, \delta) $ \emph{differentially private} if for all neighboring $\underline{x}, \underline{x}' \in \Omega^n$ and for all measurable sets of outputs $S \subseteq \Omega_{M}$, we have
	$$\pr{M(\underline{x}) \in S} \leq e^{\eps} \cdot \pr{M(\underline{x}') \in S} + \delta,$$
	where the probability is over the randomness of $M$. 
\end{definition}
Intuitively, this captures privacy because arbitrarily changing one individual's data (i.e. changing a row of $\underline{x}$ to
obtain $\underline{x}'$) has only a small effect on the output distribution of $M$.  Typically we think of $\eps$
as a small constant (e.g. $\eps=.01$) while $\delta$ should be cryptographically small (in particular, much smaller than $1/n$) to obtain a satisfactory formulation of privacy.   The case that $\delta=0$ is often called {\em pure} differential privacy, while $\delta>0$ is known as {\em approximate} differential privacy. 

Nontrivial differentially private algorithms are necessarily randomized (as the probabilities in the definition above
are taken only over the randomness of $M$, which is important for the interpretation and composability of the definition), and thus such algorithms work by injecting randomness into statistical computations to obscure the effect of each individual.

The most basic differentially private algorithm is the Laplace Mechanism~\citep{DMNS06}, which approximates an arbitrary
a function $f : \Omega^n\rightarrow \R$ by adding Laplace noise:
$$M(\ux) = f(\ux) + Z, \text{where $Z\sim \Lap(\GS_f/\eps)$.}$$
Here $\Lap(\tau)$ is the Laplace distribution with scale $\tau$ (which has standard deviation proportional to $\tau$),
and $\GS_f$ is the {\em global sensitivity} of $f$ --- the maximum of $|f(\ux)-f(\ux')|$ over all pairs of neighboring
datasets $\ux,\ux'$.

In particular, if $f(\ux)$ is the {\em empirical mean} 
of the dataset $\ux$ and $\Omega=[-B,B]$, then $\GS_f = 2B/n$,
so $M(\ux)$ approximates the empirical mean to within additive error $O(B/\eps n)$ with high probability.

\subsection{Statistical Inference with Differential Privacy}

The Laplace mechanism described above is about estimating a function $f(\ux)$ of the dataset $\ux$, rather than
the population from which $\ux$ is drawn, and much of the differential privacy literature is about
estimating such empirical statistics.  There are several important exceptions, the earliest being the work
on differentially private PAC learning (\cite{BlumDwMcNi05,kasiviswanathan2011can}), but still many basic
statistical inference questions have not been addressed.

However, a natural approach was already suggested in early works on differential privacy.  In many cases, we know
that population statistics are well-approximated by empirical statistics, and thus we can try to estimate these empirical statistics with differential privacy.  For example, the population mean $\mu$ for a normal population is well-approximated by the sample mean $\overline{X}$, which we can estimate using the Laplace mechanism: 
$$M(X_1,\ldots,X_n) = \overline{X} + Z, \text{where $Z\sim \Lap(2B/\eps n)$.}$$

On the positive side, 
observe that the noise being introduced for privacy vanishes
linearly in $1/n$, whereas
$\overline{X}$ converges to the
population mean at a rate of
$1/\sqrt{n}$, so asymptotically
we obtain privacy ``for free'' compared to the (optimal) non-private estimator $\overline{X}$.

However, this rough analysis hides some important issues.  First, it is misleading to look
only at the dependence on $n$.  The other parameters, such as $\sigma$, $\epsilon$, and $B$ can be quite significant
and should not be treated as constants.  Indeed $\sigma/\sqrt{n} \gg B/\epsilon n$ only when $n \gg (B/\epsilon \sigma)^2$, which means that the asymptotics only kick in at a very large value of $n$.  Thus it is important to determine
whether the dependence on these parameters is necessary or can be improved.  Second, the parameter $B$
is supposed to be a (worst-case) bound on the range of the data, which is incompatible with a modeling the population as following a normal distribution (which is supported on the entire real line).  
Thus, there have been several works seeking the best asymptotic approximations we can obtain for population statistics
under differential privacy, such as \cite{dworkLei,smith2011privacy,wassermanzhou,wasserman2012minimaxity,hall2013differential,duchi2013localNIPS,duchi2013localFOCS,barber2014privacy}.

\subsection{Conservative Statistical Inference with DP} \label{sec:intro-conservative}

The works discussed in the previous section focus on providing point estimates for population quantities, but
as mentioned earlier, it is also important to be able to provide measures of statistical significance, to prevent
analysts from drawing incorrect conclusions from the results.  These measures of statistical significance
need to take into account the uncertainty coming both from the sampling of the data and from the noise introduced for privacy.  Ignoring the noise introduced for privacy can result in wildly incorrect results at finite sample sizes, as
demonstrated empirically many times (e.g. \cite{Fienberg,karwa2012differentially,karwa2016inference}) and this can have severe consequences.  For example, \cite{fredrikson2014privacy} found that naive use of differential privacy in calculating warfarin dosage would lead to unsafe levels of medication, but of course one should never use any sort of statistics for life-or-death decisions without some analysis of statistical significance.  

Since calculating the exact statistical significance of differentially private computations seems difficult in general, we advocate {\em conservative} estimates of significance.  
That is,  we require that 
$\mathbb P (I(X_1,\ldots,X_n)\ni \mu)\geq 1-\alpha,$
for {\em all} values of $n$, values of the population parameters, and values of the privacy parameter.

For sample sizes that are too small or privacy parameters that are too aggressive, we may achieve this property by allowing the algorithm to sometimes produce an extremely large confidence interval, but that is preferable to producing a small interval that does
not actually contain the true parameter which may violate the desired coverage property.  Note that what constitutes a sample size that is ``too small'' can depend on the unknown parameters of the population (e.g. the unknown variance $\sigma^2$) and their interplay with other parameters (such as the privacy parameter $\epsilon$). 

Returning to our example of estimating a normal mean with known variance under differential privacy, if we use the Laplace Mechanism to approximate the empirical mean (as discussed above), we can obtain a conservative confidence interval for the population mean by increasing the length of classical, non-private confidence interval to account for the likely magnitude
of the Laplace noise.   More precisely, starting with the differentially private mechanism 
$$M(X_1,\ldots,X_n) = \overline{X} + Z, \text{where $Z\sim \Lap(2B/\eps n)$,}$$
the following is a $(1-O(\alpha))$-level confidence interval for the population mean $\mu$:
$$I(X_1,\ldots,X_n) = M(X_1,\ldots,X_n)  \pm \left(\frac{\sigma}{\sqrt{n}}\cdot z_{1-\alpha/2} 
+ \frac{B}{\eps n}\cdot \log(1/\alpha)\right).$$
The point is that with probability $1-O(\alpha)$, the Laplace noise $Z$ has magnitude at most
$(B/\eps n)\cdot \log(1/\alpha)$, so increasing the interval by this amount will preserve coverage (up to an $O(\alpha)$
change in the probability).   Again, the privacy guarantees of the Laplace mechanism relies on the data points being
guaranteed to lie in $[-B,B]$; otherwise, points need to be clamped to lie in the range, which can bias the empirical mean and compromise the coverage guarantee.  Thus, to be safe, a user may choose a very large value of $B$, but then this
makes for a much larger (and less useful) interval, as the length of the interval grows linearly with $B$.  Thus, a natural research question (which we investigate) is whether such a choice and corresponding cost is necessary.

Conservative hypothesis testing with differential privacy, where we require that the significance level is
at least $1-\alpha$, was advocated by ~\cite{gaboardi2016differentially}.  Methods aimed at calculating the combined
uncertainty due to sampling and privacy (for various differentially private algorithms) were given in \cite{vu2009differential,williams2010probabilistic,karwa2012differentially,karwa2015private,karwa2014differentially, karwa2016inference, gaboardi2016differentially,solea2014differentially,wang2015differentially,kiferRogers}, but generally the utility of these methods (e.g. the expected length of a confidence interval or
power of a hypothesis test) is only evaluated empirically or the conservativeness only holds in a particular asymptotic
regime.  Rigorous, finite-sample analyses of conservative inference were given in \cite{sheffet2017differentially} for confidence intervals on the coefficients from ordinary least-squares regression (which can be seen as a generalization of the problem we study to multivariate Gaussians) and in \cite{cai2017priv} for hypothesis testing of discrete distributions.  However, neither paper provides matching lower bounds, and in particular, the algorithms of \cite{sheffet2017differentially} only apply for bounded data (similar to the basic Laplace mechanism). 
In our work, we provide a comprehensive theoretical analysis of conservative differentially private confidence intervals for a normal mean, with both algorithms and lower bounds, without any bounded data assumption.

Before stating our results, we define more precisely the notion of a (conservative) $(1-\alpha)$-level confidence set. Let $\mathcal D = \{ \mathbb D_{\theta, \gamma}\}_{\theta \in \Theta, \gamma \in \Gamma}$ be a family of distributions supported on $\mathbb R$ where $\theta \in \Theta \subseteq \mathbb R$ is a real valued parameter and $\gamma \in \Gamma \subseteq \mathbb R^k$ is a vector of \emph{nuisance parameters}. A nuisance parameter is an unknown parameter that is not a primary object of study, but must be accounted for in the analysis. For example, when we consider estimating the
mean of a normal distribution with unknown variance, the variance is a nuisance parameter. We write $X_1, \ldots, X_n \overset{iid}{\sim}  \mathbb D_{\theta,\gamma}$ if $X_1, \ldots, X_n$ is an independent and identically distributed random sample from a distribution $ \mathbb D_{\theta,\gamma}$.  We sometimes abuse the notation and write $\mathbb D_{\theta}$ instead of $\mathbb D_{\theta, \gamma}$ when $\gamma$ is clear from the context.
\begin{definition}[$(1-\alpha)$-level confidence set]
	Let $\alpha \in (0,1)$. Let $X_1, \ldots, X_n \overset{iid}{\sim} \mathbb D_{\theta, \gamma}$ where $\theta \in \Theta \subseteq \mathbb R$ and $\gamma \in \Gamma \subseteq \mathbb R^k$ is a vector of nuisance parameters. A $(1-\alpha)$-level confidence set for $\theta$ with sample complexity $n$ is a (possibly randomized) measurable function $I: \mathbb R^n \rightarrow \mathbb S$, where $\mathbb S$ is a set of measurable subsets of $\mathbb R$, such that for all $\theta \in \Theta$ and $\gamma \in \Gamma$, we have
	$$\prm[X_1,\ldots, X_n \sim \mathbb D_{\theta,\gamma}][I]{I(X_1, \ldots, X_n) \ni \theta  } \geq 1 - \alpha,$$
	where the probability is taken over the randomness of both $I$ and the data $X_1, \ldots, X_n$.
\end{definition}

\subsection{Our Results}

As discussed above, in this paper we develop conservative differentially private estimators of confidence intervals for the mean $\mu$ of a normal distribution with known and unknown variance $\sigma^2$. Our algorithms are designed to be differentially private for all input datasets and they provide $(1-\alpha)$-level coverage whenever the data is generated from a normal distribution. Unlike the Laplace mechanism described above and many other differentially private algorithms, we do not make any assumptions on the boundedness of the data. Our pure DP (i.e. $(\eps, 0)$-DP) algorithms assume that the mean $\mu$ and variance $\sigma^2$ lie in a bounded (but possibly very large) interval, and we show (using lower bounds) that such an assumption is necessary. Our approximate (i.e. $(\eps, \delta)$) differentially private algorithms do not make any such assumptions, i.e. both the data and the parameters ($\mu, \sigma^2$) can remain unbounded. We also show that the differentially private estimators that we construct have nearly optimal expected length, up to logarithmic factors. This is done by proving lower bounds on the length of differentially private confidence intervals. A key aspect of the confidence intervals that we construct is their conservativeness --- the coverage guarantee holds in finite samples, as opposed to only holding asymptotically. We also show that as $n \rightarrow \infty$, the length of our differentially private confidence intervals is at most $1 + o(1)$ factor larger than length of their non-private counterparts.  

Let $X_1, \ldots, X_n$ be an independent and identically distributed (\emph{iid}) random sample from a normal distribution with an unknown mean $\mu$ and variance $\sigma^2 $, where $\mu \in (-R,R)$ and $\sigma \in (\sigma_{\min}, \sigma_{\max})$.  Our goal is to construct $(\eps, \delta)$-differentially private $(1 - \alpha)$-level confidence sets for $\mu$ in both the known and the unknown variance case, i.e. we seek a set $I = I(X_1, \ldots, X_n)$ such that
\begin{enumerate}
	\item $I(X_1, \ldots, X_n)$ is a $(1 - \alpha)$-level confidence interval, and
	\item $I(x_1, \ldots, x_n)$ is $(\eps, \delta)$-differentially private.
    \item $\Em[X_1, \ldots, X_n, I][]{I(X_1, \ldots, X_n) }$ is as small as possible.
\end{enumerate}

\paragraph{Known Variance:} 
We prove the following result for estimating the confidence interval with the privacy constraint:

\begin{theorem}[known variance case]
	\label{thm:intro1}
	There exists an $(\eps, \delta)$-differentially private algorithm that on input $X_1, \ldots, X_n \overset{iid}{\sim} N(\mu, \sigma^2)$ with known $\sigma^2$ and unknown mean $\mu \in (-R,R)$ outputs a $(1-\alpha)$-level confidence interval for $\mu$. Moreover, if 
	$$n > \frac{c_1}{\eps}\min\left\{
	\log \left(\frac{R}{\sigma}\right), 
	\log \left(\frac{1}{\delta}\right)
	\right\} + \frac{c_2}{\eps} \log\left(\frac{1}{\alpha}\right),$$ (where $c_1$ and $c_2$ are universal constants) 
	then the interval is of fixed width $\beta$ where
	$$
	\beta \leq \max\left\{\frac{\sigma}{\sqrt{n}}\O{\sqrt{\log\left(\frac{1}{\alpha}\right)}} , 
	\frac{\sigma }{\eps n}  \mathrm{polylog}{\left(\frac{n}{\alpha}\right)}
	\right\}
	$$
\end{theorem} 
Theorem \ref{thm:intro1} asserts that there exists a differentially private algorithm that outputs a fixed width $(1-\alpha)$-level confidence interval for any $n$. Moreover, when $n$ is large enough, the algorithm outputs a confidence interval of length $\beta$ which is non-trivial in the sense that $\beta \ll R$. Specifically, $\beta$ is a maximum of two terms: The first term is $\O{\sigma \sqrt{\log(1/\alpha)/n}}$ which is the same as the length of the non-private confidence interval discussed in Section~\ref{sec:intro-normal} up to constant factors. The second term is $\O{\sigma/(\eps n)}$ up to polylogarithmic factors $-$ it goes to $0$ at the rate of $\Otilde{1/n}$ which is faster than the rate at which the first term goes to $0$. Thus for large $n$ the increase in the length of the confidence interval due to privacy is mild.   Note that, unlike the basic approach based on the Laplace mechanism discussed in
Section~\ref{sec:intro-conservative}, the length of the confidence interval has no dependence on the range of the data, or even the range $(-R,R)$ of the mean $\mu$. 

The sample complexity required for obtaining a non-trivial confidence interval is the minimum of two terms: $\O{(1/\eps)\log (R/\alpha \sigma)}$ and $\O{(1/\eps) \log (1/\alpha \delta)}$. 
The dependence of sample complexity on $R/\sigma$ is only logarithmic. Thus one can choose a very large value of $R$.  Moreover, when $\delta >0$, we can set $R = \infty$ and hence there is no dependence of the sample complexity on $R$.  

The first term in the length of the confidence interval in Theorem \ref{thm:intro1} hides some constants which can lead to a constant multiplicative factor increase in the length of the differentially private confidence intervals when compared to the non-private confidence intervals. We show that it is possible to eliminate this multiplicative increase and obtain differentially private confidence intervals with only additive increase in the length:
\begin{theorem}
	\label{thm:intro2}
	There exists an $(\eps, \delta)$-differentially private algorithm that on input $X_1, \ldots, X_n \overset{iid}{\sim} N(\mu, \sigma^2)$ with known $\sigma^2$ and unknown mean $\mu \in (-R,R)$ outputs a $(1-\alpha)$-level confidence interval of $\mu$. Moreover, if
		$$n > \frac{c}{\eps}\min\left\{
		\log \left(\frac{R}{\sigma}\right), 
		\log \left(\frac{1}{\delta}\right)
		\right\} + \frac{c}{\eps} \log\left(\frac{\log(1/\eps)}{\alpha}\right)$$ (where $c$ is a universal constant)
	then
	\begin{align*}
	\beta =\frac{2\sigma}{\sqrt{n}}z_{1 - \alpha/2} + \frac{\sigma}{\eps n} \mathrm{polylog}\left(\frac{n}{\alpha}\right)
	\end{align*}
	
\end{theorem}

Theorem \ref{thm:intro2} asserts that with a small change to the sample complexity on $n$ by an additive term of $(1/\eps) \cdot \log (1/\eps)$, we can achieve an additive increase in the length $\beta$ of the confidence interval as opposed to a multiplicative increase.
Note that the first term in $\beta$ exactly matches the length of the 
non-private confidence interval, namely $2\sigma z_{1-\alpha/2}/\sqrt{n}$, while the second term vanishes
more quickly as a function of $n$.\footnote{We note that when the range $R$ of the mean is bounded, as we require
for our pure differentially private algorithms, the length of a non-private algorithm can be improved, but the improvement is insignificant in the regime of parameters we are interested in, namely when $R\geq \Omega(\sigma)$.  See
Theorem~\ref{thm:lowerboundwithoutprivacy}.}
An important point is that we retain an $n$ rather than an $\eps n$ in the first term, contrary to the common belief that differential privacy has a price of $1/\eps$ in sample size. (See \cite{steinke2015between} for a proof of such a statement for computing summary statistics of a dataset rather than inference, and \cite{hay2016principled} for an informal claim along these lines.) 
 
\paragraph{Unknown Variance:} 
Our $(\eps, \delta)$-differentially private confidence interval in the unknown variance case is as follows:

\begin{theorem}[unknown variance case]
	\label{thm:intro3}
	There exists an $(\eps, \delta)$-differentially private that on input $X_1, \ldots, X_n \overset{iid}{\sim} N(\mu, \sigma^2)$ with unknown mean $\mu \in (-R,R)$ and variance $\sigma^2 \in (\sigma_{\min}, \sigma_{\max})$ always outputs a $(1-\alpha)$-level confidence interval of $\mu$. Moreover, if 
	$$n \geq  \frac{c}{\eps} 
	\min \left\{
	\max \left\{
	\log \left(\frac{R}{\sigma_{\min}}\right), \log \left(\frac{\sigma_{\max}}{\sigma_{\min}}\right)
	\right\},
	\log \left(\frac{1}{\delta}\right)
	\right\} + \frac{c}{\eps}\log \left(\frac{\log \left(\frac{1}{ \eps}\right)}{\alpha}\right),
	$$ (where $c$ is a universal constant),
	then the expected length of the interval $\beta$ is such that
	$$
	\beta \leq \max\left\{\frac{\sigma}{\sqrt{n}}  \O{\sqrt{\log\left(\frac{1}{\alpha}\right)}} , 
	\frac{\sigma }{\eps n}  \mathrm{polylog}{\left(\frac{n}{\alpha}\right)}
	\right\}
	$$
		
\end{theorem} 

As in the known variance case, Theorem \ref{thm:intro3} asserts that there exists an $(\eps,\delta)$ differentially private algorithm that always outputs an $(1-\alpha)$ confidence interval of $\mu$ for all $n$.  If $n$ is large enough, the length of the confidence interval is a maximum of two terms, where the first term is same as the length of the non-private confidence interval and the second term goes to $0$ at a faster rate.  

As before the dependence of sample complexity on $R/\sigma_{\min}$ and $\sigma_{\max}/\sigma_{\min}$ is logarithmic, as opposed to linear. Hence we can set these parameters to a large number. Moreover, when $\delta > 0$, we can set $R$ and $\sigma_{\max}$ to be $\infty$ and $\sigma_{\min}$ to be $0$. Thus when $\delta > 0$, there are no assumptions on the boundedness of the parameters. 

Finally, along the lines of Theorem \ref{thm:intro2}, at the cost of a minor increase in sample complexity, we can obtain a differentially private algorithm that has only additive increase in the length of the confidence interval that is asymptotically vanishing relative to the non-private length.  Specifically, we can obtain an interval with length
\begin{align*}
	\beta  \leq \frac{2\sigma}{\sqrt{n}}\cdot k_nt_{n-1,1 - \alpha/2} + \frac{\sigma}{\eps n} \mathrm{polylog}\left(\frac{n}{\alpha}\right),
\end{align*}
where again the first term is exactly the same as in the non-private case (see Section~\ref{sec:intro-normal}) and the
second term vanishes more quickly as a function of $n$.

\paragraph{Lower Bounds.}  We also prove lower bounds on the length of any $(1-\alpha)$-level $(\eps, \delta)$-differentially private confidence set of expected size $\beta$:

\begin{theorem}[Lower bound] \label{thm:intro-lower}
	Let $M$ be any $(\eps, \delta)$-differentially private algorithm that on input $X_1, \ldots, X_n \overset{iid}{\sim} N(\mu, \sigma^2), \mu \in (-R,R)$ produces a $(1 - \alpha)$-level confidence set of $\mu$ of expected size $\beta$. 
	If $\delta < \alpha/2n$, then
	$$ \beta \geq c \cdot \min\left\{ \frac{\sigma}{\eps n} \log\left(\frac{1}{\alpha}\right), R\right\}
	$$
	Moreover, if $\beta < \sigma < R$, then
	$$
	n \geq c\min \left(\frac{1}{\eps}\log \left(\frac{R}{\sigma}\right), \frac{1}{\eps}\log \left(\frac{1}{\delta}\right)\right) + \frac{c}{\epsilon}\log \left(\frac{1}{\alpha}\right)
	$$
	where $c$ is a universal constant.
\end{theorem} 
Note that the first lower bound says that we must pay $\Omega \left(\sigma/(\eps n) \cdot \log(1/\alpha) \right)$ in the length of the confidence interval when $R$ is very large. Our algorithms come quite close to this lower bound with an extra factor of $\mathrm{polylog}(n/\alpha)$. The second lower bound shows that the  sample complexity required by Theorem \ref{thm:intro1} is necessary to obtain a confidence interval that saves more than a factor of 2 over the trivial interval $(-R,R)$.  By setting $\sigma=\sigma_{\min}$, the sample complexity lower bound also matches that of Theorem~\ref{thm:intro3} in our parameter regime of interest, namely when $R\geq \Omega(\sigma_{\max})$.

\subsection{Techniques}
\label{sec:techniques}

\paragraph{Known Variance Algorithms:}
Our algorithms for the known variance case (Theorems~\ref{thm:intro1} and \ref{thm:intro2}) are based on simple 
Laplace-mechanism-based confidence interval discussed in Section~\ref{sec:intro-conservative}, except that we calculate a suitable bound $B$ based on the data in a differentially private manner, rather than having it be an input provided by
a data analyst.  Specifically, we give a differentially private algorithm $\Mrange$ that takes $n$ real numbers and
outputs an interval (which need not be centered at 0) such that for every $\mu\in (-R,R)$, when
$X_1,\ldots,X_n\sim\Normal(\mu,\sigma^2)$, we have:
\begin{enumerate}
\item With probability at least $1-\alpha$ over $X_1,\ldots,X_n$ and the coins of $\Mrange$, we have $\{X_1,\ldots,X_n\}\subseteq \Mrange(X_1,\ldots,X_n)$, and \label{prop:contain}
\item With probability 1, $|\Mrange(X_1,\ldots,X_n)| \leq O(\sigma\cdot \sqrt{\log(n/\alpha)})$.
\end{enumerate}
Thus, if we clamp all datapoints to lie in the interval $\Mrange(X_1,\ldots,X_n)$ (which will usually have no effect
for data that comes from our normal model, by Property~\ref{prop:contain}), we can calculate an approximate mean and thus construct a confidence interval using Laplace noise of scale $O(\sigma\cdot \sqrt{\log(n/\alpha)}/\eps n)$.  

Now, estimating the range of a dataset with differential privacy is impossible to do with any nontrivial accuracy in the worst case, so we must exploit the distributional assumption on our dataset to construct $\Mrange$.  Specifically, we exploit the following properties of normal data:
\begin{enumerate}
\item A vast majority of the probability mass of $\Normal(\mu,\sigma^2)$ is concentrated in an interval of width $O(\sigma)$ around the mean $\mu$. \label{prop:normalcenter}
\item With probability at least $1-\alpha$, all datapoints $X_1,\ldots,X_n$ are at distance at most $O(\sigma\cdot \sqrt{\log(n/\alpha)})$ from $\mu$. \label{prop:normalrange}
\end{enumerate}
Similar properties hold for many other natural parameterized families of distributions, changing the factor of $\sqrt{\log(n/\alpha)}$ according to the concentration properties of the family.

Given these properties, $\Mrange$ works as follows:  we partition the original range $(-R,R)$ (where $R$
might be infinite) into ``bins'' (intervals) of width $O(\sigma)$, and calculate an differentially private approximate histogram of how many points lie in each bin.  By Property~\ref{prop:normalcenter} and a Chernoff bound, with high probability, the vast majority of our normally distributed data points will be in the bin containing $\mu$ or one of the neighboring bins.  Existing algorithms for differentially private 
histograms~\citep{DMNS06,bun2016simultaneous} allow us to identify one of these heavy bins 
with probability $1-\alpha$, provided 
$n\geq O(\min\{\log(K/\alpha),\log(1/(\delta\alpha))\})/\eps$, where $K=O(R/\sigma)$ is the number of bins.  After identifying such a bin,
Property~\ref{prop:normalrange} tells us that 
we can simply expand the bin by $O(\sigma\cdot\sqrt{\log(n/\alpha)})$ on each side and include all of the datapoints with high probability.  This proof sketch gives a $(1-O(\alpha))$-level confidence interval, and redefining $\alpha$ yields
Theorem~\ref{thm:intro1}.  To obtain, Theorem~\ref{thm:intro2}, we set parameters more carefully so that the 
failure probability in estimating the range is much smaller than $\alpha$, say $\alpha/\poly(n)$, which increases the sample complexity of the histogram algorithms only slightly.

This general approach, of finding a differentially private estimate of the range and using it to compute a differentially private mean are inspired by the work of \cite{dworkLei}. They present an $(\eps, \delta)$ differentially private algorithm to estimate the \emph{scale} of the data. They use the estimate of scale to obtain a differentially private estimate of the median without making any assumptions on the range of the data.  Their algorithms require $\delta > 0$. In contrast, our range finding algorithms for Gaussian data work for $\delta =0$  without making any assumptions on the range of the data, but instead assume that the parameters need to be bounded). Our algorithms also handle the unknown variance case, as discussed below. Also, while the general idea of eliminating the dependence on range of the data is similar, the underlying techniques and privacy and utility guarantees are different. 

\paragraph{Unknown Variance Algorithms:}
For the case of an unknown variance (Theorem~\ref{thm:intro3}), we begin with the observation that our range-finding algorithm discussed above only needs a constant-factor approximation to the variance $\sigma^2$.  Thus, we will begin
by calculating a constant-factor approximation to $\sigma^2$ in a differentially private manner, and then estimate the range as above.  To do this, we consider the dataset of size $n/2$ given by $|X_1-X_2|, 
|X_3-X_4|, |X_5-X_6|, \ldots, |X_{n-1}-X_n|$.
Here each point is distributed as the absolute value of a $\Normal(0,2\sigma^2)$ random variable, which has the vast majority of its probability mass on points of magnitude $\Theta(\sigma)$.  
Thus, if we partition the interval
$(\sigma_{\min},\sigma_{\max})$ into bins of the form $(2^i,2^{i+1})$ and apply an approximate histogram algorithm, the heaviest bin will give us an estimate of $\sigma$ to within a constant factor.  Actually, to analyze the expected
length of our confidence interval, we will need that
our estimate of $\sigma$ is within a constant factor of the true value not only with high probability but also
in expectation; this requires a finer analysis of the histogram algorithm, where the probability of picking any bin 
decays linearly with the probability mass of that bin (so bins further away from $\mu$ have exponentially decaying
probability of being chosen).  
Note that this approach for approximating $\sigma$ also exploits the symmetry of a normal distribution, so that $|X_i-X_{i+1}|$ is likely to have magnitude $\Theta(\sigma)$, independent of $\mu$; it should generalize to many other common symmetric distribution families.  For non-symmetric families, one could instead use differentially private algorithms for releasing threshold functions (i.e. estimating quantiles) at the price of a small dependence on the ranges even when $\delta>0$.  (See \citet[Sec. 7.2]{vadhan2017complexity} and references therein.)

Now, once we have found the range as in the known-variance case, we can again use the Laplace mechanism to estimate the
empirical mean to within additive error $\pm O(\sigma\cdot \sqrt{\log(n/\alpha)}/\epsilon n)$.  And we can use our constant-factor
approximation of $\sigma$ to estimate the size of the non-private confidence interval to within a constant factor.
This suffices for Theorem~\ref{thm:intro3}. 
But to obtain the tighter bound, where we only pay an additive increase over the length of the non-private interval, we cannot just use a constant-factor approximation of the variance.  Instead, we
also use the Laplace mechanism to estimate the sample variance $s^2 = \frac{1}{n-1} \sum_{i=1}^n (X_i - \bar{X})^2.$
Our bound on the range (with clamping) ensures that $s^2$ has global sensitivity $O(\sigma^2\cdot \log(n/\alpha))/(n-1)$, and thus can be
estimated quite accurately.
\paragraph{Lower Bounds:}
For our lower bounds (Thm~\ref{thm:intro-lower}), we observe that the expected length of a confidence set
$M(X_1,\ldots,X_n)$ can be written as 
$$\int_{\mu'} \mathbb P_{\mu, M} \left(\mu'\in M(X_1,\ldots,X_n) \right)$$
where $\mu'$ ranges over $(-R,R)$ and the $\mathbb P_{\mu, M}$ notation indicates that probability is taken over $(X_1,\ldots,X_n)$ generated according to $\Normal(\mu,\sigma^2)$ for a particular value of $\mu$, and over the mechanism $M$. 
Next, we use the differential privacy guarantee to deduce that 
$$\mathbb P_{\mu,M} \left(\mu'\in M(X_1,\ldots,X_n)\right) \geq  e^{-6\eps n \cdot d}\cdot \mathbb P_{\mu',M} \left(\mu'\in M(X_1,\ldots,X_n)\right)-4n\delta \cdot d,$$
where $d\leq \min\{1,|\mu - \mu'|/\sigma\}$ is the total variation distance between $\Normal(\mu,\sigma^2)$ and $\Normal(\mu',\sigma^2)$.
(This can be seen as a distributional analogue of the ``group privacy'' property used in ``packing lower bounds''
for calculating empirical statistics under differential privacy~\citep{hardt2010geometry,beimel2010bounds,wassermanzhou, hall2011random},
and is also a generalization of the ``secrecy of the sample'' property of differential privacy~\citep{kasiviswanathan2011can,smith2009differential,BNSV}. 
Finally, we know that $\mathbb P_{\mu',M}\left(\mu'\in M(X_1,\ldots,X_n)\right) \geq 1-\alpha$ by the coverage property of $M$, yielding
our lower bound (after some calculations).

\subsection{Directions for Future Work}

The most immediate direction for future work is to close the (small) gaps between our upper and lower bounds.  Most
interesting is whether the price of privacy in the length of confidence intervals needs to be even additive, as
in Theorem~\ref{thm:intro1}.  Our lower bound only implies that the length of a differentially private confidence interval
must be at least the {\em maximum} of a privacy term (namely, the lower bound in Theorem~\ref{thm:intro-lower}) and the non-private length (cf. Theorem~\ref{thm:lowerboundwithoutprivacy}), rather than the sum.  In particular,
when $n$ is sufficiently large, the non-private length is larger than the privacy term, and Theorem~\ref{thm:intro-lower} leaves open the possibility that a differentially private confidence interval can have {\em exactly} the same length as
a non-private confidence interval.  This seems unlikely, and it would be interesting to prove that there must be some price to privacy even if $n$ is very large. 

We came to the problem of constructing confidence intervals for a normal mean as part of an effort to bring
differential privacy to practice in the sharing of social science research data through the design of the
software tool PSI~\citep{psipaper}, as confidence intervals are a workhorse of data analysis in the social sciences.
However,
our algorithms are not optimized for practical performance, but rather for asymptotic analysis of the confidence interval length. Initial experiments indicate that alternative approaches (not just tuning of parameters) may be needed to 
reasonably sized confidence intervals (e.g. length at most twice that of the non-private length)
handle modest sample sizes (e.g. in the 1000's).  
Thus designing practical differentially private algorithms for confidence intervals remains an important open problem, whose solution could have wide applicability.

As mentioned earlier, we expect that much of the modelling and techniques we develop should also be applicable
more widely.  In particular, it would be natural to study the estimation of other population statistics, and 
families of distributions, such as other continuous random variables, Bernoulli random
variables, and multivariate families.  In particular, a natural generalization of the problem we consider is to 
construct confidence intervals for the parameters of a (possibly degenerate) multivariate Gaussian, which is closely related to the problem of ordinary least-squares regression (cf. \cite{sheffet2017differentially}).  

Finally, while we have advocated for conservative inference at finite sample size, to avoid spurious conclusions
coming from the introduction of privacy, many practical, non-private inference methods rely on asymptotics also for measuring statistical significance.  In particular, the standard confidence interval for a normal mean with unknown variance and its corresponding hypothesis test (see Section~\ref{sec:intro-normal}) is often applied on non-normal data,
and heuristically justified using the Central Limit Theorem.  (This is heuristic since the rate of convergence depends
on the data distribution, which is unknown.)  Is there a criterion to indicate what asymptotics are ``safe''?  In particular, can we formalize the idea of only using the ``same'' asymptotics that are used without privacy? \cite{kiferRogers} analyze their hypothesis tests using asymptotics that constrain the setting of the privacy parameter 
in terms of the sample size $n$ (e.g. $\eps\geq \Omega(1/\sqrt{n})$), but it's not clear that this relationship is safe to assume in general.

\subsection{Organization}
The rest of the paper is organized in the following manner. In Section \ref{sec:prelim}, we introduce some preliminary results on DP and techniques such as Laplace mechanism and histogram learners that are needed for our algorithms. In Section \ref{sec:range}, we present differentially private algorithms to estimate the range of the data; these algorithms serve as building blocks for estimating differentially private confidence intervals. In Sections \ref{sec:ciknownvar} and \ref{sec:ciunkownvar}, we present $(\eps, \delta)$ differentially private algorithms to estimate an $(1-\alpha)$-level confidence interval of $\mu$ with known and unknown variance respectively. Section \ref{sec:lowerbounds} is devoted to lower bounds.

\section{Preliminaries}
\label{sec:prelim}
\subsection{Notation}
We use $\log$ to denote natural log to the base $e$, unless otherwise noted. Random variables are denoted by capital roman letters and their realization by small roman letters. For example $X$ is a random variable and $x$ is its realization. We write  $(X_1, \ldots, X_n) \overset{iid}{\sim} \mathbb D_{\theta, \gamma}$ or equivalently $\underline{X} \sim \mathbb D_{\theta}$ or to denote a sample of $n$ independent and identically random variables from the distribution $\mathbb D_{\theta, \gamma}$, where $\theta \in \Theta \subseteq \mathbb R$ and $\gamma \in \Gamma \subseteq \mathbb R^k$ is a vector of nuisance parameters. We sometimes abuse notation and write $\mathbb D_{\theta}$ instead of $\mathbb D_{\theta,\gamma}$. Estimators of parameters $\theta$ are denoted by $\hat \theta = \hat\theta(X_1, \ldots, X_n)$. A differentially private mechanism is denoted by $M(X_1,\ldots, X_n)$ or $M(\underline{X})$.

There are two sources of randomness in our algorithms: The first source of randomness is from the coin flips made by the estimator or algorithm and the dataset is considered fixed.  We use the notation of conditioning to denote the probabilities and expectations with respect to the privacy mechanism, when the data is considered to be fixed. Specifically, conditional probability is denoted by $\pr{\cdot | X = x}$ and conditional expectation is denoted by $\E{\cdot|X=x}$. The second source of randomness comes from assuming that the dataset is a sample from an underlying distribution $\mathbb D_{\theta,\gamma}$. The probability and expectation with respect to this distribution is denoted by $\mathbb E_{\underline{X} \sim \mathbb D_{\theta}}[\cdot]$, and $\mathbb  P_{\underline{X} \sim \mathbb D_{\theta}}(\cdot)$. While the privacy guarantees are with respect to a fixed dataset, the accuracy guarantees are with respect to both the randomness in the data and the mechanism. In such cases, we state both sources of randomness by writing 
$$\prm[\underline{X} \sim \mathbb D_{\theta}][M]{M(\underline{X}) \in S }, $$
where $S$ is any measurable event and the subscripts denote the sources of randomness. 
 
\subsection{Differential Privacy}
We will present some key properties of differential privacy that we make use of in this paper.
One of the attractive properties of Differential Privacy is its ability to compose. To prove the privacy properties of an algorithm, we will rely on the fact that an differentially private algorithm that runs on a dataset and an output of a previous differentially private computation is also differentially private.

\begin{lemma}[Composition of DP, \cite{DMNS06}]
	\label{thm:compose}
	Let $M_1 : \Omega^n \rightarrow \Omega_{M_1}$ be an $(\eps_1, \delta_1)$ differentially private algorithm. Let $M_2: \Omega^n\times \Omega_{M_1} \rightarrow \Omega_{M_2}$ be such that $M_2(\cdot, \omega)$ is an $(\eps_2, \delta_2)$ differentially private algorithm for every fixed $\omega \in \Omega_{M_1}$. Then the algorithm $M(\underline{x}) = M_2(\underline{x}, M_1(\underline{x}))$ is $(\eps_1 + \eps_2, \delta_1 + \delta_2)$ differentially private. 
\end{lemma}

In many algorithms, we will rely on a basic mechanism for Differential Privacy that works by adding Laplace noise. Let $f: \Omega^n \rightarrow \mathbb{R}^k$ be any function of the dataset that we wish to release an approximation of. The \emph{global sensitivity} of $f$ is defined as
$$GS_f = \underset{\underline{x}, \underline{x}' neighbors}{\max} |f(\underline{x}) - f(\underline{x}')|_1$$
where $|y|_1 = \sum_{i} |y_i|$ is the $l_1$ norm of the vector $y$.
\begin{lemma}[The Laplace Mechanism, \cite{DMNS06}]
	\label{thm:LapMech}
	Let $f: \Omega^n \rightarrow \mathbb{R}^k$ be a function with global sensitivity at most $\Delta$. The mechanism 
    $$M(\underline{x}) = f(\underline{x}) + Z$$ 
is $(\eps, 0)$ differentially private, where $\underline{x}$ is the input dataset, $Z$ is a $k$ dimensional random vector where each component of $Z$ is an independent Laplace distribution (defined in \ref{prop:lap}) with mean $0$ and scale parameter $b = \Delta/ \epsilon$ 
\end{lemma}




In many estimators we design, we need a differentially private mechanism for finding a heaviest bin from a (possibly countably infinite) collection of bins, i.e. the bin with the maximum probability mass under the data generating distribution $\mathbb D$. Formally, let $B_1, \ldots, B_K$ be any collection of $K$ disjoint measurable subsets of $\Omega$, which we will refer to as \emph{bins}, where $K$ can be $\infty$.
The \emph{histogram} of a distribution $\mathbb D$ corresponding to the bins $B_1, \ldots, B_K$ is given by the vector $(p_1, \ldots, p_K)$ where each $p_k = \underset{X \sim \mathbb D}{\mathbb P}\left( X \in B_k\right)$. In Lemma \ref{thm:binfinding}, we assert the existence of an differentially private mechanism that on input $n$ iid samples from $\mathbb D$ outputs a noisy histogram from which the heaviest bin can be extracted.

\begin{lemma}[Histogram Learner, following \cite{DMNS06}, \cite{bun2016simultaneous}, \cite{vadhan2016complexity}]
\label{thm:binfinding}
For every $K \in \mathbb N \cup \{\infty\}$, domain $\Omega$,  for every collection of disjoint bins $B_1, \ldots, B_K$ defined on $\Omega$, $n \in \mathbb{N}$, $\eps, \delta \in (0, 1/n)$, $\beta > 0$ and $\alpha \in (0,1)$ there exists an $(\eps,\delta)$-differentially private algorithm $M: \Omega^n \rightarrow \mathbb R^K$ such that for every distribution $\mathbb D$ on $\Omega$, if 
\begin{enumerate}
	\item $X_1, \ldots, X_n \overset{iid}{\sim} \mathbb D$,  $p_k = \pr{X_i \in B_k}$
	\item $(\tilde p_1, \ldots, \tilde  p_K) \leftarrow M(X_1, \ldots, X_n)$, and 
    \item \begin{align}
\label{eq:samplecomplexityDL}
	n \geq \max \left\{\min \left\{ \frac{8}{\eps \beta}\log\left(\frac{2K}{\alpha} \right), \frac{8}{\eps \beta}\log\left(\frac{4}{\alpha \delta}\right) \right\}, \frac{1}{2\beta^2}\log \left(\frac{4}{\alpha}\right) \right\}
\end{align} 
\end{enumerate}
then,

\begin{align}
& \underset{\substack{ \underline{X} \sim \mathbb D \\  M}}{\mathbb P} \left(| \tilde p_k - p_k| \leq \beta \right)  \geq 1 - \alpha \text{ and,}\\
& \underset{\substack{\underline{X} \sim \mathbb D \\  M}}{\mathbb P} \left( \arg\max_{k} \tilde p_k = j \right) \leq 
\begin{cases} 
np_j + 2e^{-(\eps n/8) \cdot (\max_k p_k)} & \text{if } K < 2/\delta \\
np_j       & \text{if } K \geq 2/\delta
\end{cases}
\end{align} 
where the probability is taken over the randomness of $M$ and the data $X_1, \ldots, X_n$.
\end{lemma}

Lemma \ref{thm:binfinding} asserts the existence of an $(\eps, \delta)$-differentially private \emph{histogram learner} that takes as input an iid sample and a collection of $K$ disjoint bins, and outputs estimates $\tilde p_k$ of $p_k$, the probability of falling in bin $B_k$ for all $k$. It has the property that with high probability, the maximum generalization error is at most $\beta$. It is important to note that the error is measured with respect to the population $(p_k)$ and not the sample. In particular, the error term includes the noise due to sampling and differential privacy. To bound the generalization error, we follow a standard technique of bounding the generalization error without privacy (i.e. the difference between the sample and the population) and the error introduced for privacy (see for example the equivalence between differentially private query release and differentially private threshold learning in \cite{BNSV}).

If we take $n$ equal to the RHS of inequality \ref{eq:samplecomplexityDL} in Lemma \ref{thm:binfinding}, we obtain a bound on $\beta$ as a function of $n$, $K$, $\alpha$, and $\epsilon$, namely
\begin{align}
	\beta = \max \left\{\min \left\{ \frac{8}{\eps n}\log\left(\frac{2K}{\alpha} \right), \frac{8}{\eps n}\log\left(\frac{4}{\alpha \delta}\right) \right\}, \sqrt{\frac{1}{2n}\log \left(\frac{4}{\alpha}\right)} \right\}
\end{align}
The last term, $\Theta\left(\sqrt{\log(1/\alpha)/n}\right)$ is the sampling error, which is incurred even without privacy. For privacy, we incur an error that is the {\em minimum} of two terms: $\O{\log(K)/(\eps n)}$ and $\O{\log(1/\delta)/(\eps n)}$. Note that these two terms vanish linearly in $n$, faster than the sampling error, which vanishes as $\O{1/\sqrt{n}}$. Moreover, the dependence on the number of bins $K$ is only logarithmic or in case of $\delta > 0$, even non-existent. 
When $\delta > 0$, the choice of $K = \infty$ allows us to construct $(\eps, \delta)$-DP algorithms that have no dependence on the range of the parameters.

We will use the histogram learner to obtain the largest noisy bin from a possibly infinite collection of bins. Hence, apart from a bound on the maximum generalization error, we also need a bound on the probability of picking the wrong bin as the largest bin. Lemma \ref{thm:binfinding} asserts that the probability of choosing any bin $j$ as the largest bin is roughly upper bounded by $np_j$, the expected number of points falling in bin $j$. As before, this probability is over the sampling and the noise added by the differential privacy mechanism. Note that this bound is useful only when $p_j$ is small, in particular $p_j \ll 1/n$. Hence, the theorem bounds the probability of incorrectly choosing a bin that has very few expected points as the largest bin. 
\begin{proof}[Proof of Lemma \ref{thm:binfinding}]
	Let $C_k = \sum_{i=1}^n I(X_i \in B_k)$ be the number of points that fall in bin $k$ and $\hat p_k = C_k/n$ be the corresponding proportion of points. The distribution learner operates as follows: When $K < 2/\delta$, it uses an $(\eps,0)$-DP algorithm, and when $K \geq 2/\delta$, it uses an $(\eps, \delta)$-DP algorithm to output a noisy histogram.
	
	The key idea behind the proof of the existence of a histogram learner is the following. There exist basic $(\eps, 0)$ and $(\eps, \delta)$ differentially private mechanisms $M$ with the property that on input $\underline{X} = (X_1, \ldots, X_n)$ and bins $B_1, \ldots, B_K,$ they output $(\tilde p_1, \ldots, \tilde p_K) = M(\underline{X})$, such that if
	$$ n \geq \min \left\{ \frac{8}{\eps \beta}\log\left(\frac{2K}{\alpha} \right), \frac{8}{\eps \beta}\log\left(\frac{4}{\alpha \delta}\right) \right\}$$
	then for every $\underline{x}$, 
	$$\prm[M][]{\max_k |\tilde p_k - \hat p_k | \geq \beta \mline \underline{X} =\underline{x}} \leq \alpha/2$$
	That is, with high probability, there is a small difference between the differentially private output $\tilde{p}_k$ and the empirical estimates $\hat p_k$. 
	Moreover, the Dvoretzky-Kiefer-Wolfowitz inequality \citep{massart1990tight} tells us that with high probability the empirical estimates are close to the population parameters:
	$$\prm[\underline{X} \sim \mathbb D][]{\max_k |\hat p_k - p_k| > \beta} \leq 2\exp(-2n\beta^2).$$
	So, if $n \geq (1/2\beta^2) \cdot \log\left(4/\alpha\right)$, then 
	$\pr{\max_k |\hat p_k - p_k| > \beta} \leq \alpha/2$.
	Thus, by a union bound, if
	$$n \geq \max \left\{\min \left\{ \frac{8}{\eps \beta}\log\left(\frac{2K}{\alpha} \right), \frac{8}{\eps \beta}\log\left(\frac{4}{\alpha \delta}\right) \right\}, \frac{1}{2\beta^2}\log \left(\frac{4}{\alpha}\right) \right\}$$ 
	then,
	\begin{align*}
	&\prm[\underline{X}\sim \mathbb D][M]{\max_k |\tilde p_k - p_k| \geq \beta} \\ 
	&\leq \underset{\underline{X} \sim \mathbb D}{\mathbb E} \left[\prm[M][]{\max_k |\tilde p_k - \hat p_k| \geq \beta \mline \underline{X} = \underline{x}} \right] + \prm[\underline{X} \sim \mathbb D][]{\max_k|\hat p_k-p_k| \geq \beta} \\
	&\leq \alpha/2 + \alpha/2.
	\end{align*}
	
	Now we review the two differentially private algorithms we use and also prove the additional claim regarding the probability of any bin being selected as the maximum.
	
	 \paragraph{$(\eps, 0)$ case:} We will first start with  the $(\eps, 0)$ differentially private algorithm which we use when $K < 2/\delta$. Consider the Laplace mechanism, given in Lemma \ref{thm:LapMech} applied to the \emph{empirical histogram}. The \emph{empirical histogram} of a dataset $\underline{x} = (x_1, \ldots, x_n)$ on bins $B_1, \ldots, B_K$ is given by $f_{B_1, \ldots, B_K}(\underline{x}) = (\hat p_1, \ldots, \hat p_K)$ 
Note that the global sensitivity of $f_{B_1, \ldots, B_K}$ is $2/n$. Hence we can release the empirical histogram by adding Laplace noise with scale $b = 2/(\eps n)$ to each $\hat p_k$. Formally, let $\tilde p_k = \hat p_k  +Z_k$ where $Z_k \sim Lap(0, 2/(\eps n))$. If $n \geq 2\log (2K /\alpha)/(\eps \beta)$, then,
	  \begin{align*}
	  \prm[M][]{\max_k |\tilde p_k - \hat p_k| \leq \beta \mline \underline{X} = \underline{x}} 
	  &= \pr{\max_k |Z_k| \leq \beta} \\
      &= \left(1-\exp\left(-\eps n \beta/2\right)\right)^K\\
	  &\geq 1 - K\exp\left(-\eps n \beta /2\right)\geq 1 - \alpha/2,
	  \end{align*}
      where the second line follows from the tails of a Laplace distribution, see Proposition \ref{prop:lap}.
	 Now, let us prove that $$\prm[\underline{X} \sim \mathbb D][M]{\arg\max_{k} \tilde p_k = j} \leq np_j + 2e^{-(\eps n/8) \cdot (\max_k p_k)}.$$
	 Let $p_r = \max_{k} p_k$. Consider the event $\{\arg\max_{k} \tilde p_k = j\}$. Note that $\tilde p_k$ is continuous. Hence the probability of ties is $0$ and the argmax is unique and well defined. This implies that $\{\tilde p_j \geq \tilde p_r \}$. In turn, this 
    implies that either, $\hat p_j > 0 $, or $Z_j \geq p_r/4$ or $\hat p_r < p_r/2$ or $Z_r \leq -p_r/4$. 
Hence, by a union bound,
	 \begin{align*}
	 \pr{\arg\max_k \tilde p = j} &\leq \pr{\hat p_j > 0} + \pr{Z_j \geq p_r/4} + \pr{\hat p_r < p_r/2} + \pr{Z_r \leq -p_r/4} \\
	 &\leq 1 - (1-p_j)^n + \frac{1}{2}\exp\left(-\frac{\eps n p_r}{8}\right) +  \exp\left(-\frac{np_r}{8}\right) + \frac{1}{2}\exp\left(-\frac{\eps n p_r}{8}\right)\\
	 &\leq np_j + 2\exp\left(-\frac{\eps n}{8} \cdot \max_k p_k \right),
	 \end{align*}
	 where we have used the Chernoff bound (Proposition \ref{prop:Chernoff}), and a tail bound for a Laplace random variable from (Proposition \ref{prop:lap}).
	 
	 \paragraph{$(\eps, \delta)$ Case:} For the case that $K > 2/\delta$, we will use an $(\eps, \delta)$-differentially private algorithm, called \emph{stability-based} histogram \citep{bun2016simultaneous} to estimate $\hat p_k$, which removes the dependence on $K$ by allowing for $K$ to be infinity. Specifically, the algorithm on input $x_1, \ldots, x_n$ and $B_1, \ldots, B_K$, where $K$ is possibly $\infty$ runs as follows:
	 
	 \begin{enumerate}
	 	\item Let $\hat p_k = \sum_{i=1}^n{I(x_i \in B_k)}/n$
	 	\item If $\hat p_k = 0$, set $\tilde p_k = 0$
	 	\item If $\hat p_k > 0$, 
	 	\begin{enumerate}
	 		\item Let $\tilde p_k = \hat p_k + Z_k$, where $Z_k \sim Lap(0,2/(\eps n))$.
	 		\item Let $t = 2\log(2/\delta)/(\eps n) + (1/n)$ 
	 		\item If $\tilde p_k < t$, set $\tilde p_k = 0$
	 	\end{enumerate}
	 	\item Output $\tilde p_k, k = 1, \ldots, K$
	 \end{enumerate}
A proof of $(\eps, \delta)$ differential privacy of this algorithm can be found in Theorem 3.5 of \cite{vadhan2016complexity}. For utility, we will show that if
$$n \geq \frac{8}{\eps \beta}\log\left(\frac{4}{\delta \alpha}\right) $$
then $\pr{|\tilde p_k - \hat p_k| > \beta \mline \underline{X} = \underline{x}} \leq \alpha/2$.
Note that for any $k$ such that $\hat p_k > 0$, we have,
\begin{align*}
\pr{|\tilde p_k - \hat p_k| > \beta \mline \underline{X} = \underline{x}} 
= & \pr{|\tilde p_k - \hat p_k| > \beta, \hat p_k + Z_k > t \mline \underline{X} = \underline{x}} + \\
{ } & \pr{|\tilde p_k - \hat p_k| > \beta, \hat p_k + Z_k < t \mline \underline{X} = \underline{x}} \\
= & \pr{|Z_k| > \beta, \hat p_k + Z_k > t \mline \underline{X} = \underline{x}} + \pr{\hat p_k  > \beta, \hat p_k + Z_k < t \mline \underline{X} = \underline{x}} \\
\leq & \pr{|Z_k| > \beta} + \pr{\hat p_k > \beta, \hat p_k + Z_k < t \mline \underline{X} = \underline{x}} \\
\leq & \pr{|Z_k| > \beta} + \pr{Z_k < t- \beta} \\
\leq  &\exp\left(-\frac{\eps n \beta}{2}\right) + \exp\left(\frac{-\eps n (\beta -t)}{2}\right) \\
\leq &  \exp\left(-\frac{\eps n \beta}{2}\right) + \exp\left(\frac{-\eps n \beta}{4}\right)  \leq 2\exp\left(-\frac{\eps n \beta}{4}\right) \leq \alpha\delta /2
\end{align*}
where the last line holds if $\beta > 2t$ and if 
$$n \geq \frac{4}{\eps \beta} \log \left(\frac{4}{\alpha \delta}\right)$$
For $\beta > 2t$, we need
$$n \geq \frac{4}{\eps \beta}\log\left(\frac{2}{\delta}\right) + \frac{2}{\beta}.$$
Since $\alpha < 1$, it suffices to have $$ n \geq \frac{8}{\eps \beta}\log\left(\frac{4}{\delta \alpha}\right).$$
There are at most $n$ points such that $\hat p_k > 0$. Also, for the points where $\hat p_k =0$, we have $\tilde p_k = 0$. Hence for such points, $\pr{|\tilde p_k - \hat p_k| > \beta \mline \underline{X} = \underline{x}} = 0$. Thus by a union bound we have,
\begin{align*}
\pr{\max_k |\tilde p_k - \hat p_k| > \beta \mline \underline{X} = \underline{x}} \leq n\alpha \delta/2 < \alpha/2
\end{align*}
since $\delta < 1/n$.
	 
Finally, we need to prove that $\pr{\arg\max_k \tilde p_k = j} \leq np_j$ where the probability is over the randomness of the data and the mechanism. In the stability based algorithm, if $\tilde p_k = 0$ for all $k$, the largest bin is undefined, and we set $\arg\max_k \tilde p_k$ as $\perp$. On the other hand, if $\arg \max_k \tilde p_k = j$, then $\hat p_j > 0$. 
Hence $\pr{\arg\max \tilde p_k = j} \leq \pr{\hat p_j > 0} = 1 - (1-p_j)^n \leq np_j$.	 
	 
\end{proof}

%
%

\section{A differentially private estimate of the range of Gaussian random variables}
\label{sec:range}
In this section, we present differentially private algorithms to estimate the range of an iid sample from a Gaussian distribution with known and unknown variances. These algorithms serve as building blocks for estimating the confidence intervals, but may also be of independent interest.

Let $X_1, \ldots, X_n$ be iid samples from a normal distribution with mean $\mu$ and variance $\sigma^2.$ Our algorithms don't make any assumption on the boundedness of the data. The $(\eps, 0)$ DP algorithms assume that $\mu$ and $\sigma^2$ lie in a bounded domain and the $(\eps, \delta)$ algorithms make no boundedness assumptions. Under these assumptions, the algorithms below output a range with the guarantee that with high probability, the range includes all the data points, and if the sample size is large enough, the estimated range is at most a constant factor larger than the true range. We will begin with the simpler case of estimation of range when the variance is known, followed by the case when the variance is unknown.

 
\subsection{Estimation of range with known variance}

\begin{theorem}
	\label{thm:rangeKnownVariance}
	For every $n \in \mathbb{N}$, $\sigma, \eps, \delta > 0$, $\alpha \in \left(0,1/2\right)$, $R \in (0, \infty]$, there is a $w > 0$ and an $(\eps,\delta)$-differentially private algorithm $M: \mathbb{R}^n \rightarrow \mathbb{R} \times \mathbb{R}$ such that whenever $\mu \in (-R,R)$ and 
	$$
	n \geq   c \min \left\{\frac{1}{\eps}\log \left(\frac{R}{\sigma \alpha}\right), \frac{1}{\eps}\log \left(\frac{1}{\delta \alpha}\right) \right\},
	$$
	(where $c$ is a universal constant), 
	if  $X_1, \ldots, X_n$ are iid Gaussian random variables with mean $\mu$ and variance $\sigma_0^2 \leq \sigma^2$ (where $\sigma^2$ is known) and 
	$$(X_{\min}, X_{\max}) \leftarrow M(X_1, \ldots, X_n),$$ we have:

	$$\prm[\underline{X} \sim N(\mu,\sigma_0^2)][M]{ \forall i \text{ } X_{\min} \leq X_i \leq X_{\max} } \geq 1 - \alpha.$$
	and  with probability $1$,
	$$ |X_{\max} - X_{\min}| = w = \O{\sigma \sqrt{\log (n/\alpha)}}$$
\end{theorem}

Theorem \ref{thm:rangeKnownVariance} asserts the existence of a differentially private algorithm that takes as input $n$ iid samples from a normal distribution and outputs a range $[X_{\min}, X_{\max}]$ with a guarantee that with high probability all points are included in the range. Note that when $\delta =0$, $n$ needs to be at least $c(1/\eps)\log(R/(\sigma \alpha)$. Hence, when $\delta=0$, $R$ needs to be finite i.e., we need $\mu$ to be bounded. On the other hand, when $\delta > 0$, $R$ can be $\infty$. Hence Theorem \ref{thm:rangeKnownVariance} asserts that when $\delta > 0$, there is no dependence on the range of $\mu$. Moreover, the range of the data $|X_{\max} - X_{\min}|$ estimated by the algorithm is within a constant factor of the true range. This is because, it is well known that for Gaussian data, $\mathbb E[\max_i X_i] = \Theta\left( \sigma \sqrt{\log n} \right)$, and in fact, $\max_i X_i \geq \sigma \sqrt{\log(n/\alpha)}$ with probability $\Omega(\alpha)$.
\begin{proof}[Proof of Theorem \ref{thm:rangeKnownVariance}]
Consider the following algorithm.

\begin{algorithm}[H]
	\caption{Differentially Private estimate of range with known variance}
	\begin{algorithmic}[1]
		\label{alg:dataRange}
		\REQUIRE $X_1, \ldots, X_{n}$, $\eps$, $\alpha$, $R \in (0,\infty]$, $\sigma$.
		\ENSURE An $(\eps,\delta)$ differentially private estimate of the range of $X_1, \ldots, X_n$.
		\STATE Let $r = \lceil R/\sigma \rceil$. Divide $[-R-\sigma/2,R+\sigma/2]$ into $2r+1$ bins of length at most $\sigma$ each in the following manner - a bin $B_j$ equals $((j-0.5)\sigma, (j+0.5)\sigma]$, for $j \in \{-r,\ldots, r\}$.
		\STATE Run the histogram learner of Lemma \ref{thm:binfinding} with privacy parameters $(\eps, \delta)$ and bins $B_{-r}, \ldots, B_{r}$ on input $X_1, \ldots, X_n$ to obtain noisy estimates $\tilde p_{-r}, \ldots, \tilde p_{r}$. Let the largest noisy bin be $\hat l$, i.e.
		$$\hat l = \underset{j = -r, \ldots, r}{\mbox{argmax }} \tilde p_j$$
		\STATE Output $(X_{\min}, X_{\max})$, where
		$$
		X_{\min} = \sigma\hat{l}  - 4\sigma \sqrt{\log \left(\frac{n}{\alpha}\right)}, 
		X_{\max} = \sigma\hat{l} + 4\sigma \sqrt{\log \left(\frac{n}{\alpha}\right)}
		$$
	\end{algorithmic}
\end{algorithm}
Let $\sigma^* \leq \sigma$ be the standard deviation of $X_1$. We will prove the following two claims:
	
\noindent \emph{Claim 1:} If 	
$$n \geq   c \min \left\{\frac{1}{\eps}\log \left(\frac{R}{\sigma \alpha}\right), \frac{1}{\eps}\log \left(\frac{1}{\delta \alpha}\right) \right\},$$ 
then with probability at least $1 - \alpha/2$, we have 
	$$|\mu - \hat l \sigma| \leq 2\sigma.$$
\emph{Claim 2:} With probability at least $1 - \alpha/2$, we have  
	$$\left\{\forall i: |X_i - \mu | \leq \sigma^* \sqrt{2\log \left(\frac{4n}{\alpha}\right)} \right\}.$$
By Claims 1, 2, a union bound and the fact that $\sigma^* \leq \sigma$, we have, with probability at least $1 - \alpha$, for all $i$,
	$$|X_i - \hat l \sigma| \leq |X_i - \mu| + |\hat l \sigma - \mu| \leq \sigma^* \sqrt{2\log \left(\frac{4n}{\alpha}\right)} + 2\sigma \leq 4\sigma \sqrt{\log \left(\frac{n}{\alpha}\right)}$$
which gives the result. We will now prove the two claims.

\paragraph{Proof of Claim 1}
Let $B_j$ be the $j^{th}$ bin and $p_j = \pr{X_i \in B_j}$. Let $p_{(1)} \geq p_{(2)} \geq p_{(3)}\geq \ldots$ be sorted $p_j$'s and let $j_{(1)}, j_{(2)}, \ldots$ be the corresponding bins, the tie breaking rule will be described below.
Consider the event $E = \{\hat l = j_{(1)} \mbox{ or } j_{(2)}\}$. 
We will first show that bins with largest and second largest mass are adjacent to each other, and $\mu$ always lies in a bin with the largest mass. This will imply that under the event $E$, $|\hat l - \mu| \leq 2\sigma$. 

We begin by finding the bins with the largest, second largest and third largest mass. Let $\Phi(\cdot)$ denote the cdf of a standard normal distribution. Note that
	\begin{align*}
	p_j &= \pr{X_i \in ((j-0.5)\sigma, (j+0.5)\sigma]} \\
    &= \Phi\left( \frac{(j+0.5)\sigma - \mu}{\sigma^*} \right) - 
    \Phi\left(\frac{(j-0.5)\sigma - \mu}{\sigma^*} \right) \\
	&= f\left(j\frac{\sigma}{\sigma^*} - \frac{\mu}{\sigma^*}\right),
	\end{align*} 
where $f(\gamma) = \Phi\left(\gamma + \frac{\sigma}{2\sigma^*}\right) - \Phi\left(\gamma - \frac{\sigma}{2\sigma^*}\right)$.	
One can verify that $f(\gamma)$ is maximized at $\gamma =0$, is symmetric and decreasing with $|\gamma|$. 
Hence maximizing $p_j$ on integers amounts to minimizing $|\gamma| = |j(\sigma/\sigma^*) - \mu/\sigma^*|$, or equivalently minimizing $|j- \mu/\sigma|$. So we are sorting the integers $j$ according to their distance from $\mu/\sigma$. Thus we can take $j_{(1)} = \lceil\mu/\sigma \rfloor$ where $\lceil x \rfloor$ is $x$ rounded to the nearest integer, where we break ties by rounding down. 

Notice that $\mu \in ( (j_{(1)} - 0.5)\sigma, (j_{(1)} + 0.5)\sigma ] = B_{j_{(1)}}$. To determine $j_{(2)}$ and $j_{(3)}$, we consider the cases whether $j_{(1)} = \lceil\mu/ \sigma \rceil$ or $j_{(1)} = \lfloor\mu/ \sigma \rfloor$.
\begin{enumerate}
\item If $j_{(1)} = \lceil\mu/ \sigma \rceil$, then we can take $j_{(2)} = \lceil\mu/ \sigma \rceil - 1$ and $j_{(3)} = \lceil\mu/ \sigma \rceil +1$.
\item If $j_{(1)} = \lfloor\mu/ \sigma \rfloor$, then we can take $j_{(2)} = \lfloor\mu/ \sigma \rfloor + 1$ and $j_{(3)} = \lfloor\mu/ \sigma \rfloor -1$.
\end{enumerate}

Notice that in either case, $|j_{(1)} - j_{(2)}| = 1$, which implies that the largest and the second largest bins are always adjacent to each other, and $|j_{(3)} - \mu/\sigma | = |j_{(1)} - \mu/\sigma| + 1$. Hence we have,
\begin{align*}
g = p_{(1)} - p_{(3)} 
&= f\left(\left| j_{(1)} - \frac{\mu}{\sigma}\right| \frac{\sigma}{\sigma^*}  \right) - 
f\left( \left| j_{(3)} - \frac{\mu}{\sigma}\right| \frac{\sigma}{\sigma^*}   \right)  \\
&\geq f\left( (1/2)\cdot (\sigma/\sigma^*) \right) - f\left( (3/2) \cdot (\sigma/\sigma^*)\right) \\
\intertext{(Since $f(x) - f(x + \sigma/\sigma^*)$ is decreasing for $x \in [0,\infty)$ and $|j_{(1)}- \mu/\sigma| \leq 1/2$),}
&= \Phi(\sigma/\sigma^*) - \Phi(0) - \Phi(2\sigma/\sigma^*) + \Phi(\sigma/\sigma^*) \\
&\geq 2\Phi(1) - \Phi(0) - \Phi(2) \\
\intertext{(Since $2\Phi(x) - \Phi(0) - \Phi(2x)$ is a decreasing function for $x \in (0,\infty)$ and $\sigma/\sigma^* \leq 1$}
&\geq 0.1 \mbox{ (By explicit calculation.)}
\end{align*}
Hence under the event $E = \{\hat l = j_{(1)} \mbox{ or } j_{(2)}\}$, $|\hat l - \mu| \leq 2\sigma$.



Next we will lower bound the probability of event $E$. If $\max_k |\tilde p_k - p_k| < g/2$, then the largest noisy bin $\hat l$ is either $j_{(1)}$ or $j_{(2)}$.  Since $\mu \in B_{(j_{(1)})}$, and the bin width is $\sigma$, we have $|\hat l\sigma - \mu| \leq 2\sigma$. Applying Lemma \ref{thm:binfinding}, with $\beta = g$, we get, if 
$$ n \geq \O{\min \left\{ 
	\frac{\log\left( \frac{K}{\alpha}\right)}{g\eps}, 
	\frac{\log\left(\frac{1}{\alpha \delta}\right)}{g\eps }
	\right\}}$$
then with probability at least $1 - \alpha/2$, we have $\max_k |\tilde p_k - p_k| < g/2$. Here $K = 2R/\sigma$ is the number of bins.

\paragraph{Proof of Claim 2:}
From Proposition \ref{prop:lap} due to Gaussian tails we have that for all $i$,
		 $$\pr{|X_i-\mu| > c} \leq 2e^{-c^2/2\sigma^{*^2}}$$ 
		 By applying a union bound, we get,
		 \begin{align*}
         \pr{\exists \mbox{ i} |X_i -\mu| \geq c} \leq 2ne^{-c^2/2\sigma^{*^2}}
		 \end{align*}
		 Letting $c = \sigma^* \sqrt{2\log \left(4n/\alpha\right)}$ gives the desired result.
\end{proof}

\subsection{An estimate of variance}
\label{sec:varianceest}
In this section, we will construct a differentially private estimator of the variance of data from a Gaussian distribution with unknown variance. This estimate will be used for finding the range of a sample with unknown variance.

\begin{theorem}
	\label{lem:varLabel}
	For every $n \in \mathbb{N}$, $\sigma_{\min} < \sigma_{\max} \in [0, \infty], \eps, \delta > 0$, $\alpha \in \left(0,1/2\right)$, there is an $(\eps,\delta)$-differentially private algorithm $M: \mathbb{R}^n \rightarrow [0, \infty)$ such that
	if  $X_1, \ldots, X_n$ are iid Gaussian random variables with mean $\mu$ and with variance $\sigma^2 \in (\sigma_{\min}^2, \sigma_{\max}^2)$, and 
	$$\hat \sigma \leftarrow M(X_1, \ldots, X_n),$$ we have:
	\begin{enumerate}
		\item High probability bound: If 
		$$
		n \geq   c \min \left\{\frac{1}{\eps}
		\log \left(
		\frac{\log \left(\frac{\sigma_{\max}}{\sigma_{\min}}\right)}{\alpha}
		\right), \frac{1}{\eps}\log \left(\frac{1}{\delta \alpha}\right) \right\},
		$$
		(where $c$ is a universal constant), 
		$$\prm[\underline{X} \sim N(\mu,\sigma^2)][M]{ \sigma \leq \hat \sigma \leq 8 \sigma } \geq 1 - \alpha$$
		\item Expectation bound: If 
        $$	n \geq   c \min \left\{\frac{1}{\eps}
				\log\left(\frac{\sigma_{\max}}{\sigma_{\min}}\right), \frac{1}{\eps}\log \left(\frac{1}{\delta \alpha}\right) \right\},$$
				then 
           $$
                \Em[\underline{X} \sim N(\mu,\sigma^2)][M]{\hat \sigma^2} \leq \sigma^2 \cdot \left(c_1 + c_2 \log^2(n) \cdot \alpha\right)  
             $$ 
    for some universal constants $c_1$ and $c_2$.
	\end{enumerate}
\end{theorem}
As before, note that when $\delta >0$, $\sigma_{\max}$ can be set to $\infty$ and $\sigma_{\min}$ can be set to $0$, and the required sample complexity remains finite. However, when $\delta =0$, we need $\sigma_{\min}$ and $\sigma_{\max}$ to be bounded away from $0$ and $\infty$ respectively. Theorem \ref{lem:varLabel} asserts that with high probability the estimate of $\sigma$ is not too far from the true $\sigma$. Moreover, we can also obtain a bound on the expected value of $\hat \sigma^2$ with an increase in sample complexity (to a logarithmic dependence on $\sigma_{\min}/\sigma_{\max}$ rather than doubly-logarithmic). For instance, if we set $\alpha = \alpha_0/\log^2 n$, we get $\mathbb E[\hat \sigma^2] = \O{\sigma^2}$. Reduction of $\log^2n$ in $\alpha$ amounts to an extra requirement that $n \geq c \cdot \log^2 n /\eps$, i.e., $n \geq \Otilde{1/\eps}$. The upper bound on $\mathbb E[\hat \sigma^2]$ will be used to derive the upper bound on the expected length of the width of the confidence interval in Theorem \ref{thm:ciunknownvar}.   

\begin{proof}[Proof of Theorem \ref{lem:varLabel}]
We claim that Algorithm \ref{alg:estvar} given below is $(\eps, \delta)$ differentially private with the required properties.
\begin{algorithm}[H]
	\caption{}
	\begin{algorithmic}[1]
		\label{alg:estvar}
		\REQUIRE $X_1, \ldots, X_{n}$, $\eps$, $R$, $\sigma_{\min}$, $\sigma_{\max}, \alpha$.
		\ENSURE An $(\eps, \delta)$ differentially private approximate estimate $\hat \sigma$ of the standard deviation of $X_1, \ldots, X_n$.
		
		\STATE If 	$$n <   c \min \left\{\frac{1}{\eps}
		\log \left(
		\frac{\log \left(\frac{\sigma_{\max}}{\sigma_{\min}}\right)}{\alpha}
		\right), \frac{1}{\eps}\log \left(\frac{1}{\delta \alpha}\right) \right\},$$
		output $\perp.$
		\STATE Divide the positive half of the real line into bins of exponentially increasing length. The bins are of the form $B_j = (2^{j},2^{j+1}]$ for 
		$j = j_{\min}, \ldots,j_{\max}$.
		where $j_{\max} = \lceil \log_2 \sigma_{\max} \rceil + 1$ and $j_{\min} = \lfloor \log_2 \sigma_{\min} \rfloor - 2$. 
		  
	
		\STATE Let $Y_i = X_{2i} - X_{2i-1}$ for $i = 1, \ldots, \lfloor n/2 \rfloor$ 
		
		\STATE Run the histogram learner of Lemma \ref{thm:binfinding} with budget $(\eps, \delta)$ and bins $B_{j_{\min}}, \ldots, B_{j_{\max}}$ on input $|Y_1|, \ldots, |Y_{\lfloor n/2 \rfloor}|$ to obtain noisy estimates $\tilde p_{j_{\min}}, \ldots, \tilde p_{j_{\max}}$. Let the largest noisy bin be $\hat l$, i.e.
			$$\hat l = \underset{j}{\mbox{argmax }} \tilde p_j$$
			
		
		\STATE Output $\hat \sigma = 2^{\hat l+2}$
	\end{algorithmic}
\end{algorithm}
\paragraph{Proof of Part (1):}
Let us start by proving Part (1). The proof will be based on the following two informal steps:
	
\noindent \emph{Step 1}: The bins with largest or second largest probability mass are always the bin containing $\sigma$ or the bin next to the bin containing $\sigma$.

\noindent \emph{Step 2:} Under the right conditions, the largest noisy bin $\hat l$ is the either the bin with largest or the second largest mass.

Let us start with some simple facts. Since $\sigma \in  (\sigma_{\min},\sigma_{\max})$,  there exists a bin $B_l$ with label $l \in (\lfloor \log_2 \sigma_{\min} \rfloor - 1, \lceil \log_2 \sigma_{\max} \rceil )$ such that $\sigma \in (2^{l},2^{(l+1)} ] = B_l$. Hence we can write $\sigma = 2^{l+c}, \mbox{ for some } c \in (0,1]$.
	Next, note that $Y_i \sim N(0,2\sigma^2)$ for $i = 1, \ldots, \lfloor n/2 \rfloor$.
	Following the statement of Lemma \ref{thm:binfinding}, define
	$$p_j = \pr{|Y_i| \in B_j }$$
  Sort the $p_j$'s as $p_{(1)}\geq p_{(2)} \geq p_{(3)} \geq \ldots$ and let  $j_{(1)}, j_{(2)}, j_{(3)}, \ldots$ be the corresponding bins, where the tie breaking rules will be specified below.
  
	\paragraph{Claim 1:} The bins corresponding to the largest and second largest mass $p_{(1)}, p_{(2)}$ are $(j_{(1)}, j_{(2)}) \in  \{(l,l-1), (l,l+1), (l+1,l)\}$.
	\paragraph{Claim 2:} $p_{(1)} - p_{(3)} > \frac{1}{100}$.
	
	Deferring the proof of these claims, let us prove the Part (1). Consider the following event 
	$$E = \{ \hat l \in (j_{(1)}, j_{(2)})\}.$$ 
	Under $E$, from Claim 1, it follows that $\hat \sigma = 2^{\hat l +2}$ will be $2^{l+1}, 2^{l+2}, $ or $2^{l+3}$. Since $\sigma =2^{l+c}$, for some $c \in (0,1]$, we have, $\sigma \leq \hat \sigma \leq 8\sigma.$ 
	To finish the proof of Part (1), we need to lower bound the probability of $E$. This follows from the properties of the histogram learner of Lemma \ref{thm:binfinding}. Let $g = p_{(1)} - p_{(3)}$. If  $\max_i |\tilde p_i - p_i | < g/2$, then the largest noisy bin is either $p_{(1)}$ or $p_{(2)}$. From Lemma \ref{thm:binfinding} this happens with probability at least $1 - \alpha$ if
	$$ n \geq \O{\min \left\{ 
		\frac{\log\left( \frac{K}{\alpha}\right)}{g\eps}, 
		\frac{\log\left(\frac{1}{\alpha \delta}\right)}{g    \eps }
		\right\}}
	$$
	where $K = c \cdot \log_2 \left(\sigma_{\max} / \sigma_{\min}\right)$ is the number of bins. 
	We will now prove the two claims.
	For any $\Delta \in \mathbb R$, define
		$$q(\Delta) = \phi\left(2^{\Delta+0.5}\right) - \phi\left(2^{\Delta - 0.5}\right),$$
		and note that,
		\begin{align*}
		p_{j}
		&=2\left(\Phi\left( \frac{2^{j+1}}{\sqrt{2}\sigma} \right) - 
		\Phi\left( \frac{2^{j}}{\sqrt{2}\sigma} \right)\right)\\
		&= 2q\left(j - l -c\right)
		\end{align*} 
		where $\Phi(\cdot)$ is the cdf of a standard normal distribution.
		
	\begin{proof}[Proof of Claim 1]
		One can verify that when $\Delta \in \mathbb R$, 
		\begin{align}
		\label{eq:derivativeOfScale}
		\frac{d q(\Delta)}{d \Delta} = g(\Delta) \left(1 -\frac{1}{2}\exp\left(3\cdot 2^{2(\Delta -1)}\right)\right)
		\end{align}
		where $g(\Delta) > 0$  $\forall \Delta$. Let $d =  1+ 0.5\log_2((\log 2)/3)$. The derivative is $0$ when $\Delta = d$.  $q(\Delta)$ is increasing when $\Delta < d$ and decreasing when $\Delta > d$, and the maximum occurs at $\Delta = d$. 
		Since $p_j = 2q(j-l-c)$, and $j \in \{j_{\min}, \ldots, j_{\max}\}$, the largest value of $p_{j}$ occurs at $j  = l+ \lceil c + d\rceil$ or $l+ \lceil c + d \rceil -1$. Since $c \in (0,1]$ and $d \approx -0.056$, $\lceil c+ d \rceil \in \{0,1\}$ and the possible values of $j$ where the maximum occurs are either $l$ or $l-1$ or $l+1$. One can verify that for all $c \in (0,1]$, $p_l > p_{l-1}$ (For example, note that $p_l - p_{l-1} = 2q(-c) - q(-1-c)$ and $q(-c) - q(-1-c)$ is a monotonic function for $c \in (0,1]$ and positive for $c=0$ and $c=1$). Hence the maximum occurs at either $l$ or $l+1$. 
		
		Now let us find the second largest bin. Note that due to the concavity of $q(\Delta)$, the second largest bin must be adjacent to the largest bin. When the largest bin is $l$, the second largest bin can be either $l-1$ or $l+1$. Similarly, when the largest bin is $l+1$, the second largest bin can be either $l$ or $l+2$. One can verify that for all $c \in (0,1]$, $p_{l+2} < p_l$. Hence the largest and the second largest bins $(j_{(1)}, j_{(2)})$ are always the pairs $(l,l-1), (l,l+1), (l+1,l)$ as desired.

		
	\end{proof}


\begin{proof}[Proof of Claim 2]
	 Note that
	$$p_{(1)} \geq p_l = 2q(-c) = 2\left[\Phi(2^{0.5-c}) - \Phi(2^{-0.5-c})\right]$$
	Also for any $k_1, k_2$, we have $p_{(3)} \leq \max_{k \neq k_1, k_2} p_k$.
	When $c \in (0,1/2]$, we will use $k_1 = l, k_2 = l-1$ and when $c \in (1/2,1]$ we will use $k_1 = l, k_2 = l+1$. This gives us the following:
	 \paragraph{Case 1:} When $c \in (0,1/2]$,
	 $$p_{(3)} \leq \underset{t \in \mathbb Z \backslash \{0,-1\}}{\max} 2q(t-c) 
	 $$
	 By inspecting the derivative  (equation \ref{eq:derivativeOfScale}), one can verify that the maximum occurs when $t = 1$.
	 Hence we have,
	 \begin{align*}
	 p_{(1)} - p_{(3)}
	 &\geq \underset{c \in (0,1/2]}{\min}2\left[\phi(2^{1/2-c}) - \phi(2^{-1/2-c})  - \phi(2^{3/2-c}) + \phi(2^{1/2-c})\right]
	 \end{align*}
	 By taking the derivative one can verify that this function is decreasing in $c$ and the minimum occurs at $c = 1/2$ and $g \geq 0.0139$.
	 \paragraph{Case 2:}When $c \in (1/2,1]$,
	 $$p_{(3)} \leq \underset{t \in \mathbb Z \backslash \{0,1\}}{\max}2q(t-c) 
	 $$
	 By inspecting the derivative (equation \ref{eq:derivativeOfScale}), one can verify that the maximum occurs when $t = 2$.
	 Hence we have,
	 \begin{align*}
	 g &= p_{(1)} - p_{(3)} \\
	 &\geq \underset{c \in (1/2,1]}{\min}2\left[\phi(2^{1/2-c}) - \phi(2^{-1/2-c})  - \phi(2^{5/2-c}) + \phi(2^{3/2-c})\right]
	 \end{align*}
	 By taking the derivative one can verify that this function is decreasing in $c$ and the minimum occurs at $c = 1$ and $g \geq 0.091$.
	 In both cases, $g \geq 0.01$
\end{proof}

\paragraph{Proof of Part (2):} Next we need to show that when 
\begin{align}
\label{assup:samplecomplexity}
n \geq (c/ \eps ) \cdot \min \{\log \left(\sigma_{\max}/\sigma_{\min}\right) , \log(1/\alpha \delta)\},
\end{align} we have,  
$\E{\hat \sigma^2} \leq \sigma^2 \cdot \left(c_1 + c_2 \log^2(n) \cdot \alpha \right)$. 
Note that under the Assumption \ref{assup:samplecomplexity} on $n$, the results of Part (1) apply, since this sample complexity is larger than what is needed for the high probability bound. Let $q_i = \pr{\arg\max_j \tilde p_j =i}$. 
Recall that $\sigma = 2^{l+c}$ for some $c \in (0,1]$. 

\begin{align*}
\E{\hat \sigma^2} &= \sum_{i=j_{\min}}^{i=j_{\max}} (2^{i+2})^2 \cdot \pr{\arg\max_j \tilde p_j = i }\\
&= 16 \left(
\sum_{i \leq l+3} 2^{2i} q_i +  
\sum_{l+3 < i \leq l + \log_2 (\log n)} 2^{2i} q_i +
\sum_{i > l + \log_2(\log n)} 2^{2i} q_i 
\right)\\
&= 16 \left( A + B + C \right)
\end{align*}

Now we will bound each of these three terms:
\begin{align*}
A = \sum_{i \leq l+3} 2^{2i} q_i \leq 2^{2(l+3)} \cdot \left(\sum_{i\leq l+3} q_i\right) \leq 64\sigma^2 \cdot 1
\end{align*}

\begin{align*}
B&= \sum_{l+3 < i \leq l+\log \log n} 2^{2i} q_i \leq 2^{2(l+\log \log n)} \cdot \left(\sum_{l+3 < i < l+ \log \log n} q_i\right) \\
&\leq \sigma^2 \log^2 n \cdot \pr{\hat \sigma > 8\sigma } \leq \sigma^2 \log^2 n \cdot \alpha
\end{align*}
where we have used Part (1) to bound $\pr{\hat \sigma > 8\sigma}$ by $\alpha$.
To bound the last summand $C$ we will consider two cases. First note that if $X \sim N(0,\sigma^2)$, we have.
$$
p_i = \pr{|X| \in (2^i, 2^{i+1}]} < \pr{ |X| > 2^i} \leq 2\exp(-2^{2i}/2\sigma^2) \leq 2\exp(-2^{2(i-l-1)}/2).
$$
Let $K = j_{\max} - j_{\min}$ be the number of bins.

\paragraph{Case 1: $\delta \geq 2/K$}
From Lemma \ref{thm:binfinding}, when $\delta \geq  2/K$, 
$$q_i \leq np_i \leq 2n\exp(-2^{2(i-l)-3}).$$ 
Hence
\begin{align*}
C &= \sum_{i > l + \log \log n} 2^{2i} q_i < 2n \cdot \left(\sum_{i > l + \log \log n} 2^{2i} \exp(-2^{2(i-l)-3}) \right) \\
&\leq 2n\sigma^2 \left(\sum_{t > \log \log n} 2^{2t} \exp(-2^{2t-3})\right) \\
&\leq 2n \sigma^2 \int_{\log n}^{\infty} k^2 \exp(-k^2/8)dk \\
&\leq \sigma^2 \mbox{, for sufficiently large $n$,}
\end{align*}
and $c_3$ is a fixed positive constant.
Combing the three terms, we get:
\begin{align*}
\E{\hat \sigma^2}
&\leq \sigma^2 \left(c_1 + c_2 \log^2(n) \alpha + c_3\right)
\end{align*}
\paragraph{Case 2: $\delta < 2/K$}

From Lemma \ref{thm:binfinding}, when $\delta < 2/K$, 
$$q_i \leq np_i +\exp\left(-\Om{-\eps n \max_j p_j}\right) \leq 2n \exp\left(-2^{2(i-l)-3}\right) + \exp\left(-\Om{-\eps n}\right),$$
since $\max_j p_j = p_{(1)} > 1/100$ from the proof of Part (1), Claim 2.
We get:
\begin{align*}
C&= \sum_{i > l + \log \log n} 2^{2i} q_i \\
&\leq \sum_{i > l+ \log\log n} 2^{2i} \exp(-2^{2(i-l)-3}) + \sum_{l + \log \log n < i \leq j_{\max} } 2^{2i} \exp\left(-\Om{-\eps n}\right)\\
&\leq  \sigma^2 +  \sigma_{\max}^2 \exp\left(-\Om{-\eps n}\right), 
\end{align*}
for all sufficiently large $n$.
Hence,
\begin{align*}
\E{\hat \sigma^2}
&\leq \sigma^2 \left(c_1 + c_2 \log^2(n) \alpha + c_3 + e^{-\Om{\eps n}}\cdot \frac{\sigma_{\max}^2}{\sigma_{\min}^2} \right)\\
&\leq \sigma^2 \left(c_1 + c_2 \log^2(n) \cdot \alpha \right)
\end{align*}
if 
$$n \geq \frac{c}{\eps} \log \left(\frac{\sigma_{\max}}{\sigma_{\min}}\right)$$
Hence in both cases, we have:
\begin{align*}
\mathbb E[\hat \sigma^2]
&\leq \sigma^2 \left(c_1 + c_2 \log^2(n) \cdot \alpha\right)
\end{align*}

\end{proof}

\subsection{Estimation of range with unknown variance}
In this section we will present an algorithm that estimates the range of data from a normal distribution when both the mean and the variance are unknown. 
\begin{theorem}
	\label{thm:rangeunkownvar}
	For every $n \in \mathbb{N}$, $\sigma_{\max} > \sigma_{\min} \in [0,\infty], \eps, \delta > 0$, $\alpha \in \left(0,1/2\right)$, $R \in [0,\infty]$, there is an $(\eps,\delta)$-differentially private algorithm $M: \mathbb{R}^n \rightarrow (0,\infty) \times \mathbb{R} \times \mathbb{R}$ such that whenever  
	$$n \geq  c \cdot
	\min \left\{
	\max \left\{ 
	\frac{1}{\eps}\log \left(\frac{R}{\sigma_{\min} \alpha}\right), \frac{1}{\eps}\log \left(\frac{\log (\sigma_{\max}/\sigma_{\min})}{\alpha}\right)
	\right\},
	\frac{1}{\eps}\log \left(\frac{1}{\delta \alpha}\right)
	\right\},
	$$
	(where $c$ is a universal constant), 
	if  $X_1, \ldots, X_n$ are iid Gaussian random variables with mean $\mu \in (-R,R)$ and with variance $\sigma^2 \in (\sigma_{\min}^2, \sigma_{\max}^2)$, and 
	$$(\hat \sigma, X_{\min}, X_{\max}) \leftarrow M(X_1, \ldots, X_n),$$ we have:
	\begin{enumerate}
		\item $\prm[\underline{X} \sim N(\mu,\sigma^2)][M]{ \forall i \text{ } X_{\min} \leq X_i \leq X_{\max} } \geq 1 - \alpha$
		\item $\prm[\underline{X} \sim N(\mu,\sigma^2)][M]{|X_{\max} - X_{\min}| \leq \O{\sigma \sqrt{\log (n/\alpha)}}} \geq 1  - \alpha$
		\item $\prm[\underline{X} \sim N(\mu,\sigma^2)][M]{\sigma \leq \hat \sigma \leq 8 \sigma} \geq 1 - \alpha$
	\end{enumerate}
Moreover, if 
	$$n \geq  c \cdot
	\min \left\{
					\frac{1}{\eps} \log \left(\frac{\sigma_{\max}}{\sigma_{\min}}\right)
				,
				\frac{1}{\eps}\log \left(\frac{1}{\delta \alpha}\right)
		\right\},
	$$
	then
	$$\Em[\underline{X} \sim N(\mu,\sigma^2)][M]{\hat \sigma^2} \leq \sigma^2\left(c_1 + c_2 \log^2(n) \cdot \alpha\right)$$
	
\end{theorem}

We briefly comment on Theorem \ref{thm:rangeunkownvar}. As before, when $\delta > 0$, $\mu$ and $\sigma^2$ can remain unbounded as there is no dependence of $n$ on $R$, $\sigma_{\min}$ and $\sigma_{\max}$. On the other hand, when $\delta =0$, $\mu$ and $\sigma^2$ must be bounded, and the dependence of $n$ is logarithmic on $R/\sigma_{\min}$ and $\sigma_{\max}/\sigma_{\min}$. Hence even when $\delta =0$, these parameters can be set to a large value as the dependence is only logarithmic. Moreover, with high probability, the estimated range is of the same order as the expected range for Gaussian data, which is $\Theta\left(\sigma \sqrt{\log n}\right)$

\begin{proof}[Proof of Theorem \ref{thm:rangeunkownvar}]
We first obtain an algorithm that estimates the range by combining Algorithms \ref{alg:estvar} and \ref{alg:dataRange}. The idea is to first find an good estimate $\hat \sigma$ of $\sigma$ using Algorithm \ref{alg:estvar}. By Theorem \ref{lem:varLabel}, $\hat \sigma$ is an upper bound on $\sigma$. We then run Algorithm \ref{alg:dataRange} which finds the range of the data when an upper bound on the variance is known.

\begin{algorithm}[H]
	\caption{Differentially private estimate of range with unknown variance}
	\begin{algorithmic}[1]
		\label{alg:dataRangeunkownVariance}
		\REQUIRE $x_1, \ldots, x_{n}$, $\eps$, $\alpha$, $R$, $\sigma_{\min}$, $\sigma_{\max}$.
		\ENSURE An $(\eps, \delta)$ differentially private estimate of the range of the data $x_1, \ldots, x_n$.
		\STATE If $$n <  c \cdot
		\min \left\{
		\max \left( 
		\frac{1}{\eps}\log \left(\frac{R}{\sigma_{\min} \alpha}\right), \frac{1}{\eps}\log \left(\frac{\log (\sigma_{\max}/\sigma_{\min})}{\alpha}\right)
		\right),
		\frac{1}{\eps}\log \left(\frac{1}{\delta \alpha}\right)
		\right\},
		$$
		Output $\perp$.
		\STATE Run Algorithm \ref{alg:estvar} of Lemma \ref{thm:binfinding} with $\eps_1 = \eps/2, \delta_1 = \delta/2$, $\alpha_1 = \alpha/2$, $\sigma_{\min}$ and $\sigma_{\max}$ on the data $x_1, \ldots, x_{n}$ to obtain an estimate of scale $\hat \sigma$.
		\STATE Run Algorithm \ref{alg:dataRange} with $\eps_2 = \eps/2$, $\delta_2 = \delta/2$, $\alpha_2 = \alpha/2$, $\sigma = \hat \sigma$ and $R$ to get $[X_{\min}, X_{\max}]$.
		\STATE Output $\hat \sigma$ and $[X_{\min}, X_{\max}]$.
	\end{algorithmic}
\end{algorithm}
By the composition property of differential privacy given in Lemma \ref{thm:compose}, and the fact that Algorithms \ref{alg:dataRange} and \ref{alg:estvar} are differentially private (from Theorems \ref{thm:rangeKnownVariance} and \ref{lem:varLabel}), it follows that Algorithm \ref{alg:dataRangeunkownVariance} is $(\eps, \delta)$ differentially private. 

Consider the event $E = \{\sigma < \hat \sigma <  8\sigma\} $. Under $E$, the results from Theorem \ref{thm:rangeKnownVariance} for the case of range finding when an upper bound on the variance is known can be used. Hence under the event $E$, we have if,
$$
n \geq   c \min \left\{\frac{1}{\eps}\log \left(\frac{R}{\sigma_{\min} \alpha}\right), \frac{1}{\eps}\log \left(\frac{1}{\delta \alpha}\right) \right\} 
\geq c \min \left\{\frac{1}{\eps}\log \left(\frac{R}{\hat \sigma \alpha}\right), \frac{1}{\eps}\log \left(\frac{1}{\delta \alpha}\right) \right\},
$$
(1) With probability at least $1 - \alpha/2$,
$$ \forall i, X_{\min} \leq X_i \leq X_{\max}, \mbox{ and} $$
(2) With probability 1 we have $|X_{\max} - X_{\min}| =\O{ \hat \sigma \sqrt{\log (n/\alpha)}} = \O{\sigma \sqrt{\log (n/\alpha)}}$
By the properties of the variance estimation algorithm given in Theorem \ref{lem:varLabel}, if 
	$$
	n \geq   c \min \left\{\frac{1}{\eps}
	\log \left(
	\frac{\log_2 \left(\frac{\sigma_{\max}}{\sigma_{\min}}\right)}{\alpha}
	\right), \frac{1}{\eps}\log \left(\frac{1}{\delta \alpha}\right) \right\},
	$$
	we have with probability at least $1 - \alpha/2$,
	$$\sigma \leq \hat \sigma \leq 8 \sigma.$$
	The results (1), (2), and (3) follow from a union bound.
	Moreover, by Theorem \ref{lem:varLabel}, if $$	
	n \geq   c \min \left\{\frac{1}{\eps}
	\log\left(\frac{\sigma_{\max}}{\sigma_{\min}}\right), \frac{1}{\eps}\log \left(\frac{1}{\delta \alpha}\right) \right\},
	$$
	then $\E{\hat \sigma^2} \leq \sigma^2 \cdot \left(c_1 + c_2 \log^2(n) \cdot \alpha\right)  $ for some universal constants $c_1$ and $c_2$.

\end{proof}
\section{Estimating a confidence interval of a mean with known variance}
\label{sec:ciknownvar}
In this section, we present algorithms to estimate a differentially private confidence interval of a mean of Gaussian distribution when the variance is known. When the variance is fixed and known, the non-private interval is $\bar{X} \pm (\sigma/\sqrt{n})z_{1 - \alpha/2}$ where $\bar{X}$ is the sample mean, $\sigma$ is the known standard deviation and $z_{1 - \alpha/2}$ is the $1 - \alpha/2$ quantile of a standard normal distribution. (See section \ref{sec:defs} for Definitions). Note that the interval is of fixed width namely $2\sigma/\sqrt{n} z_{1-\alpha/2} = \Theta \left((\sigma / \sqrt{n}) \cdot \sqrt{\log (n/\alpha)}\right)$. We will show that one can also obtain a fixed width $(\eps, \delta)$-DP confidence interval when the variance is known. 
\begin{theorem}
	\label{thm:ciknownvar}
	Let $\mathbb{I}$ be the set of all possible intervals in $\mathbb{R}$. For every $n \in \mathbb{N}$, $\sigma, \eps, \delta > 0$, $\alpha \in \left(0,1/2\right)$, $R \in (0, \infty]$, there is $\beta >0$ and an $(\eps,\delta)$-differentially private algorithm $M: \mathbb{R}^n \rightarrow \mathbb{I}$ such that 
	if $X_1, \ldots, X_n$ are iid random variables from $N(\mu,\sigma^2)$, with $\mu \in (-R,R)$ and 
	$I \leftarrow M(X_1, \ldots, X_n)$, then 
	\begin{align*}
	\prm[\underline{X} \sim N(\mu,\sigma^2)][M]{ I(X_1,\ldots, X_n) \ni \mu} &\geq 1-\alpha.
	\end{align*}
	Moreover, $|I| = \beta$ and if,
  	$$
	n \geq   c \min \left\{\frac{1}{\eps}\log \left(\frac{R}{\sigma \alpha}\right), \frac{1}{\eps}\log \left(\frac{1}{\delta \alpha}\right) \right\},
	$$
	(where $c$ is a universal constant), 
	then,
	\begin{align*}
	\beta &\leq \max \left\{
	\frac{\sigma}{\sqrt{n}}  \O{\sqrt{\log \left(\frac{1}{\alpha}\right)}},
	\frac{\sigma}{\epsilon n} \mathrm{polylog}\left(\frac{n}{\alpha}\right)
	\right\}.
	\end{align*}
\end{theorem}

Theorem \ref{thm:ciknownvar} asserts that the $(\eps, \delta)$-differentially private algorithm produces a $(1-\alpha)$-fixed with confidence interval, no matter what the sample size is. Moreover, if $n$ is large enough, the width of confidence interval produced is the maximum of two terms, one of the terms being the width obtained without privacy which is $\Theta \left((\sigma/\sqrt{n})\cdot \sqrt{\log(1/\alpha)} \right)$, and the other term vanishing linearly in $n$. We will later show that a width of  $\Om{\sigma/(\eps n) \cdot \log(1/\alpha )}$ is necessary for privacy, (see Theorem \ref{thm:lowerbound} in Section \ref{sec:lowerbounds}.), so the second term cannot be significantly improved. 
 
\begin{proof}[Proof of Theorem \ref{thm:ciknownvar}]
We will show that the following algorithm produces the required confidence interval. First, we obtain a differentially private estimate of range of the data using Theorem \ref{thm:rangeKnownVariance}. The data is truncated to lie within this range and the mean of the truncated data is released using the Laplace mechanism calibrated to the estimated range.  Finally, the confidence interval is estimated using the noisy truncated mean by explicitly accounting for all sources of randomness. We present this algorithm followed by the proof of it's correctness.

\begin{algorithm}[H]
	\caption{Differentially private confidence interval with known variance.}
	\begin{algorithmic}[1]
		\label{alg:ciknownvar}
		\REQUIRE $X_1, \ldots, X_n$, $\alpha_0, \alpha_1, \alpha_2$, $\eps, \delta$, $\sigma$, $R$\\
		\ENSURE An $(\eps = \eps_1 + \eps_2, \delta)$-differentially private $(\alpha = \alpha_0+\alpha_1+\alpha_2)$-level confidence interval of $\mu$.
		
		\STATE If 	$$
		n <   c \min \left\{\frac{1}{\eps}\log \left(\frac{R}{\sigma \alpha_2}\right), \frac{1}{\eps}\log \left(\frac{1}{\delta \alpha_2}\right) \right\},
		$$
		output $(-R,R)$.
		\STATE Run the $(\eps_2, \delta)$ differentially private range estimation algorithm for known variance (Algorithm \ref{alg:dataRange}) on $(x_1, \ldots, x_n)$ to get a $(1 - \alpha_2)$ confident estimate $[X_{\min}, X_{\max}]$  of the range. Let $w_0 = X_{\max} - X_{\min}$.
		
		\STATE Let 
		$$ Y_i = \begin{cases}
		X_{i} & \text{if }X_i \in [X_{\min}, X_{\max}]\\
		X_{\max} & \text{if } X_i > X_{\max}\\
		X_{\min} & \text{if } X_i < X_{\min}
		\end{cases}$$
		
		
		\STATE Let 
		$$\tilde{\mu} = \frac{\sum_i{Y_i}}{n} + Z_1$$ where $Z_1$ is a Laplace random variable with mean $0$ and scale parameter 
		$$b_1  = \frac{w_0}{\eps_1 n}.$$
		
		\STATE Let 
		$$w = \frac{\sigma}{\sqrt{n}}z_{1-\alpha_0/2} + 
		b_1\log \left(\frac{1}{\alpha_1}\right) $$
		where $\alpha_0 + \alpha_1 + \alpha_2 = \alpha$, and $z_{1-\alpha_0/2}$ is the $(1-\alpha_0/2)$-quantile of a standard normal distribution.
		
		\STATE Output the interval: 
		$$\left( \tilde{\mu} - w, \tilde{\mu} + w \right).$$
	\end{algorithmic}
\end{algorithm}
In Step 1, if $$
n <  c \min \left\{\frac{1}{\eps}\log \left(\frac{R}{\sigma \alpha_2}\right), \frac{1}{\eps}\log \left(\frac{1}{\delta \alpha_2}\right) \right\},
$$
Algorithm \ref{alg:ciknownvar} outputs $(-R,R)$. This step is trivially $(\eps,\delta)$-differentially private. Moreover, since it is given that $\mu \in (-R,R)$, the interval is trivially an $(1-\alpha)$-level confidence interval. When the algorithm runs beyond Step 1, by composition (Theorem \ref{thm:compose}) and the privacy property of Laplace mechanism (Theorem \ref{thm:LapMech}), it follows that the algorithm is $(\eps,\delta)$-differentially private. We now prove that when the algorithm runs beyond Step 1, the interval produced in Step 6 is an $(1 - \alpha)$-level confidence interval, of fixed width $2w$ i.e.,
$$\prm[\underline{X} \sim N(\mu,\sigma^2)][M]{|\mu - \tilde{\mu}| < w} \geq 1-\alpha,$$ 
where $w$ is defined in Step 5 of the algorithm.
We have,
			\begin{align*}
			|\tilde{\mu}-\mu|  
			&= \left|\hat{\mu} + Z_1 + \frac{1}{n}\sum_i{(X_i - Y_i)} - \mu\right|  \\
			&\leq |\hat{\mu} - \mu| + \left|\frac{1}{n}\sum_i{(X_i - Y_i)}\right| + |Z_1|.
			\end{align*} 
We will now analyze these three terms. Note that
\begin{align*}
\pr{|\hat \mu - \mu| > \frac{\sigma}{\sqrt{n}}z_{1 - \alpha_0/2}} \leq \alpha_0,
\end{align*}
since $\hat \mu - \mu$ has a normal distribution with mean $0$ and variance $\sigma^2/n$ and $z_{1 - \alpha_0/2} = \phi^{-1}(1 - \alpha_0/2)$.
Next, 
\begin{align*}
\pr{\frac{1}{n} \left|\sum_i(X_i - Y_i)\right| > 0} \leq 1 - \pr{\forall i X_{\min} \leq X_i \leq X_{\max}} \leq \alpha_2,
\end{align*}
from Theorem \ref{thm:rangeKnownVariance}, since $[X_{\min}, X_{\max}]$ is a $(1 - \alpha_2)$ confident range of $X_1, \ldots, X_n$.
Finally,
\begin{align*}
\pr{|Z_1| > b_1\log\left(\frac{1}{\alpha_1}\right)} \leq \alpha_1,
\end{align*}
due to the tails of a Laplace distribution, Proposition \ref{prop:lap}. 
Thus, with probability at least $1 - (\alpha_0 + \alpha_1 + \alpha_2) = 1 -\alpha$, we have
\begin{align*}
|\tilde \mu - \mu| \leq \frac{\sigma}{\sqrt{n}}z_{1 - \alpha_0/2} + 0 + b_1\log\left(\frac{1}{\alpha_1}\right) = w,
\end{align*}
as desired.
To finish the proof, we need to upper bound the width of the confidence interval. Let $\alpha_0 =  \alpha_1 = \alpha_2 = \alpha/3$. From Theorem \ref{thm:rangeKnownVariance}, we have
$$w_0 = c \sigma \sqrt{\log \left(\frac{n}{\alpha_1}\right)}. $$
Hence, the width of the confidence interval is $2w$ where,
\begin{align*}
w 
&= \frac{\sigma}{\sqrt{n}}z_{1-\alpha_0/2} + 
\frac{c \sigma }{\eps n}
\sqrt{\log \left(\frac{n}{\alpha_1}\right)} 
\cdot \log\left( \frac{1}{\alpha_2}\right) \\
&=\frac{\sigma}{\sqrt{n}}z_{1-\alpha/6} + 
\frac{c \sigma }{\eps n}
\sqrt{\log \left(\frac{3n}{\alpha}\right)} 
\cdot \log\left( \frac{3}{\alpha}\right)
\end{align*}
By using the tail bound on the normal distribution in Proposition \ref{prop:gauss} and the relation between the tail bound and a quantile in Proposition \ref{prop:quantileBound}, we obtain the fact that $z_{1 - \alpha} \leq \sqrt{2 \log\left(1/\alpha\right)}$. 
Hence we get, 
\begin{align*}
w &\leq  \frac{\sigma}{\sqrt{n}} \O{\sqrt{\log \left(\frac{1}{\alpha}\right)}} + \frac{c_2\sigma}{\eps n} \sqrt{\log \left(\frac{n}{\alpha}\right)} \cdot \log \left(\frac{3}{\alpha}\right)\\
&= \frac{\sigma}{\sqrt{n}} \O{\sqrt{\log \left(\frac{1}{\alpha}\right)}} + \frac{\sigma}{\eps n}  \cdot \mathrm{polylog}{\left(\frac{n}{\alpha}\right)}
\end{align*}

\end{proof}


\paragraph{Nearly optimal width for finite sample sizes}
The differentially private confidence intervals obtained in Theorem \ref{thm:ciknownvar} have a multiplicative increase in their length when compared to a classical non-private confidence interval due to the hidden constant in the term $\sigma/\sqrt{n} \cdot \left(\sqrt{\log(1/\alpha)}\right)$. We show that one can avoid this behavior and obtain a differentially private confidence interval with an additive increase, i.e. the differentially private length is the sum of two terms: the first term is the non-private length and the second term that vanishes faster than the non-private length, with a mildly worse dependence on other parameters.

\begin{theorem}
	\label{thm:finitesample}
			
	 Let $\alpha_1 = \alpha_2 = \alpha/(2\sqrt{n})$. If, 
	$$n \geq c\min
	\left\{
	\frac{1}{\eps}\log \left(\frac{R}{\eps \sigma \alpha}\right),
	\frac{1}{\eps}\log \left(\frac{1}{\delta \alpha \eps}\right)
	\right\},$$ 
	then
	\begin{align*}
	w =\frac{\sigma}{\sqrt{n}}z_{1 - \alpha/2} + \frac{\sigma}{\eps n} \mathrm{polylog}\left(\frac{n}{\alpha}\right) 
	\end{align*}
The key point is that there is no hidden constant in the first term, which matches the length of the non-private confidence interval.
\end{theorem}

We will first start with an asymptotic expansion of the quantile function of a standard normal variable, which is given in Proposition \ref{prop:asymptoticQuantile} below.
\begin{proposition}
	\label{prop:asymptoticQuantile}
	Let $z_{1 - p} = \phi^{-1}(1 - p)$ denote the quantile function of the standard normal distribution. For any $\alpha \in (0,1)$ and $ 0 < \Delta < \alpha$, we have
	\begin{align*}
		z_{1-(\alpha-\Delta)/2}
		&\leq z_{1-\alpha/2} + \frac{c_1\Delta}{\alpha - \Delta}
	\end{align*}	
	where $c_1$ is a fixed constant.
\end{proposition}
\begin{proof}
	
	Let $q(p) = z_{1-p/2} = \phi^{-1}(1 - p/2)$. By the mean value theorem, we have
	\begin{align*}
		q\left(\alpha- \Delta\right) = q\left(\alpha\right) - \Delta q'\left(c\right)
	\end{align*}
	for some $c \in [ \alpha - \Delta, \alpha]$.
	One can show that 
	$$q'(p) = -\sqrt{\frac{\pi}{2}}\cdot e^{q(p)^2/2}$$ and
	Also, $q(p) = \phi^{-1}(1 - p/2) \leq \sqrt{2\log(2/p)}$, which can be obtained by applying Proposition \ref{prop:quantileBound} to a Gaussian tail bound. 
	Hence we have,
	\begin{align*}
		q(\alpha - \Delta) &= q(\alpha) + \sqrt{\frac{\pi}{2}}\cdot e^{q(c)^2/2} \cdot \Delta \\
		&\leq q(\alpha) + \sqrt{\frac{\pi}{2}}\cdot e^{q(\alpha-\Delta)^2/2} \cdot \Delta, \\
		&\mbox{ (Since $q(p)$ is decreasing in $[\alpha-\Delta, \alpha]$) } \\
		&\leq q(\alpha) + \sqrt{\frac{\pi}{2}}\cdot \frac{2}{\alpha-\Delta}\cdot \Delta  \mbox{ ( Since $q^2(p) \leq 2\log(2/p)$ )}
	\end{align*}	
\end{proof}

\begin{proof}[Proof of Theorem \ref{thm:finitesample}]
Let $\alpha_1 = \alpha_2 = \alpha/(2\sqrt{n})$. Hence $\alpha_0 = \alpha - \alpha/\sqrt{n}$.
Applying Proposition \ref{prop:asymptoticQuantile} with $\Delta = \alpha/\sqrt{n}$, we get,
	\begin{align*}
	z_{1 - \alpha_0/2} &\leq z_{1 - \alpha/2} + \frac{c_1}{\sqrt{n}}
	\end{align*}	
	The width of the confidence interval is
	\begin{align*}
	w &= \frac{\sigma}{\sqrt{n}} z_{1 - \alpha_0/2} + \frac{c\sigma}{\eps n}\sqrt{\log \left(\frac{n}{\alpha_1}\right)}\cdot \log \left(\frac{1}{\alpha_2}\right) \\
	&\leq \frac{\sigma}{\sqrt{n}} \cdot \left( z_{1 - \alpha/2} + \frac{c_1}{\sqrt{n}} \right) + \frac{c\sigma}{\eps n}\sqrt{\log \left(\frac{n\sqrt{n}}{\alpha}\right)}\cdot \log \left(\frac{\sqrt{n}}{\alpha}\right)\\
	&\leq \frac{\sigma}{\sqrt{n}}z_{1 - \alpha/2} + \frac{\sigma}{\eps n} \mathrm{polylog}\left(\frac{n}{\alpha}\right) 
	\end{align*}
	Finally, we need 
	$$ n \geq c_1\min
	\left\{
	\frac{1}{\eps}\log \left(\frac{R\sqrt{n}}{\sigma \alpha}\right),
	\frac{1}{\eps}\log \left(\frac{\sqrt{n}}{\delta \alpha }\right)
	\right\}.
	$$
	For this, it suffices to have
	$$ n \geq c_2\min
	\left\{
	\frac{1}{\eps}\log \left(\frac{R}{\eps \sigma \alpha}\right),
	\frac{1}{\eps}\log \left(\frac{1}{\delta \alpha \eps}\right)
	\right\}.
	$$
%
	
\end{proof}

\section{Estimating a confidence interval of a mean with unknown variance}

\label{sec:ciunkownvar}
In this section, we present a differentially private algorithm for estimating the confidence interval of the mean of a normal population when the variance is unknown. Let us first recall the interval in the non-private case. The  standard confidence interval for $\mu$ in the unknown variance case is given by $\bar{X} \pm s/\sqrt{n} \cdot t_{n-1,\alpha/2}$, where 
$$s^2 = \frac{1}{n-1} \sum_{i=1}^n (X_i - \bar{X})^2,$$
is the sample variance, $t_{n-1,\alpha/2}$ is the $\alpha/2^{th}$ quantile of a $t$-distribution with $n-1$ degrees of freedom (see Section \ref{sec:defs} for definitions.)

Note that unlike the known variance case, we need to use the sample variance $s^2$ as an estimate of $\sigma^2$, and to account for this, we need to replace the $z$-quantile with a $t$-quantile. Moreover, when the variance is unknown, the optimal confidence interval is no longer of fixed length, in contrast with the known variance case. The expected length of the interval is given by  
$$\frac{2\sigma}{\sqrt{n}}\cdot k_n\cdot t_{n-1,1 - \alpha/2}= \Theta \left(\sigma \sqrt{\log(1/\alpha)/n}\right),$$ 
where 
$$k_n = \sqrt{\frac{2}{n-1}}\cdot \frac{\Gamma\left(\frac{n}{2}\right)}{\Gamma\left(\frac{n-1}{2}\right)} = \O{1 - \frac{1}{n}}.$$
(see for example \cite{lehmann2006testing} and \cite{Bolch}.)

Let us now consider a differentially private algorithm for estimating the confidence interval when the variance is unknown. Recall that in the known variance case given in Algorithm \ref{alg:ciknownvar}, we first estimate the range of the data using $\sigma$, and then estimate the mean of the data truncated to the estimated range, using a Laplace mechanism. A conservative confidence interval is estimated by considering all sources of randomness, including the noise due to privacy, and the range estimation and truncation step. When $\sigma$ is unknown, we proceed by using Algorithm \ref{alg:estvar} to get a crude estimate of $\sigma$. Once we have an estimate of $\sigma$, we proceed as in the known variance case to estimate the range,followed by using a Laplace mechanism to compute the mean of the data truncated to the estimated range. For construction of a differentially private confidence interval in the unknown variance case, we also need a conservative estimate of the sample variance of the truncated data using Laplace Mechanism. Finally,  the analysis of the algorithm takes into account all sources of randomness to ensure the required coverage is obtained.
Algorithm \ref{alg:ciunknownvar} describes the estimator of confidence interval with unknown variance.

\begin{theorem}
	\label{thm:ciunknownvar}
	Let $\mathbb{I}$ be the set of all possible intervals in $\mathbb{R}$. For every $n \in \mathbb{N}$, $\sigma_{\min} < \sigma_{\max} \in [0,\infty],\eps, \delta > 0$, $\alpha \in \left(0,1/2\right)$, $R \in [0,\infty)$, there is an $(\eps,\delta)$-differentially private algorithm $M: \mathbb{R}^n \rightarrow \mathbb{I}$ such that 
	if $X_1, \ldots, X_n$ are iid random variables from $N(\mu,\sigma^2)$, where $\mu \in (-R,R)$ and $\sigma \in (\sigma_{\min} , \sigma_{\max})$,and 
	$I \leftarrow M(X_1, \ldots, X_n)$, then 
	\begin{align*}
	\prm[\underline{X} \sim N(\mu,\sigma^2)][M]{I(X_1,\ldots, X_n) \ni \mu} &\geq 1-\alpha.
	\end{align*}
	Moreover, if
	$$n \geq  \frac{c_1}{\eps} 
	\min \left\{
	\max \left\{
	\log \left(\frac{R}{\sigma_{\min}}\right), \log \left(\frac{\sigma_{\max}}{\sigma_{\min}}\right)
	\right\},
	\log \left(\frac{1}{\delta}\right)
	\right\} + \frac{c_2}{\eps}\log \left(\frac{\log \left(\frac{1}{ \eps}\right)}{\alpha}\right).
	$$		
(where $c_1$ and $c_2$ are universal constants), 
	then,
	\begin{align*}
	\beta := \Em[\underline{X} \sim N(\mu,\sigma^2)][M]{|I(X_1,\ldots, X_n)|} &\leq \max \left\{
	\frac{\sigma}{\sqrt{n}} \O{ \sqrt{\log\left(\frac{1}{\alpha}\right)}},
	\frac{\sigma}{\epsilon n} \mathrm{polylog} \left(\frac{1}{\alpha}\right)
	\right\}
	\end{align*}
\end{theorem}

As in the case of known variance, Theorem \ref{thm:ciunknownvar} asserts that there exists an $(\eps, \delta)$-differentially private algorithm that produces a $(1-\alpha)$-confidence interval, no matter what the sample size is. Moreover, if $n$ is large enough, the width of the confidence interval is non-trivial. It is the maximum of two terms. The first term being the width of interval obtained without privacy which is $\Theta \left((\sigma/\sqrt{n})\cdot \sqrt{\log(1/\alpha)} \right)$, and hence is necessary. Using the fact that a lower bound of the known variance applies to the unknown variance, it will be shown that the second term of  $\Om{\sigma/(\eps n) \cdot \log(1/\alpha )}$ is also necessary for privacy. (See Theorem \ref{thm:lowerbound} in Section \ref{sec:lowerbounds}.) Thus we can match the lower bound upto polylog factors. 


\begin{proof}
	We first start with the algorithm:
	\begin{algorithm}[H]
		\caption{Differentially private confidence interval with unknown variance.}
		\begin{algorithmic}[1]
			\label{alg:ciunknownvar}
			\REQUIRE $X_1, \ldots, X_n$, $\alpha_0, \alpha_1, \alpha_2, \alpha_3$, $\eps_1, \eps_2, \eps_3, \delta$, $\sigma_{\min}, \sigma_{\max}$, $R$\\
			\ENSURE An $(\eps = \eps_1 + \eps_2 + \eps_3, \delta)$-differentially private $(1 - \alpha)$-level confidence interval of $\mu$, where $\alpha = \alpha_0+\alpha_1+\alpha_2+ \alpha_3$.
			
			\STATE If 		$$n <  c 
			\min \left\{
			\max \left( 
			\frac{1}{\eps_3}\log \left(\frac{R}{\sigma_{\min} \alpha_3}\right), \frac{1}{\eps_3}\log \left(\frac{\log_2 \frac{\sigma_{\max}}{\sigma_{\min}}}{\alpha_3}\right)
			\right),
			\frac{1}{\eps_3}\log \left(\frac{1}{\delta \alpha_3}\right)
			\right\},
			$$
			
			output $(-R,R)$.
			\STATE Run an $(\eps_3, \delta)$ differentially private range estimation algorithm for unknown variance from Theorem \ref{thm:rangeunkownvar} on $(X_1, \ldots, X_n)$ to get a $(1 - \alpha_3)$ confident estimate of the range $[X_{\min}, X_{\max}]$. Let $w_0 = X_{\max} - X_{\min}$.
			
			\STATE Let 
			$$ Y_i = \begin{cases}
			 X_{i} & \text{if }X_i \in [X_{\min}, X_{\max}]\\
			 X_{\max} & \text{if } X_i > X_{\max}\\
			 X_{\min} & \text{if } X_i < X_{\min}
			\end{cases}$$
			
			
			\STATE Let 
			$$\tilde{\mu} = \frac{\sum_i{Y_i}}{n} + Z_1$$ where $Z_1$ is a Laplace random variable with mean $0$ and scale parameter 
			$$b_1  = \frac{w_0}{\eps_1 n}.$$
			
			\STATE Truncate $\tilde{\mu}$ to lie in the interval $[X_{\min}, X_{\max}]$.
			\STATE Let $$s_1^2 =\frac{\sum_i (Y_i-\tilde{\mu})^2}{n-1}.$$ 
			
			\STATE Let $\tilde{s}^2 = s_1^2 + Z_2 + b_2\log \left(\frac{1}{\alpha_2}\right)$ where $Z_2 \sim Lap(0,b_2)$ and 
			$$ b_2  = \frac{w_0^2}{\eps_2\cdot (n-1)}.$$
			\STATE If $\tilde{s}^2 < 0$ or $\tilde{s}^2 > \sigma_{\max}^2$, set $\tilde{s}^2 = \sigma_{\max}^2$.
			\STATE 
			$$w = \frac{\tilde{s}}{\sqrt{n}}t_{n-1,\alpha_0/2} + 
			b_1\log \left(\frac{1}{\alpha_1}\right) $$
			where $\alpha_0 + \alpha_1 + \alpha_2 + \alpha_3 = \alpha$, and $t_{n-1,\alpha_0/2}$ is the $(1-\alpha_0/2)$-quantile of a $t$-distribution with $n-1$ degrees of freedom.  (See Section \ref{sec:defs} for Definitions.) 
			
			\STATE Output the interval: 
			$$\left[ \tilde{\mu} - w, \tilde{\mu} + w \right].$$
		\end{algorithmic}
	\end{algorithm}
	
	In the algorithm, we need to partition $\epsilon = \eps_1 + \eps_2 + \eps_3$ and $\alpha = \alpha_0 + \alpha_1 + \alpha_2 + \alpha_3$. We will defer this choice to the step when we bound the expected length of the confidence interval. 
	
	If the condition in Step 1 is satisfied, the output is trivially an $(\eps, \delta)$-differentially private $(1 - \alpha)$-level confidence interval. So we focus on the case when the algorithm runs beyond Step 1. The proof has three parts, proof of privacy, coverage guarantee, and the expected length. We begin with the proof of privacy.
	
	\paragraph{Privacy:}  The proof of privacy is the same as in Theorem \ref{thm:ciknownvar}, except we also need to account for the computation of the estimate $\tilde s$ of the standard deviation. Note that the global sensitivity of $s_1^2$ in Step 6 is $w_0^2/(n-1)$, so $\tilde s^2$ is $(\eps_2,0)$-differentially private. By composition of differential privacy (Theorem \ref{thm:compose}) and the properties of Laplace mechanism (Theorem \ref{thm:LapMech}), it follows that the algorithm is $(\eps, \delta)$-DP. 
	
	\paragraph{Coverage:} Now we will show that the algorithm outputs an interval with the desired coverage.  
%
Let $E$ be the following event:
		$$
		\left\{ |\tilde \mu - \mu| \leq \frac{\tilde s}{\sqrt{n}}t_{n-1,\alpha_0/2} + b_1 \log\left(\frac{1}{\alpha_1}\right)
		\right\}
		$$
	We need to show $\pr{E} \geq 1 - \alpha$.
		Let $\hat{\mu} = (\sum_i X_i)/n$, and
		$$s^2 = \frac{1}{(n-1)}\sum_i (X_i - \hat{\mu})^2.$$
		Consider the following events:
		\begin{align*}
		E_0 & = \left\{ |\hat{\mu} - \mu| \leq \frac{s}{\sqrt{n}}t_{n-1,\alpha_0/2} \right\}\\
		E_1 &= \left\{|Z_1| \leq  	b_1\log \left( \frac{1}{\alpha_1} \right) \right\} \\
		E_2 &= \left\{|Z_2| \leq    b_2\log \left( \frac{1}{\alpha_2} \right) \right\}\\
		E_3 &= \left\{\forall i, X_i = Y_i \text{ and }  1 \leq \frac{\hat \sigma}{\sigma} \leq 8\right\}\\ 
		\end{align*}
		\paragraph{Claim 1} $E_0 \cap E_1 \cap E_2 \cap E_3 \implies E$
		\begin{proof}[Proof of Claim 1]
			Assuming that $E_0, E_1, E_2, E_3$ all hold. We will first show that $\tilde s^2$ is a conservative estimate of $s^2$.
			From $E_3$, we have $X_i = Y_i$. Hence $\tilde \mu = \hat \mu+ Z_1$ and 
			\begin{align*}
			s_1 ^2 &= \frac{\sum_i{(Y_i - \hat \mu + Z_1)^2}}{n-1} \\
			&= \frac{\sum_i (X_i - \hat \mu)^2}{n-1} + \frac{n}{n-1}\cdot Z_1^2 + 2 \frac{\sum_i (X_i - \hat \mu)}{n-1}\\
			&= \frac{\sum_i (X_i - \hat \mu)^2}{n-1} + \frac{n}{n-1} \cdot Z_1^2 + 2 \frac{\sum_i X_i - n \hat \mu}{n-1}\\
			& = s^2 + \frac{n}{n-1}\cdot Z_1^2.
			\end{align*}
			We have 
			\begin{align}
			\tilde{s}^2 &= s_1^2 + Z_2 + b_2\log \left(\frac{1}{\alpha_2}\right) \nonumber \\
			& = s^2 + \frac{n}{n-1}Z_1^2 + Z_2 + b_2\log \left(\frac{1}{\alpha_2}\right) \label{eq:sunderA}\\
			& > s^2, \mbox{ since by $E_2$, $|Z_2| \leq b_2 \log \left(\frac{1}{\alpha_2}\right)$ } \nonumber
			\end{align}
			Hence, since $s^2 \leq \tilde s^2$, from $E_0$ we have, 
			$$|\hat{\mu} - \mu| \leq \frac{s}{\sqrt{n}}t_{n-1,\frac{\alpha_0}{2}} \leq \frac{\tilde{s}}{\sqrt{n}}t_{n-1,\frac{\alpha_0}{2}}$$
			Finally, using $E_1, E_3$ and the fact that 
			$$|\tilde{\mu} - \mu| \leq |\hat{\mu}-\mu|+ \left|\frac{\sum_i(X_i - Y_i)}{n}\right|+ |Z_1|,$$ we get,
			$$|\tilde{\mu} - \mu| \leq  \frac{\tilde{s}}{\sqrt{n}}t_{n-1,\frac{\alpha_0}{2}} + b_1\log \left(\frac{1}{\alpha_1}\right)$$	
		\end{proof}
Thus, by union bound, to prove that $\pr{E} > 1 - \alpha$, it is enough to show the following claim:
		\paragraph{Claim 2} $\pr{E_0} \geq 1 - \alpha_0$, $\pr{E_1} \geq 1 - \alpha_1$, $\pr{E_2} \geq 1- \alpha_2$ and $\pr{E_3} \geq 1 - \alpha_3$.
		\begin{proof}[Proof of Claim 2]
			By the properties of a non-private confidence interval, (see for e.g. \citep{lehmann2006testing}), we have, $\pr{E_0} \geq 1-\alpha_0$. By tail properties of a Laplace distribution given in Proposition \ref{prop:lap}, we have  $\pr{E_1} \geq 1-\alpha_1$, $\pr{E_2} \geq 1- \alpha_2$.  Finally, by Theorem \ref{thm:rangeunkownvar}, if 	
				$$n \geq  c 
				\min \left\{
				\max \left( 
				\frac{1}{\eps_3}\log \left(\frac{R}{\sigma_{\min} \alpha_3}\right), \frac{1}{\eps_3}\log \left(\frac{\log (\sigma_{\max}/\sigma_{\min})}{\alpha_3}\right)
				\right),
				\frac{1}{\eps_3}\log \left(\frac{1}{\delta \alpha_3}\right)
				\right\},
				$$ then,
			$$\pr{\forall i, X_i = Y_i} = \pr{ \forall i, X_{\min} \leq X_i \leq X_{\max}} \geq 1 - \alpha_3$$
		\end{proof}
	
	\paragraph{Length of the Interval:} 
	We now need to upper bound the expected length of the interval. Recall that the length of the interval is $2w$, where 
	$$w = \frac{\tilde{s}}{\sqrt{n}}t_{n-1,\alpha_0/2} + b_1\log \left(\frac{1}{\alpha_1}\right). $$
	Here, $\tilde s^2$ and $b_1$ are the only random variables; The following proposition upper bounds the expectation of both these terms for a specific choice of $\alpha_0,\alpha_1, \alpha_2,$ and $\alpha_3$.
	
	\begin{proposition}
		\label{prop:upperboundvariance}
			Let 
			$$\alpha_3 = \min \left\{ \frac{\alpha}{4}, \frac{1}{\log^2 n}\right\},$$
				$\alpha_0 = \alpha_1 = \alpha_2 = (\alpha-\alpha_3)/3$, and $\eps_1 = \eps_2 = \eps_3 = \epsilon/3$, then we have
	\begin{align}
	\mathbb E [\tilde s] \leq k_n \cdot \sigma +  \sigma \cdot \O{\sqrt{\frac{\log \left(\frac{n}{\alpha}\right)}{\eps n}}\cdot \sqrt{\log \left(\frac{1}{\alpha}\right)} }
	\end{align}
	where $$k_n = \sqrt{\frac{2}{n-1}}\cdot \frac{\Gamma\left(\frac{n}{2}\right)}{\Gamma\left(\frac{n-1}{2}\right)}$$ and
	\begin{align}
		\mathbb E[b_1] \leq \frac{\sigma}{\eps n} \cdot \O{\sqrt{\log(n/\alpha)}}.
	\end{align}	
	\end{proposition} 
	Deferring the proof of Proposition \ref{prop:upperboundvariance} to the end, let us upper bound $w$.
	Using the settings of $\alpha_i$'s specified by Proposition \ref{prop:upperboundvariance}, the expected value of $w$ is:
		\begin{align}
		\label{eq:expectedwidth}
		\mathbb E[w] 
		= 		&\frac{\mathbb E[\tilde{s}]}{\sqrt{n}}\cdot t_{n-1,\alpha_0/2} + \mathbb E[b_1] \cdot \log \left(\frac{1}{\alpha_1}\right) \nonumber \\ 
		\leq	&	\frac{\sigma}{\sqrt{n}} \cdot t_{n-1,\alpha_0/2}\left[ k_n +  \O{\sqrt{\frac{\log(n/\alpha)}{\eps n}} \cdot \sqrt{\log \left(\frac{1}{\alpha}\right)} }  \right] \nonumber \\
		&+ 
		\O{\frac{\sigma}{\eps n}\sqrt{\log \left(\frac{n}{\alpha}\right) } \cdot  \log \left(\frac{1}{\alpha}\right)} \nonumber \\
		\leq	& k_n \cdot \frac{\sigma}{\sqrt{n}}t_{n-1,\alpha_0/2} + \frac{\sigma}{\eps n} \cdot t_{n-1,\alpha_0/2} \cdot \mathrm{polylog} \left(\frac{n}{\alpha}\right) 
		\end{align}
		From Proposition \ref{prop:tailBoundT} Part (2), if $n \geq \O{\log(1/\alpha_0)} = \O{\log(1/\alpha)}$, we have,
		\begin{align*}
		t_{n-1,\alpha_0/2} &\leq \O{\sqrt{\log(1/\alpha_0)}} = \O{\sqrt{\log(1/\alpha)}},
		\end{align*}
		since $\alpha_0 \in [\alpha/4, \alpha/3]$.
		Substituting this in equation \ref{eq:expectedwidth}, we get
		\begin{align*}
		\E{w} &\leq \frac{k_n\sigma}{\sqrt{n}}\O{\sqrt{\log(1/\alpha)}} + \frac{\sigma}{\eps n} \O{\sqrt{\log(1/\alpha)}} \mathrm{polylog} \left(\frac{n}{\alpha}\right)  \\
		&=\frac{k_n\sigma}{\sqrt{n}}\O{\sqrt{\log(1/\alpha)}} + \frac{\sigma}{\eps n} \mathrm{polylog} \left(\frac{n}{\alpha}\right)
		\end{align*}
		Finally, let us verify the conditions on the sample size $n$. We have used the following two conditions on $n$: We need 
		$$n \geq  \frac{c}{\eps_3} 
		\min \left\{
		\max \left( 
		\log \left(\frac{R}{\sigma_{\min} \alpha_3}\right), \log \left(\frac{\log (\sigma_{\max}/\sigma_{\min})}{\alpha_3}\right)
		\right),
		\log \left(\frac{1}{\delta \alpha_3}\right)
		\right\},
		$$ and from, Theorem \ref{thm:rangeunkownvar} to obtain the upper bound on the expectation of the variance, we need 
		$$n \geq \frac{c}{\eps_3} \min 
		\left\{
		\log \left(\frac{\sigma_{\max}}{\sigma_{\min}}\right), 
		\log \left(\frac{1}{\alpha_3 \delta}\right) 
		\right\}.
		$$
		Since $\alpha_3 = \min \left\{\alpha/4, 1/\log^2 n \right\}$ and $\eps_3 = \eps/3$, it suffices to require
		$$n \geq  c 
		\min \left\{
		\max \left( 
		\frac{1}{\eps}\log \left(\frac{R\log(1/\eps)}{\sigma_{\min} \alpha}\right), \frac{1}{\eps}\log \left(\frac{\sigma_{\max}\log (1/\eps)}{\sigma_{\min}\alpha}\right)
		\right),
		\frac{1}{\eps}\log \left(\frac{\log (1/\eps)}{\delta \alpha}\right)
		\right\}.
		$$
		Taking the common terms out, we get,
		$$n \geq  c 
		\min \left\{
		\max \left\{ 
		\log \left(\frac{R}{\sigma_{\min}}\right), \log \left(\frac{\sigma_{\max}}{\sigma_{\min}}\right)
		\right\},
		\log \left(\frac{1}{\delta}\right)
		\right\} + \frac{c}{\eps}\log \left(\frac{\log(1/\eps)}{\alpha} \right).
		$$
	All that remains is a proof of Proposition \ref{prop:upperboundvariance}, which is given below.
	\begin{proof}
	Let $A = \left\{\forall i, X_i = Y_i \right\}.$ 
	Under the event $A$, from equation \ref{eq:sunderA}, we have,
		\begin{align*}
		\tilde s^2	& = s^2 + \frac{n}{n-1}Z_1^2 + Z_2 + b_2\log \left(\frac{1}{\alpha_2}\right)\\
		            & := s^2 + Y_1.
		\end{align*}
  Under the event $A^c$, we can use the following upper bound of $\tilde s^2$:
	$$ \tilde s^2 \leq \frac{w_0^2}{\eps_2 \cdot (n-1)} := Y_2.$$
	 Let $\mathbb I_{A}$ be the indicator function of event $A$. We have,	
	\begin{align*}
		\tilde s^2     = 	&  \tilde s^2 \cdot \mathbb I_A + \tilde s^2 \cdot \mathbb I_{A^c} \\
					\leq 	&  (s^2 + Y_1) \cdot \mathbb I_A +  Y_2 \cdot \mathbb I_{A^c}  \\
					\leq 	&	s^2 + Y_1 + Y_2.
	\end{align*}
Using the inequality $\sqrt{a + b + c} \leq \sqrt{a} + \sqrt{b} + \sqrt{c}$ for non-negative reals and Jensen's inequality for $f(x) = \sqrt{x}$, we have,
	\begin{align}
	\label{eq:divideS^2}
		\mathbb E[\tilde s]  \leq \E{ s } + \sqrt{\E{Y_1}} + \sqrt{\E{Y_2}}. 		
	\end{align}
		Let us compute the expectation of the three terms in equation \ref{eq:divideS^2}. 
		
		\noindent \textbf{(1)} First note that 
		$\E{s} = k_n \sigma,$ where $$k_n = \sqrt{\frac{2}{n-1}}\cdot \frac{\Gamma\left(\frac{n}{2}\right)}{\Gamma\left(\frac{n-1}{2}\right)}.$$
		(See, for example, \cite{Bolch}.)
		
		\noindent \textbf{(2)} Next, we have
		\begin{align*}
		\mathbb E[Y_1]	& = \E{\frac{n}{n-1} \cdot Z_1^2 + Z_2 + b_2\log \left(\frac{1}{\alpha_2}\right)}\\
		&=  \frac{2n}{n-1}\cdot \mathbb E[b_1^2] + 0 + \mathbb E[b_2]\cdot \log \left(\frac{1}{\alpha_2}\right) \\
		& = \frac{2n}{n-1}\cdot \mathbb E\left[\frac{w_0^2}{\eps_1^2 n^2} \right] + \mathbb E\left[\frac{w_0^2}{\eps_2 \cdot (n-1)} \right]\cdot \log \left(\frac{1}{\alpha_2}\right)
		\end{align*}
		Note that $b_1$ and $b_2$,  the scales of $Z_1$ and $Z_2$ respectively, are also random. We have used the fact that if $Z$ is a Laplace random variable with mean $0$ and scale $b$, then $\E{Z} = 0$ and $\E{Z^2} = 2b^2$, see Proposition \ref{prop:lap}. 
		Now let us upper bound the expectation of $w_0^2$. 	Recall that 
		$$w_0 = X_{\max} - X_{\min} = \hat \sigma \sqrt{\log(n/\alpha_3)}.$$
		Given that  
		$$\alpha_3 = \min \left\{ \frac{\alpha}{4}, \frac{1}{\log^2 n}\right\},$$
		and	$\alpha_0 = \alpha_1 = \alpha_2 = (\alpha-\alpha_3)/3 \in [\alpha/4,\alpha/3]$.
		From Theorem \ref{thm:rangeunkownvar}, if $n \geq (c/\eps)\log (\sigma_{\max}/\sigma_{\min})$, we have,
		\begin{align*}
		\mathbb E[\hat \sigma^2] &\leq \sigma^2 \left( c_1 + c_2 \alpha_3 \log^2(n)\right) \leq \sigma^2 \left(c_1 + c_2 \right) \leq K \sigma^2
		\end{align*}
		Hence, 
		$$\mathbb E[w_0^2] \leq K \sigma^2 \cdot \log\left(\frac{n}{\min \left\{\alpha/4, 1/(\log^2 n) \right\}}\right) = \sigma^2 \cdot \O{\log (n/\alpha)}.$$
		Hence we have,
		\begin{align}
		\mathbb	E[Y_1]
		& \leq \frac{2n}{n-1}\cdot \mathbb E\left[\frac{w_0^2}{\eps_1^2 n^2}\right] + \mathbb E\left[\frac{w_0^2}{\eps_2 \cdot (n-1)} \right]\log \left(\frac{1}{\alpha_2}\right) \nonumber \\
		&\leq \sigma^2\left(\O{\frac{\log(n/\alpha)}{\eps^2 n^2} } + \O{\frac{\log(n/\alpha)}{\eps n} \cdot \log\left(\frac{1}{\alpha}\right)}\right) \label{eq:s_sum_part1}
		\end{align}	
since $\eps_1 = \eps_2 =\eps_3 = \eps/3$ and $\alpha_2 \in [\alpha/4, \alpha/3]$.
        
\noindent \textbf{(3)} Finally, we have
		\begin{align}
		\E{Y_2} &\leq \frac{\mathbb E[w_0^2]}{\eps_2\cdot (n-1)} \nonumber \\
				&\leq \O{\frac{\sigma^2}{\eps n} \cdot \log \left(\frac{n}{\alpha}\right)} \label{eq:s_sum_part2}
		\end{align}
		Combining inequalities \ref{eq:s_sum_part1} and \ref{eq:s_sum_part2}, and plugging in equation \ref{eq:divideS^2}, we have
		\begin{align*}
		\mathbb	E[\tilde s]	
		&\leq k_n\sigma + \sigma \cdot \left[ \O{\sqrt{\frac{\log(n/\alpha)}{\eps n }} \cdot \sqrt{\log (1/\alpha)}} + \O{\frac{\sqrt{\log(n/\alpha)}}{\eps n} } 
		\right].
		\end{align*}
		Similarly, we have,
		\begin{align*}
		\mathbb E[b_1] = \mathbb E\left[\frac{w_0}{\eps_1 n}\right] \leq 
        \frac{\sqrt{\E{w_0^2}}}{\epsilon_1 n } \leq \frac{\sigma}{\eps n} \cdot \O{\sqrt{\log(n/\alpha)}},
		\end{align*}
		since $\epsilon_1 = \epsilon/3$.
	\end{proof}

\end{proof}

\paragraph{Nearly optimal width for finite sample sizes.}

 As in the case of known variance, we can obtain a differentially private confidence interval whose increase in length is additive as opposed to multiplicative with a minor increase in the sample complexity.

\begin{theorem}
	\label{thm:finitesampleT}
	
	Let $\alpha_1 = \alpha_2 = \alpha_3 = \alpha/(3\sqrt{n})$. If, 
	$$n \geq  c 
	\min \left\{
	\max \left( 
	\frac{1}{\eps}\log \left(\frac{R}{\eps \sigma_{\min} \alpha}\right), \frac{1}{\eps}\log \left(\frac{\sigma_{\max}}{\eps\sigma_{\min}\alpha}\right)
	\right),
	\frac{1}{\eps}\log \left(\frac{1}{\eps \delta \alpha}\right)
	\right\},
	$$
	then
	\begin{align*}
	w =k_n \cdot \frac{\sigma}{\sqrt{n}}t_{n-1,1 - \alpha/2} + \frac{\sigma}{\eps n} \cdot \mathrm{polylog}\left(\frac{n}{\alpha}\right) + \O{\frac{\sigma}{n }}t_{n-1,1-\alpha/2}
	\end{align*}
	where $$k_n = \sqrt{\frac{2}{n-1}}\cdot \frac{\Gamma\left(\frac{n}{2}\right)}{\Gamma\left(\frac{n-1}{2}\right)}$$
\end{theorem}

We will first start with an asymptotic expansion of the quantile function of a $t$-distribution, which is given in Proposition \ref{prop:asymptoticQuantileT} below.
\begin{proposition}
	\label{prop:asymptoticQuantileT}
	Let $t_{n,1 - p} = \phi^{-1}(1 - p)$ denote the $1-p^{th}$ quantile of a $t$-distribution with $n$ degrees of freedom. For any $\alpha \in (0,1/2)$ and $ 0 < \Delta < \alpha$, we have
	\begin{align*}
	t_{n,1-(\alpha-\Delta)/2}
	&\leq t_{n,1-\alpha/2} + \frac{c_1\Delta}{(\alpha-\Delta)} 
	\end{align*}	
	where $c_1$ is a fixed constant.
\end{proposition}
\begin{proof}
	
	Let $q(p) = t_{n,1-p/2} = G_n^{-1}(1 - p/2)$ where $G_n(.)$ is the cdf of a $t$-distribution with $n$ degrees of freedom. By the mean value theorem, we have
	\begin{align*}
	q\left(\alpha- \Delta\right) = q\left(\alpha\right) - \Delta q'\left(c\right)  
	\end{align*}
	for some $ c \in [\alpha-\Delta, \alpha]$.
	By taking the derivative, using Proposition \ref{prop:tailBoundT} for a bound on the quantile of a $t$-distribution, one can show that 
	\begin{align*}
	- q'(p) &= \frac{\sqrt{n\pi}}{2} \frac{\Gamma\left(n/2\right)}{\Gamma\left((n+1)/2\right)}\left( 1+ \frac{q^2(p)}{n}\right)^{(n+1)/2} \\
	&\leq \sqrt{\frac{\pi}{2}}\cdot \sqrt{\frac{n}{n-1}}\left(1+ \frac{q^2(p)}{n}\right)^{(n+1)/2} \\
    & \leq c \cdot \exp\left(q^2(p) \cdot \frac{n+1}{2n} \right) \\
    & \leq \frac{c}{p} \mbox{ (Proposition \ref{prop:t-tail} Part(2))}
	\end{align*}
	Hence we have,
	\begin{align*}
	q(\alpha - \Delta) &\leq q(\alpha) + \Delta\frac{c_1}{(\alpha-\Delta)}
	\end{align*}	
\end{proof}

\begin{proof}[Proof of Theorem \ref{thm:finitesampleT}]
	Let $\alpha_1 = \alpha_2 = \alpha_3 = \alpha/(3\sqrt{n})$. Hence $\alpha_0 = \alpha - \alpha/\sqrt{n}$.
	Applying Proposition \ref{prop:asymptoticQuantileT} with $\Delta = \alpha/\sqrt{n}$, we get,
\begin{align*}
	t_{n-1,1 - \alpha_0/2} &\leq t_{n-1,1 - \alpha/2} + \frac{c_1}{\sqrt{n}} 
\end{align*}
	From equation \ref{eq:expectedwidth} in the proof of Theorem  \ref{thm:ciunknownvar}, we have,
\begin{align*}
	\mathbb E[w]
	&\leq k_n \cdot \frac{\sigma}{\sqrt{n}} \cdot t_{n-1,1- \alpha_0/2} + \frac{\sigma}{\eps n} \cdot \mathrm{polylog} \left(\frac{n}{\alpha}\right) \\
	&\leq k_n \cdot \frac{\sigma}{\sqrt{n}} \left( t_{n-1,1-\alpha/2} + \frac{c_1}{\sqrt{n}} \right) + \frac{\sigma}{\eps n} \mathrm{polylog} \left(\frac{n}{\alpha}\right)  \\
    &\leq k_n \cdot \frac{\sigma}{\sqrt{n}} t_{n-1,1-\alpha/2} + \frac{\sigma}{\eps n} \mathrm{polylog} \left(\frac{n}{\alpha}\right)
\end{align*}
	For the sample complexity, we need 
	$$n \geq  c 
		\min \left\{
		\max \left( 
		\frac{1}{\eps}\log \left(\frac{R\sqrt{n}}{\sigma_{\min} \alpha}\right), \frac{1}{\eps}\log \left(\frac{\sigma_{\max} \sqrt{n}}{\sigma_{\min}\alpha}\right)
		\right),
		\frac{1}{\eps}\log \left(\frac{\sqrt{n}}{\delta \alpha}\right)
		\right\},
	$$
	It suffices to have
	$$n \geq  c 
		\min \left\{
		\max \left( 
		\frac{1}{\eps}\log \left(\frac{R}{\eps \sigma_{\min} \alpha}\right), \frac{1}{\eps}\log \left(\frac{\sigma_{\max}}{\eps\sigma_{\min}\alpha}\right)
		\right),
		\frac{1}{\eps}\log \left(\frac{1}{\eps \delta \alpha}\right)
		\right\}.
	$$
\end{proof}

\section{Lower Bounds for confidence intervals with known variance}
\label{sec:lowerbounds}
In this section, we prove lower bounds on the expected size of $(1 - \alpha)$-level confidence sets obtained with and without differential privacy. We begin with deriving lower bounds on the expected size of confidence sets without the privacy requirement, followed by the case when the confidence sets are required to be $(\eps, \delta)$ differentially private. The non-private lower bounds are useful for comparisons and to understand the price that we have to pay to ensure differential privacy. 

\subsection{Lower bounds on Confidence Intervals without privacy}
In this subsection, we derive lower bounds on confidence sets for the normal mean where there are no privacy constraints. The non-private confidence set estimation problem that we consider here differs from the standard problem of a estimating confidence set for the mean of the normal population. For $\epsilon$-differential privacy, we require $\mu \in (-R,R)$. The corresponding problem without privacy is called the \emph{bounded normal mean problem} and has a long history, see for example \cite{pratt1963shorter, evans2005minimax,schafer2009constructing} for the problem of estimating a minimax confidence interval for the bounded normal mean, and \cite{casella1981estimating,bickel1981minimax, marchand2004estimation} for the problem of point estimation of the bounded normal mean. 

It is well known (see for example \cite{lehmann2006testing}), that the classic confidence interval $\bar{X} \pm \sigma/\sqrt{n}$ is of minimax expected size among all $1-\alpha$ level confidence sets for the normal mean with known variance, when $\mu \in \mathbb{R}$. On the other hand, when we assume that $\mu \in (-R,R)$, the classic confidence interval is no longer minimax.  When $R \leq 2\sigma z_{1-\alpha/2}$, \cite{evans2005minimax} show that the so called \emph{truncated Pratt interval} \citep{pratt1963shorter} is of expected minimax length. In general, the minimax confidence set is not known.

However, we are interested in estimating a confidence set in a regime where $R$ is on the order of, or much larger than the standard deviation, i.e. $R =  \Theta(\sigma)$ or $R \gg \sigma$. Under this regime, we derive a lower bound on the expected size of any $(1-\alpha)$ confidence set of the bounded mean. The main result of this section is stated below.

\begin{theorem}
	\label{thm:lowerboundwithoutprivacy}
	Let $X_1, \ldots, X_n \overset{iid}{\sim} N(\mu, \sigma^2)$ where $\mu \in (-R,R)$, $\sigma^2$ is known and $R > c\sigma^2$ for a fixed (but arbitrarily small) constant $c>0$. Let $I = I(X_1, \ldots, X_n)$ be any  measurable set where $I \subseteq (-R,R)$ such that for every $\alpha \in (0,1/2)$, and $\mu \in (-R,R),$ if,
	$$\prm[\underline{X} \sim N(\mu,\sigma^2)][]{I(X_1,\ldots, X_n) \ni \mu }  \geq 1 - \alpha.$$		
	$$\Em[\underline{X} \sim N(\mu,\sigma^2)][]{|I(X_1,\ldots, X_n)|} \geq \frac{2\sigma}{\sqrt{n}}z_{1 - \frac{\alpha}{2}} - \Otilde{\frac{1}{n}}.
	$$	
\end{theorem}  
 
Theorem \ref{thm:lowerboundwithoutprivacy} states that when $R = \Omega\left(\sigma\right)$, the lower bound on the expected size of a confidence set without privacy is equal to the length of the classic interval minus a term that goes to $0$ at the rate $1/n$, upto logarithmic factors.
Before we present a proof of Theorem \ref{thm:lowerboundwithoutprivacy}, we need some intermediate results that are stated below:
\begin{fact}[See for example \cite{hogg1995introduction}]
	\label{fact:posterior}
	If $\mu \sim N(0,\gamma^2)$, and $X_1, \ldots, X_n| \mu \overset{iid}{\sim} N(\mu,\sigma^2)$, then the conditional distribution of $\mu|X_1,\ldots X_n$ is a normal with mean $\mu_p$ and variance $\sigma^2_p$ where
	$$\mu_p = \frac{n/\sigma^2}{n/\sigma^2 + 1/\gamma^2} \cdot \bar{X}$$\
	and 
	$$\sigma^2_p = \frac{1}{n/\sigma^2 + 1/\gamma^2}$$
\end{fact}

\begin{proposition}
	\label{prop:anti-concentration}
	Let $X \sim N(\mu,\sigma^2)$ and $I(X)$ be any measurable set, then $\mathbb P( X \in I) \leq 2\Phi\left( |I|/(2\sigma) \right) -1$, where $\Phi(\cdot)$ is the cdf of a standard normal distribution and $|I|$ is the Lebesgue measure of the set $I$.
\end{proposition}
\begin{proof}[Proof sketch]
Let $I$ be any measurable set.  Since the normal distribution is symmetric and unimodal at $\mu$, the density is highest at $\mu$ and decreasing away from $\mu$. Thus, probability mass in $I$ will be maximized when $I$ is an interval centered at $\mu$. Hence an interval $I$ with the highest mass is $[\mu - |I|/2, \mu + |I|/2]$. The result follows by computing the probability of this interval.	
\end{proof}

\begin{proposition}
\label{prop:zquantileinverse}
	Let $\alpha \in (0,1/2)$ and $0 < \Delta < \alpha$, then
	$$z_{1 - (\alpha + \Delta)/2} \geq z_{1 - \alpha/2} - \frac{c\Delta}{\alpha}$$
\end{proposition}
\begin{proof}
	Apply Proposition \ref{prop:asymptoticQuantile} with $\alpha' = \alpha + \Delta$.
\end{proof}

We are now ready to prove Theorem \ref{thm:lowerboundwithoutprivacy}. The key idea in proving this result is to use a normal prior distribution on $\mu$ and obtain a posterior distribution of $\mu$ given the data. For a normal prior, the posterior has a simple form given by Fact \ref{fact:posterior}. Then  using Proposition \ref{prop:anti-concentration}, we obtain an upper bound on how much posterior probability mass can be packed into any set $I$. On the other hand, since the set must have $1-\alpha$ coverage, we have a lower bound on the posterior probability mass contained inside $I$ whenever $\mu \in [-R,R]$. We have to eliminate the case when $\mu$ does not lie $[-R,R]$. Combining this with Jensen's inequality gives the result. 
\begin{proof}
	Let $\mu \sim N(0,R^2 \beta^2)$ where $\beta$ will be specified later.  Let $\underline{X}|\mu = (X_1, \ldots, X_n)|\mu \sim N(\mu,\sigma^2)$ and $I(\underline{X})$ be any subset of $(-R,R)$. With a minor abuse of notation, let $\mathbb P_{\underline{X}}(\cdot)$ denote the probability with respect to $\underline{X}$ and $\mathbb E_{\mu}[\cdot]$ denote the expectation taken with respect to the distribution of $\mu$. We have the following lower bound:
\begin{align}
		\Em[\mu \sim N(0,R^2\beta^2)][]{\prm[\underline{X} \sim N(\mu,\sigma^2)][]{\mu \in I} }
	\geq &\mathbb E_{\mu} \left[\mathbb P_{\underline{X}} \left( \mu \in I \right) \cdot \mathbb I(\mu \in (-R,R)) \right] 
    \nonumber \\
	\geq & (1-\alpha) \cdot  \mathbb P_{\mu}(\mu \in (-R,R) ) \nonumber \\
	\geq & (1-\alpha) \cdot \gamma, \label{lb}
\end{align}
where $\gamma = \mathbb P_{\mu}(\mu \in (-R,R))$. 
Next, using the Fact \ref{fact:posterior}, $\mu|\underline{X} \sim N(\mu_p, \sigma_p^2)$ where 
	$$\sigma_p = \frac{\sigma}{\sqrt{n}} \frac{1}{\sqrt{1+ \sigma^2/(nR^2 \beta^2)}},$$ 
	we have the following upper bound:
	\begin{align}
	\mathbb E_{\mu} \left[ \mathbb P_{\underline{X}} (\mu \in I) \right] 
	&= \mathbb E_{\underline{X}} \left[ \mathbb P_{\mu|\underline{X}}(\mu \in I) \right] \nonumber \\
	&\leq \mathbb E_{\underline{X}} \left[2\Phi\left(\frac{|I|}{2\sigma_p} - 1\right)\right]\mbox{ (By Proposition \ref{prop:anti-concentration})} \nonumber \\
	&\leq 2\Phi\left(\frac{E[|I|]}{2\sigma_p}\right) -1 \mbox{ (By Jensen's inequality applied to $\Phi(x)$ for $x >0$ } \label{ub}
	\end{align}
	Combining the upper bound and the lower bound from equations \ref{ub} and \ref{lb}, we get,
	\begin{align}
	(1-\alpha)\gamma \leq  2\Phi\left(\frac{E[|I|]}{2\sigma_p}\right) -1
	\implies \mathbb E[|I|] \geq \Phi^{-1}\left(\frac{\gamma(1-\alpha)+1}{2}\right) 2\sigma_p \label{eq:lblength}
	\end{align}
	Let us consider the term $\Phi^{-1}(\cdot)$ in the lower bound. Let $\beta = 1/\log(\sqrt{n})$. By the Gaussian tail bound in Proposition \ref{prop:gauss}, we have 
	$$\gamma \geq 1 - 2\exp\left(-1/(2\beta^2)\right) = 1 - \Omega(1/\sqrt{n}).$$
	Hence $\gamma(1-\alpha) = 1 - \alpha - \Omega(1/\sqrt{n})$ and from Proposition \ref{prop:zquantileinverse}, we have
	\begin{align*}
	\Phi^{-1}\left(\frac{\gamma(1-\alpha)+1}{2}\right) &\geq \Phi^{-1}\left(1 - \alpha/2 - \Omega(1/\sqrt{n})\right) \\
	& \geq \Phi^{-1}\left(1 - \alpha/2\right) - \frac{c/\sqrt{n}}{\alpha} \\
	&\geq  \Phi^{-1}\left(1 - \alpha/2\right) - \frac{c}{\alpha \sqrt{n}}
	\end{align*}
	Next, let us now consider the $\sigma_p$ term. We have,
	\begin{align*}
	\sigma_p &= \frac{\sigma}{\sqrt{n}} \cdot \frac{1}{\sqrt{1+ \sigma^2/(nR^2 \beta^2)}} \\
	&\geq \frac{\sigma}{\sqrt{n}}\left(1 - \Omega\left( \frac{\sigma^2}{n R^2 \beta^2}\right)\right) \\
	&\geq \frac{\sigma}{\sqrt{n}}\left(1 - \Omega\left( \frac{\sigma^2 \log n}{n R^2 }\right)\right) \\
	&\geq \frac{\sigma}{\sqrt{n}}\left(1 - \tilde{\Omega}\left( \frac{1}{n}\right)\right)
	\end{align*}
	when $R = \O{\sigma}$.
	Hence combining the two terms and applying them to equation \ref{eq:lblength}, we have
	\begin{align}
	\mathbb E[|I|] &\geq \frac{2\sigma}{\sqrt{n}} \cdot \left(1 - \tilde{\Omega}\left(\frac{1}{n}\right)\right) \cdot \left(\Phi^{-1}\left(1 - \frac{\alpha}{2} - \Omega\left(\frac{1}{n}\right)\right)\right)  \nonumber \\
	&\geq \frac{2\sigma}{\sqrt{n}}\Phi^{-1}\left(1 - \frac{\alpha}{2}\right) - \tilde{\Omega}\left(\frac{1}{n}\right)
	\end{align}
\end{proof}

\subsection{Lower bounds on confidence sets with privacy}
 
In this subsection, we derive lower bounds on the expected size of any $(1-\alpha)$-level $(\eps, \delta)$-differentially private confidence sets. The following theorem is the main result of this section.

\begin{theorem}
	\label{thm:lowerbound}
	 Let $M(X_1, \ldots, X_n)$ be any $(\epsilon, \delta)$-DP algorithm that outputs a random measurable set $S(X_1,\ldots, X_n) \subset (-R,R)$ such that, for every $\alpha \in (0,1)$, and $\mu \in (-R,R),$ whenever $X_1, \ldots, X_n \overset{iid}{\sim} N(\mu, \sigma^2)$ where $\mu \in (-R,R)$ and $\sigma^2$ is known, if
	 \begin{enumerate}
      \item $\prm[\underline{X} \sim N(\mu,\sigma^2)][M]{S(X_1,\ldots, X_n) \ni \mu} \geq 1 - \alpha.$
      \item $\Em[\underline{X} \sim N(\mu,\sigma^2)][M]{|S(X_1,\ldots, X_n)|} \leq \beta$
	 \end{enumerate}
	then, 
	\begin{enumerate}
		\item For all $\alpha \in \left(0,\frac{1}{2}\right)$, and $\delta < \alpha/2n$,
		\begin{align*}
		\beta & \geq c \cdot \min \left( \frac{\sigma}{\eps n} \log \left(\frac{1}{\alpha}\right), R \right).
		\end{align*}
		\item If $\beta < \sigma < R$, then, 
		$$n \geq c_1 \cdot  \min\left(\frac{1}{\eps}\log \left(\frac{1}{\alpha}\right), \frac{1}{\eps}\log\left(\frac{1}{\delta}\right)\right)$$ and
		$$n \geq c_2 \cdot \min\left(\frac{1}{\eps}\log \left(\frac{R}{\sigma}\right), \frac{1}{\eps}\log\left(\frac{1}{\delta}\right)\right).
		$$	
	\end{enumerate}
	
\end{theorem}

Theorem \ref{thm:lowerbound} gives a lower bound on the expected size of any $(1-\alpha)$-level confidence set produced by a differentially private algorithm. Note that the lower bounds apply to any set and not just intervals.

We will now compare the lower bounds with the upper bounds obtained in the previous sections. Let us focus on the case of known variance. Theorem \ref{thm:ciknownvar} asserts the existence of a differentially private algorithm that outputs a $(1-\alpha)$-level confidence interval of length $\beta$ such that if
	$$
	n \geq   \frac{c}{\eps} \min \left\{\log \left(\frac{R}{\sigma}\right), \log \left(\frac{1}{\delta}\right) \right\} + \frac{1}{\eps}\log\left(\frac{1}{\alpha}\right),
	$$
	then,
	\begin{align*}
	\beta &\leq \max \left\{
	\frac{\sigma}{\sqrt{n}}  \O{\sqrt{\log \left(\frac{1}{\alpha}\right)}},
	\frac{\sigma}{\epsilon n} \mathrm{polylog}\left(\frac{n}{\alpha}\right)
	\right\}.
	\end{align*}
Note that the upper bound on the width of $\beta$ is a maximum two terms. The first term $\O{(\sigma/\sqrt{n}) \cdot \sqrt{\log(1/\alpha)}}$ is needed even without privacy, since it is the optimal width of the interval without privacy, as shown in Theorem \ref{thm:lowerboundwithoutprivacy}. The second term is $\sigma/(\eps n)\cdot \log \left(1/\alpha\right)$ upto polylog$(n,1/\alpha)$ factors. The lower bound in Part (1) of Theorem \ref{thm:lowerbound} shows that this term is unavoidable, unless we produce a very large interval of length proportional to the original range $(-R,R)$ of $\mu$. 

Part (2) of Theorem \ref{thm:lowerbound} gives a lower bound on $n$ if we require $\beta$ to be smaller than $\sigma$ and $R$. Note that the lower bound on the sample size matches the sample size requirement of Theorem \ref{thm:ciknownvar} to obtain a non-trivial confidence interval. In the unknown variance case, the lower bounds of Theorem \ref{thm:lowerbound} apply. We are interested in the regime when $R \geq \Omega(\sigma_{\max})$. In this case the known-variance lower bound on sample complexity for $\sigma=\sigma_{\min}$ matches our unknown-variance upper bound given in Theorem \ref{thm:ciunknownvar}.


Before we present a proof of Theorem \ref{thm:lowerbound}, we will introduce some additional notation for the proof. Let $X_1, \ldots, X_n$ be iid samples from a distribution $\mathbb D_{(\theta, \gamma)}$. As before, we abuse the notation and use $\mathbb D_{\theta}$ to denote the distribution when $\gamma$ is a nuisance parameter. Let $M(x_1, \ldots, x_n)$ be an $(\epsilon, \delta)$-differentially private mechanism that runs on the realization $x_1, \ldots, x_n$ and  let $Q(E|X_1 = x_1, \ldots X_n = x_n)$ denote the conditional distribution of $M$ given $x_1,\ldots, x_n$, where $E$ is any event. For any event $E$, let
\begin{align}
\label{eq:marginals}
\mathbb M_{\theta}(E) = \int{Q(E|X_1,\ldots, X_n)d\mathbb P_{\theta}(X_1, \ldots, X_n)}
\end{align} 
be the marginal distribution on the outputs induced by the DP mechanism when the data are generated from the distribution $\mathbb{D}_{\theta}$; here $\mathbb P_{\theta}$ is short hand to denote the probability computed under the distribution $\mathbb D_{\theta}$. Similalry $\mathbb M_{\theta}$ is a shorthand to denote the probability computed uner the distribution $\mathbb D_{\theta}$ and the mechansim $M$. Let $\mathbb D_{\theta_0}$ and $\mathbb D_{\theta_1}$ be two distributions, and denote the total variation distance between $\mathbb D_{\theta_0}$ and $\mathbb D_{\theta_1}$ by $||\mathbb D_{\theta_0} - \mathbb D_{\theta_1}||_{tv}$. 

To prove Theorem \ref{thm:lowerbound} 
we need a lemma that bounds the multiplicative distances between two differentially private  marginal distributions induced by two different distributions $\mathbb D_{\theta_0}$ and $\mathbb D_{\theta_1}$ on the data. This is done in Lemma \ref{lemma:di} below. We built on Lemma 4.12 of \cite{BNSV} where they use a ``secrecy of sampling argument'' that appeared in \cite{kasiviswanathan2011can} and \cite{smith2009differential} to amplify privacy when a differentially private algorithm is run on a re-sampled version of a dataset $\underline{X}$. We use a similar technique to compute the privacy amplification for datasets generated by two different distributions. Lemma 4.12 of \cite{BNSV} can be recovered from our result by letting $\mathbb D_{\theta_j}$ be an empirical cdf of $\underline{X}$ and $\underline{X}'$ of two neighboring datasets.

\begin{lemma}
\label{lemma:di}
For every pair of distributions $\mathbb D_{\theta_0}$ and $\mathbb D_{\theta_1}$, every $(\eps, \delta)$-differentially private algorithm $M(x_1, \ldots, x_n)$, if $\mathbb M_{\theta_0}$ and $\mathbb M_{\theta_1}$ are two induced marginal distributions on the output of $M$ evaluated on input dataset $X_1,\ldots,X_n$ sampled iid from $\mathbb D_{\theta_0}$ and $\mathbb D_{\theta_1}$ respectively, $\epsilon' = 6\epsilon n ||\mathbb D_{\theta_0} - \mathbb D_{\theta_1}||_{tv}$ and 
	$\delta' = 4e^{\eps'}n\delta || \mathbb D_{\theta_0} - \mathbb D_{\theta_1}||_{tv}$, then, for every event $E$,
	$$\mathbb{M}_{\theta_0}(E) \leq e^{\epsilon'} \mathbb{M}_{\theta_1}( E) + \delta'.$$  
\end{lemma}

\begin{proof}
 We will first construct a coupling between $\mathbb D_{\theta_0}$ and $\mathbb D_{\theta_1}$ that allows us to control the hamming distance between iid samples generated from $\mathbb D_{\theta_0}$ and $\mathbb D_{\theta_1}$. Let us start with some notation. Let $p = ||\mathbb D_{\theta_0} - \mathbb D_{\theta_1}||_{tv}$, $F = \max(\mathbb D_{\theta_0} - \mathbb D_{\theta_1}, 0)$, $G = \max(\mathbb D_{\theta_1} - \mathbb D_{\theta_0}, 0)$, $C = \min (\mathbb D_{\theta_0}, \mathbb D_{\theta_1})$. It is easy to see that $\mathbb D_{\theta_0} = F+C$ and $\mathbb D_{\theta_1} = G+C$. Consider the following algorithm to generate $2n$ samples:
 \begin{enumerate}
 	\item Generate $H_1, \ldots, H_n$ iid $\text{Bernoulli}(p)$ and let $H = \sum_{i=1}^n H_n$.
 	\item For $i=1$ to $n$,
 	\begin{enumerate}
		\item If $H_i = 1$, sample $X_i \propto F$ and $X_i' \propto G$
		\item If $H_i= 0$, sample $X_i \propto C$ and set $X_i' = X_i$.
 	\end{enumerate}
 \end{enumerate}
 Here $X \propto F$ means that $X$ is generated from a distribution defined by normalizing $F$.
 Under this construction, one can verify the following:
 \begin{enumerate}
	\item $\underline{X}:= (X_1, \ldots, X_n) \overset{iid}{\sim} \mathbb D_{\theta_0}$
	\item $\underline{X}':= (X_1', \ldots, X_n') \overset{iid}{\sim} \mathbb D_{\theta_1}$
	\item $d(\underline{X},\underline{X'}) = H$.
 \end{enumerate}
Fix any event $E$. Let $$q_{\theta_j}(h) =  \mathbb P_{\theta_j}(E|H=h) := \int_{\underline{x}} \mathbb Q(E|\underline{X} = \underline{x})d\mathbb P_{\theta_j}(\underline{X}|H=h)$$ 
for $j \in \{0,1\}$ and 
 $$p(h) = \mathbb P(H=h) = {n \choose h} p^h (1-p)^{n-h},$$
 since $H \sim Binomial(n,p)$.
For $j \in \{0,1\}$, we have, by definition,
 $$\mathbb M_{\theta_j}(E) =
 \sum_{h=0}^n q_{\theta_j}(h) p(h) 
 $$ 
 \paragraph{Claim 1:} For $j \in \{0,1\}$, $q_{\theta_j}(h) \leq e^{\eps} q_{\theta_j}(h-1) + \delta$ for $h = 1, \ldots, n$, and $q_{\theta_1}(0) = q_{\theta_0}(0)$. \\
We will defer the proof of Claim 1 towards the end. 
\noindent By Claim 1, for $\theta_j \in \{0,1\}$, we have
\begin{align}
\label{eq:q}
q_{\theta_j}(h) \leq e^{h\eps}q_{\theta_j}(0) + \frac{e^{h\eps} - 1}{e^{\eps}-1}\delta
\end{align}

In the following, we will use the standard fact of the moment generating function of a Binomial distribution:
\begin{fact}
\label{fact:Binommgf}
If $H \sim$ Binomial$(n,p)$ then, for $any t> 0$, 
$$\E{e^{tH}} = (1-p + p\cdot e^t)^n$$
\end{fact}
Consider,
\begin{align}
 \label{eq:up}
\mathbb M_{\theta_j}(E) &=\sum_{h=0}^n p(h)q_{\theta_j}(h) = \E{q_{\theta_j}(H)}  \nonumber \\ 
&\leq \E{  e^{H\eps}q_{\theta_j}(0) + \frac{e^{H\eps} - 1}{e^{\eps}-1}\delta} \mbox{(From equation \ref{eq:q})} \nonumber \\ 
&= q_{\theta_j}(0) \cdot \E{e^{H\eps}} + \frac{\delta}{e^{\eps}-1} \cdot \left(\E{e^{H\eps}} - 1\right)  \nonumber \\
&= q_{\theta_j}(0) \cdot \left( 1 - p + p \cdot e^{\eps}\right)^n + \frac{\delta}{e^{\eps}-1} \cdot \left(\left(1 - p + p \cdot e^{\eps}\right)^n - 1\right)
\end{align}
Similarly, we can show that,
\begin{align}
\label{eq:lb}
\mathbb M_{\theta_j}(E) 
\geq q_{\theta_j}(0) \left(1 - p + p \cdot e^{-\eps}\right)^n + \frac{\delta}{e^{-\eps}-1} \cdot \left((1-p+pe^{-\eps})^n-1 \right) 
\end{align}
Combining inequalities \ref{eq:up} and \ref{eq:lb}, we get:
\begin{align}
\mathbb M_{\theta_0}(E) \leq \left( \frac{1 - p + p \cdot e^{\eps}}{1 - p + p \cdot e^{-\eps}} \right)^n \cdot \left( \mathbb M_{\theta_1}(E)  + \frac{1 - \left(1 - p + p \cdot e^{-\eps} \right)^n }{1 - e^{-\eps} } \cdot \delta \right) + \frac{\left(1 - p + p \cdot e^{\eps} \right)^n -1 }{e^{\eps}-1 } \cdot \delta
\end{align}
Finally, we have, 
$$
 n \log \left( \frac{1 - p + p \cdot e^{\eps}}{1 - p + p \cdot e^{-\eps}} \right) \leq n \cdot 6\eps p
$$
and 
\begin{align*}
	& e^{6 \eps n \cdot p} \cdot \frac{1 - (1+ p\cdot (e^{\eps} - 1))^n }{1 - e^{-\eps} } \cdot \delta + \frac{(1 + p\cdot (e^{\eps}-1) )^n - 1}{e^{\eps} - 1} \cdot \delta \\
\leq & e^{6 \eps n \cdot p} \cdot \frac{1 - \exp(2np \cdot (e^{-\eps}-1)) }{1 - e^{-\eps} } \cdot \delta + \frac{\exp(2np \cdot (e^{\eps}-1)) - 1}{e^{\eps} - 1} \cdot \delta \\
& \leq e^{6 \eps n \cdot p} \cdot 2np \cdot \delta + 2n p \cdot \delta\\
&\leq e^{6 \eps n \cdot p} \cdot 4np \cdot \delta,
\end{align*}
which shows that
$$\mathbb M_{\theta_0}(E) \leq e^{\eps'}\mathbb M_{\theta_1}(E) + \delta'$$
for $\eps' = 6np \cdot \eps$ and $\delta' = 4e^{\eps'} n p \cdot \delta$.
We will now prove Claim 1.
\begin{proof}[Proof of Claim 1]
	We will prove the claim for $j = 0$, the other case is similar. Let us introduce some notation. Fix a $(h_1, \ldots, h_n) \in \{0,1\}^n$. Let $I' = \{i: h_i=1\}$, $J = \{i:h_i=0\}$ and let $r$ be any index in $I'$.  Let $I = I'/\{r\}$ and consider the following partition of $\underline{X}$ into three parts:
	$$\underline{X} = (\underline{X}_{I}, X_{r},\underline{X}_{J} )$$ where 
	$\underline{X}_{I}$ is the vector $\underline{X}$ subsetted by the indices in $I$.
	By definition of the coupling, $\underline{X}_{I} \overset{iid}{\sim} F$, $X_{r} \sim F$, $\underline{X}_{J} \overset{iid}{\sim} C$.
	Now let $X'_r \sim C$ and let 
	$$\underline{X}' = (\underline{X}_{I}, X'_r,\underline{X}_J ).$$
	Also, let $h_1',\ldots, h_n'$ be the binary indicators corresponding to $\underline{X}$.
	By construction, we have the following:
	\begin{enumerate}
		\item $h_i = h_i'$ for all $i \neq r$
		\item $h_r = 1$ and $h_r = 0$
		\item $\sum_{i=1}^n h_i = h$ and $\sum_{i=1}^n h_i' = h-1$
		\item $\mathbb P_{\theta_j}(\underline{X}|H_1=h_1,\ldots, H_n=h_n) = \mathbb P_F(\underline{X}_I) \mathbb P_F(X_r) \mathbb P_C(\underline{X}_J)$
		\item $\mathbb P_{\theta_j}(\underline{X}'|H_1=h'_1,\ldots, H_n=h'_n) = \mathbb P_F(\underline{X}_I) \mathbb P_F(X'_r) \mathbb P_C(\underline{X}_J)$  
	\end{enumerate}
	
	Now consider the following:
	\begin{align*}
    &\mathbb P_{\theta_j}(E|H_1=h_1,\ldots, H_n = h_n) \\
	=& \int_{\underline{x}} \mathbb Q(E|\underline{X} = \underline{x})d\mathbb P_{\theta_j}(\underline{X}|H_1=h_1, \ldots H_n=h_n) \\
	=& \int_{\underline{x}_I} \int_{x_r}\int_{\underline{x}_J}  \mathbb Q(E|\underline{x}_I,x_r, \underline{x}_J)
	d\mathbb P_{F}(\underline{X}_I)d\mathbb P_F(X_r) \mathbb P_C (\underline{X}_J) \\ 	 
	\leq& \int_{\underline{x}_I} \int_{x_r}\int_{\underline{x}_J}  \left(e^{\eps}\mathbb Q(E|\underline{x}_I,x'_r, \underline{x}_J) +\delta \right)
	d\mathbb P_{F}(\underline{X}_I)d\mathbb P_F(X_r) \mathbb P_C (\underline{X}_J) \\ 	 
	\leq& \int_{\underline{x}_I}\int_{\underline{x}_J}  \left(e^{\eps}\mathbb Q(E|\underline{x}_I,x'_r, \underline{x}_J) +\delta \right)
	d\mathbb P_{F}(\underline{X}_I) d\mathbb P_C (\underline{X}_J) \\ 	 
	\leq& \int_{\underline{x}_I} \int_{x_r'}\int_{\underline{x}_J}  \left(e^{\eps}\mathbb Q(E|\underline{x}_I,x'_r, \underline{x}_J) +\delta \right)
	d\mathbb P_{F}(\underline{X}_I)d\mathbb P_C(X_r') \mathbb P_C (\underline{X}_J) \\ 	 	 	 
	\leq& \int_{\underline{x}} \left(e^{\eps}\mathbb Q(E|\underline{x}') +\delta \right)
	d\mathbb P_{\theta_j}(\underline{X}|H_1=h_1', \ldots, H_n=h_n') \\
	\leq& e^{\eps} \mathbb P_{\theta_j}(E|H_1=h'_1,\ldots, H_n = h'_n) + \delta.
	\end{align*}
	Hence we have shown that
	$$
	\mathbb P_{\theta_j}(E|H_1=h_1,\ldots, H_n = h_n, H=h)
	\leq e^{\eps} \mathbb P_{\theta_j}(E|H_1=h'_1,\ldots, H_n = h'_n, H = h-1) + \delta
	$$
	Taking expectations on both sides with respect to $(H_1, \ldots, H_n)$, we get
	$$
	\mathbb P_{\theta_j}(E|H=h)
	\leq e^{\eps} \mathbb P_{\theta_j}(E| H = h-1) + \delta
	$$
which proves the claim.	
\end{proof}

\end{proof}

Let us consider a corollary of Lemma \ref{lemma:di} when $\mathbb P_{\theta_0}$ and $\mathbb P_{\theta_1}$ are two normal distributions with $\theta_0 = (\mu_0,\sigma^2)$ and $\theta_1 = (\mu_1,\sigma^2)$. Since the total variation distance between $N(\mu_0,\sigma^2)$ and $N(\mu_1, \sigma^2)$ is upper bounded by $|\mu_0 - \mu_1|/\sigma$ (see for example \cite{dasgupta2008asymptotic}), we get the following Corollary:
\begin{corollary}
	\label{cor:normalLemma}
	Let $\mathbb P_{\theta_0}$ and $\mathbb P_{\theta_1}$ be two normal distributions where $\theta_0 = (\mu_0,\sigma^2)$ and $\theta_1 = (\mu_1,\sigma^2)$, where $\mu_1$ and $\mu_0$ are the means $\mu_0$ and $\sigma^2$ is the variance. Let $M(x_1, \ldots, x_n)$ be any $(\eps,\delta)$-differentially private algorithm that induces the marginal distribution $\mathbb M_{\theta_0}$ and $\mathbb M_{\theta_1}$ as defined in equation \ref{eq:marginals}. For any event $E$, we have
	$$\mathbb M_{\theta_0}(E) \leq e^{6\eps n \cdot k} \left(\mathbb M_{\theta_1}(E) + 4n\delta \cdot k   \right),
	$$
	where $$k = \min\left\{\frac{|\mu_0-\mu_1|}{\sigma},1\right\}.$$
\end{corollary}
We are now ready to prove Theorem \ref{thm:lowerbound}.
\begin{proof}[Proof of Theorem \ref{thm:lowerbound}]
	By definition of measure of a set $S$, 
	$$
	|S| = \int_{-R}^{R} \mathbb I(\mu \in S)d\mu,
	$$
	where $\mathbb I(\cdot)$ is the indicator function. Hence, taking expectation on both sides with respect to $S \rightarrow M(X_1,\ldots, X_n)$ and $X_1,\ldots, X_n \sim N(\mu_0,\sigma^2)$, and changing the order of integration, (since $\mathbb I(\cdot)$ is non-negative, one can apply Tonelli's theorem), we get
	$$\mathbb E_{\mu_0} [|S|] = \int_{-R}^{R}\mathbb M_{\mu_0}\left( \mu \in S \right) d\mu.$$
	We will start by proving the first bound in Part (1). Let $\mu_0$ and $\mu_1$ be any two points in $[-R, R]$. 
	Note that, $\forall \mu \in [\mu_0, \mu_1]$, we have,
	\begin{align}
	\label{eq:thm:lowerbound:1}
	\exp\left( 6\eps n \frac{|\mu_0 - \mu|}{\sigma} \right) \leq  \exp\left( 6\eps n \frac{|\mu_0 - \mu_1|}{\sigma} \right)
	\end{align} 
	and by Corollary \ref{cor:normalLemma}, we have,
	\begin{align}
	\label{eq:thm:lowerbound:2}
	\mathbb M_{\mu_0}(\mu \notin S) &\leq e^{6\eps n \cdot k} \left(\mathbb M_{\mu}(\mu \notin S ) + 4n\delta \cdot k\right) \nonumber \\
	&\leq e^{6\eps n \cdot |\mu_0 - \mu|/\sigma} \left( \mathbb M_{\mu}(\mu \notin S ) + 4n\delta \cdot 1 \right)
	\end{align}
	where $k = \min\{1,|\mu_0 - \mu|/\sigma\} $ and we have used an upper bound of $1$ for $k$ in the second term.
	Then we have,
	\begin{align}
		\mathbb E_{\mu_0}[|S|] & \geq \int_{\mu_0}^{\mu_1}\mathbb M_{\mu_0}\left( \mu \in S \right) d\mu  \nonumber \\
		&\geq \int_{\mu_0}^{\mu_1}\left(1 - \mathbb M_{\mu_0}(\mu \notin S) \right) d\mu \nonumber \\
		&\geq |\mu_1 - \mu_0| - \int_{\mu_0}^{\mu_1} \left(\mathbb M_{\mu}(\mu \notin S ) + 4n\delta \right)\cdot e^{6\eps n \cdot |\mu_0 - \mu|/\sigma} d\mu, \text{ (Equation \ref{eq:thm:lowerbound:2})} \nonumber \\
		&\geq |\mu_1 - \mu_0| - (\alpha + 4 n\delta) \cdot \int_{\mu_0}^{\mu_1} e^{6\eps n |\mu_0 - \mu_1|/\sigma} d\mu \text{ (Equation \ref{eq:thm:lowerbound:1})} \nonumber \\
		&\geq |\mu_1 - \mu_0| - (\alpha  + 4n\delta) \cdot |\mu_1 - \mu_0|  \cdot e^{6\eps n |\mu_0 - \mu_1| /\sigma}  \nonumber \\
		&= |\mu_1 - \mu_0| \cdot \left( 1 - (\alpha + 4n\delta) \cdot e^{6\eps n |\mu_0 - \mu_1|/\sigma}\right) \label{eq:thm:lowerbound:3}
	\end{align}
	We will consider two cases depending on how large $R$ is relative to $\sigma/(\eps n)$:
	\paragraph{Case 1: $2R \geq (\sigma/(6\eps n)) \cdot \log(1/(4\alpha)).$} 
	
	Let $\mu_0$ and $\mu_1$ be two points such that 
	$$|\mu_1 - \mu_0| = \frac{\sigma}{6\eps n} \cdot \log \left(\frac{1}{4\alpha}\right).$$ 
	By the assumption on $R$ two such points always exist. Substituting this in equation \ref{eq:thm:lowerbound:3}, we get,
\begin{align*}
	\mathbb E_{\mu_0}[|S|] 
	&\geq |\mu_1 - \mu_0| \cdot \left( 1 - ( \alpha + 4n\delta) \cdot e^{6\eps n |\mu_0 - \mu_1|/\sigma}  \right)\\
	&\geq \frac{\sigma}{6 \eps n} \cdot \log \left(\frac{1}{4\alpha}\right) \left( 1 - (\alpha + 4n\delta) \cdot \left(\frac{1}{4\alpha}\right) \right)\\
	&\geq \frac{\sigma}{6 \eps n} \cdot \log \left(\frac{1}{4\alpha}\right) \left( \frac{3}{4} - \frac{n\delta}{\alpha} \right)\\
	&\geq \frac{\sigma}{24 \eps n} \cdot \log \left(\frac{1}{4\alpha}\right) \mbox{ (Since $\delta < \alpha/(2n)$ which implies $3/4-(n\delta/\alpha) > 1/4$)}
\end{align*}
	
\paragraph{Case 2: $2R < \sigma/(6\eps n) \cdot \log(1/(4\alpha)) < \infty.$} 
	Let $\mu_0 = -R$ and $\mu_1 = R$. Substituting this in equation \ref{eq:thm:lowerbound:3}, we get:
	\begin{align*}
	\mathbb E_{\mu_0}[|S|] 
	&\geq |\mu_1 - \mu_0| \cdot \left( 1 - ( \alpha + 4n\delta) \cdot e^{6\eps n |\mu_0 - \mu_1|/\sigma}  \right)\\
	&= 2R \cdot \left( 1 - ( \alpha + 4n\delta) \cdot e^{12\eps n R/\sigma}  \right)\\
	&\geq 2R \cdot \left( 1 - (\alpha + 4n\delta) \left(\frac{1}{4\alpha} \right)\right) \\
	&\geq 2R \left( \frac{3}{4} - \frac{n\delta}{\alpha} \right)\\
	&\geq \frac{R}{2} \mbox{ (Since $\delta < \alpha/(2n)$ which implies $3/4-(n\delta/\alpha) > 1/4$)}
\end{align*}
Combining the two cases , we have $\beta > R/2$ or $\beta > \sigma /(24 \eps n) \cdot \log (1/(4\alpha))$ which proves the bound in Part (1).
	
%
\noindent Let us now prove the bounds in Part (2). 
	We will first show,
	$$ n \geq c\cdot \min \left\{\frac{1}{\eps }\log \left(\frac{R}{\sigma}\right), \frac{1}{\eps}\log \left(\frac{1}{\delta}\right) \right\}$$
	By assumption, since $\beta < \sigma < R$, we have,
	\begin{align*}
	\sigma > \beta & \geq \int_{-R}^{R} \mathbb M_{\mu_0}(\mu \in S) d\mu \\ 
	& \geq \int_{-R}^{R} \left(e^{-6\eps n} \cdot \mathbb M_{\mu} (\mu \in S) - 4n\delta\right) d\mu \text{ (By Corollary \ref{cor:normalLemma})}\\
	& \geq \int_{-R}^{R} \left(e^{-6\eps n} \cdot (1-\alpha) - 4n\delta \right) d\mu \\
	& \geq 2R\cdot (1-\alpha) \cdot e^{-6\eps n} - 4n\delta \cdot (2R) \\
	& \geq R e^{-6\eps n} - 4n\delta \cdot (2R)
	\end{align*}
	Hence we have, either
	$$\sigma > \frac{R}{2}e^{-6\eps n}  \text{ or } 4n\delta \cdot (2R) > \frac{R}{2} \cdot e^{-6 \eps n}$$
	which implies that 
	$$ n \geq \frac{c_1}{ \eps} \log \left(\frac{R}{\sigma}\right) \text{ or } n \geq \frac{c_2}{\eps} \log \left(\frac{1}{\delta}\right).$$
	Next we will show that 
	$$ n \geq c\cdot \min \left\{\frac{1}{\eps }\log \left(\frac{1}{\alpha}\right), \frac{1}{\eps}\log \left(\frac{1}{\delta}\right) \right\}$$
	Let $\mu_0$ be a fixed point in $(-R,R)$. The fact that if $ \mathbb E_{\mu_0}(|S|) < R$, implies that there exists a $\mu_1 \in (-R,R)$ such that 
	$$\mathbb M_{\mu_0} \left(\mu_1 \notin S\right) > \frac{1}{2}$$
	Indeed, if not, then 
	$$\mathbb E_{\mu_0}(|S|) = \int_{-R}^{R}\mathbb M_{\mu_0} \left(\mu \in S\right)d \mu \geq \frac{1}{2}2R = R.$$ 
	Thus from Corollary \ref{cor:normalLemma} we have
	\begin{align*}
	\frac{1}{2} &< \mathbb M_{\mu_0} \left(\mu_1 \notin S\right) < e^{6\eps n \cdot |\mu_0 - \mu_1|/\sigma} \cdot \mathbb M_{\mu_1} \left(\mu_1 \notin S\right) + 4e^{6\eps n \cdot |\mu_0-\mu_1|/\sigma} \cdot n\delta \cdot  ||\mathbb P_{\mu_1} - \mathbb P_{\mu_0}||_{tv} \\
	& \leq (\alpha + 4n\delta) \cdot e^{6\eps n \frac{|\mu_0 - \mu_1|}{\sigma}} \leq (\alpha + 4n\delta) \cdot e^{6\eps n}.
	\end{align*}
	 This gives us
	 $$\alpha > \frac{1}{2}e^{-6\eps n} - 4\delta n$$
	 Hence either $4n \delta > e^{-6\eps n}/4$ or $\alpha > e^{-6\eps n} /4$, which gives us
	 either
	$n \geq (c_1/\eps) \cdot \log(1/4\alpha)$ or $ n \geq (c_2/\eps) \cdot \log(1/\delta)$.
\end{proof}

\section{Appendix: Basic distributions and tail bounds}

We make use of several standard distributions and tail bounds that are defined here for completeness.

\subsection{Basic Distributions}
\label{sec:defs}
\begin{definition}[Normal or Gaussian distribution]
	A random variable $X$ has a normal distribution with mean $\mu$ and variance $\sigma^2$ if it's probability  density function is given by
	$$\frac{1}{\sqrt{2\pi \sigma^2}} e^{-(x-\mu)^2/(2\sigma^2)}$$
\end{definition}

A standard Gaussian distribution is a Gaussian distribution with mean $0$ and variance $1$. The cumulative distribution function (cdf) of a standard normal distribution is denoted by $\Phi(\cdot)$ and the pdf is denoted by $\phi(\cdot)$.

\begin{definition}[$\chi^2$-distribution]
	A random variable $X$ has a $\chi^2$ distribution with $k$ degrees of freedom if it can be written as the sum of squares of $k$  standard normal distributions, i.e. $X = \sum_{i=1}^k Z_i^2$ where each $Z_i$ is a standard normal random variable.
\end{definition}

\begin{definition}[$t$-distribution]
	A random variable $T_k$ has a $t$-distribution with $k$ degrees of freedom, if it can be represented as follows: $T_k = Z/(\sqrt{Y/k})$ where $Z$ is a standard normal distribution and $Y$ is a $\chi^2$-distribution with $k$ degrees of freedom. The density of a $t$-distribution is given by
	$$\frac{\Gamma\left(\frac{k+1}{2}\right)}{\sqrt{k \pi} \Gamma \left(\frac{k}{2}\right)} \left(1 + \frac{t^2}{k}\right)^{-\frac{k+1}{2}}$$	
\end{definition}

\begin{definition}[Laplace distribution]
\label{def:lap}
	A random variable $X$ has a  Laplace distribution with mean $0$ and scale parameter $b$ if it's density is given by
	$$\frac{1}{2b}e^{-\frac{|x|}b}$$
\end{definition}

\subsection{Tail bounds}
We make use of several tail bounds, some of which are well known and listed without proof and others are proved for completeness.
\begin{proposition}[Chernoff Bounds]
	\label{prop:Chernoff}
	Let $X_1, \ldots, X_n$ be independent random variables such that  $X_i \in \{0,1\}$, $\mu = \sum_i \E{X_i}$ and $X = \sum_iX_i$. Then for every $\delta \in (0,1)$,
	\begin{align*}
	\pr{X \leq (1-\delta)\mu} &\leq \exp\left(-\frac{ \delta^2}{2} \mu\right)\\
	\pr{X \geq (1+\delta)\mu} &\leq \exp\left(-\frac{ \delta^2}{3} \mu\right).
	\end{align*}
\end{proposition}

\begin{proposition}[Laplace tail bound]
	\label{prop:lap}
	Let $Z$ be a Laplace random variable with mean $0$ and scale $b$ 
	then, for every $t > 0$,
	\begin{align*}
	\pr{Z > t}  &= \frac{1}{2}\exp\left(-\frac{t}{b}\right) \mbox{ and}\\
	\pr{|Z| > t} &= \exp\left(-\frac{t}{b}\right). 
	\end{align*}
\end{proposition}

\begin{proposition}[Gaussian tail bound]
	\label{prop:gauss}
	Let $Z$ be a standard normal random variable with mean $0$ and variance $1$, then, for every $t > 0$, we have
	\begin{align*}
	\pr{|Z| > t}  &\leq 2\exp\left(-t^2/2\right).
	\end{align*}
\end{proposition}

\begin{proposition}[$\chi^2$ tail bound, Lemma 1 in \cite{laurent2000adaptive}]
	\label{prop:chi-square}
	Let $Y$ be a $\chi^2$ random variable with $n$ degrees of freedom, then for every $x>0$ we have, 
	$$\pr{Y \leq n - 2\sqrt{nx}} \leq \exp(-x).$$
\end{proposition}

\begin{proposition}[A tail bound on $t$-distribution from \cite{soms1976asymptotic}]
	\label{prop:t-tail}
		Let $T_n$ be a $t$-distribution with $n$ degrees of freedom, then for every $t>0$, we have, 
		$$\pr{T_n \geq t} \leq \frac{1}{t}\left(1 + \frac{t^2}{n}\right)f_{T_n}(t)$$ 
		where $f_{T_n}(.)$ is the density of a $t$-distribution.
\end{proposition}

\begin{proposition}[Relation between quantile and tail bound]
	\label{prop:quantileBound}
	Let $X$ be a random variable with invertible cdf $F(.)$ and let $q(\alpha) = F^{-1}(1-\alpha)$. If  $\pr{X\geq t} \leq \alpha$, then $q(\alpha) \leq t$.
\end{proposition}
\begin{proof}
	By definition, $F(q(\alpha)) = 1 - \alpha$ and it is given that $F(t) \geq 1 - \alpha$. Since $F(.)$ is a non-decreasing function, we have $t \geq q(\alpha)$.
\end{proof}

\begin{proposition}[Quantiles of a t-distribution]
	\label{prop:tailBoundT}
	Let $T_n$ be a $t$-distribution with $n$ degrees of freedom and $t_{n,\alpha}$ be the $1 - \alpha$ quantile of $T_n$. Then, 
	\begin{enumerate}
		\item  If $t_{n,\alpha}>1$, 
		$$t_{n,\alpha}^2 \leq n\left(\sqrt{\frac{2\pi n}{n+1}}\left(\frac{1}{\alpha}\right)^{2/(n-1)} - 1\right).$$ 
		\item If $n \geq (64/9)\log(2/\alpha)$, then $t_{n,\alpha} \leq \sqrt{8\log(2/\alpha)}$.
	\end{enumerate}
\end{proposition}
\begin{proof}
	In both the cases, the upper bound on the quantile follows from application of Proposition \ref{prop:quantileBound} and a tail bound on $T_n$. Thus, we only need to prove the tail bound on $T_n$ in each case.
	\paragraph{Part 1:} We will prove the following tail bound:
	$$
	\pr{T_n \geq t} \leq  \frac{1}{\sqrt{2\pi}}\sqrt{\frac{n+1}{n}}\left(\frac{1}{t}\right)\left(1+\frac{t^2}{n}\right)^{-(n-1)/2}
	$$
	From Proposition \ref{prop:t-tail}, we have, $\pr{T_n \geq t} \leq \frac{1}{t}\left(1 + \frac{t^2}{n}\right)f_{T_n}(t)$ where $f_{T_n}(.)$ is the density of a $t$-distribution.
	Hence we have,
	\begin{align*}
	\pr{T_n \geq t} \leq \frac{\Gamma\left(\frac{n+1}{2}\right)}{\Gamma\left(\frac{n}{2}\right)} \sqrt{\frac{1}{\pi n}}\cdot \frac{1}{t}\left(1 + \frac{t^2}{n}\right)^{-\frac{n-1}{2}}
	\end{align*}
	Using the fact that for $x > 0$ and $0 < s < 1$,
	$$ x^{1-s} < \frac{\Gamma\left(x+1\right)}{\Gamma\left(x+s\right)} \leq (x+1)^{1-s}$$
	and letting $x = (n-1)/2$ and $s = 1/2$, gives the desired result.
	\paragraph{Part 2:} We will prove the following tail bound:
	$\pr{T_n \geq \sqrt{8\log(2/\alpha)} } \leq \alpha$
	Recall that $T_n = Z/(\sqrt{Y/n})$ where $Z$ is a standard normal distribution and $Y$ is a $\chi^2$-distribution with $n$ degrees of freedom. Also, if
	$\{Z < t/2\}$ and $\{Y > n/4\}$, then $T_n < t$. Hence $T_n \geq t$ implies that either $Z \geq t/4$ or $Y \leq n/4$. By a union bound, we have:
	\begin{align*}
	\pr{T_n \geq t} &\leq \pr{Z \geq t/2}  + \pr{Y \leq n/4} \\
	&\leq \exp\left(-t^2/8\right) + \exp\left(-9n/64\right),
	\end{align*}	
	where the last equation follows from the tail bounds on the Normal distribution in Proposition \ref{prop:gauss} and the $\chi^2$-distribution in Proposition \ref{prop:chi-square}. Finally, by setting $t = \sqrt{8\log(2/\alpha)}$ ,
	and using the fact that $n \geq (64/9)\log(2/\alpha)$ gives the result.
\end{proof}

\paragraph{Acknowledgments.} We are grateful to the Harvard Privacy Tools differential privacy research group for illuminating and motivating discussions, particularly James Honaker, Gary King, Kobbi Nissim and Uri Stemmer.  We would also like to thank Philip Leclerc, CDAR, Mathematical Statistician for carefully reading our paper and giving helpful feedback. 

\bibliography{ref}
\end{document}